\documentclass[a4paper]{article}
\usepackage{fullpage}
\usepackage{enumitem}
\usepackage{calc}

\usepackage{lmodern} %Type1-font for non-english texts and characters
\usepackage{graphicx} %%For loading graphic files

\usepackage{amsmath, accents}
\usepackage{amssymb,tikz}
\usepackage{pict2e}
\usepackage{amsthm}
\usepackage{amscd}
\usepackage{mathrsfs}
\usepackage{todonotes}

\usepackage[colorlinks=true, pdfstartview=FitV, linkcolor=blue, citecolor=blue, urlcolor=blue]{hyperref}

\usepackage{yfonts}

\usepackage{marvosym}

\newtheorem{theorem}			     {Theorem} [section]
\newtheorem{proposition}[theorem]	 {Proposition}	
\newtheorem{corollary}	  [theorem]	 {Corollary}	
\newtheorem{lemma}	      [theorem]  {Lemma}		
\theoremstyle{definition}
\newtheorem{definition}	         {Definition}

\newtheorem{remark} {Remark}

\newcommand{\C}{\mathbb{C}}
\newcommand{\R}{\mathbb{R}}

\newcommand{\Or}{\mathcal{O}}

\newcommand{\Ai}{{\rm Ai}}

\usepackage{tikz}
\usetikzlibrary{arrows}
\usetikzlibrary{decorations.pathmorphing}
\usetikzlibrary{decorations.markings}
\usetikzlibrary{patterns}
\usetikzlibrary{automata}
\usetikzlibrary{positioning}
\usepackage{tikz-cd}

\usepackage{pgfplots}
\numberwithin{equation}{section}

\def\ds{\displaystyle}
\def\bigO{{\cal O}}

\tikzset{->-/.style={decoration={markings,mark=at position #1 with {\arrow[thick]{>}}},postaction={decorate}}}
	
\tikzset{-<-/.style={decoration={markings,mark=at position #1 with{\arrow[thick]{<}}},postaction={decorate}}}
				
\begin{document}
\title{Oscillatory asymptotics for the Airy kernel determinant \\ on two intervals}
\author{Elliot Blackstone, Christophe Charlier, Jonatan Lenells \footnote{Department of Mathematics, KTH Royal Institute of Technology, Lindstedtsv\"{a}gen 25, SE-114 28 Stockholm, Sweden. We are grateful to the referees for several excellent suggestions. We are also grateful to B. Fahs, I. Krasovsky, and T. Maroudas for useful discussions. The work of all three authors is supported by the European Research Council, Grant Agreement No. 682537. Lenells also acknowledges support from the G\"oran Gustafsson Foundation, the Swedish Research Council, Grant No. 2015-05430, and the Ruth and Nils-Erik Stenb\"ack Foundation.
E-mail: elliotb@kth.se, cchar@kth.se, jlenells@kth.se.}}
\maketitle

\begin{abstract}
We obtain asymptotics for the Airy kernel Fredholm determinant on two intervals. We give explicit formulas for all the terms up to and including the oscillations of order $1$, which are expressed in terms of Jacobi $\theta$-functions.
\end{abstract}
 
%\vspace{0.2cm}\hspace{-0.5cm}AMS subject classiciation: 41A60, 35Q15, 

%\tableofcontents

\section{Introduction}

%\todo[inline]{C: We still need to mention somewhere the references \cite{BotDeiItsKra1} and \cite{Widom1995} as suggested by the referee}

The Airy point process is a well-known universal point process from the theory of random matrices and is used to model the behavior of the largest eigenvalue for a wide class of large Hermitian random matrices \cite{BEY, Deift, DeiftGioev, Forrester2, Soshnikov2}. The Airy point process also appears in other related models. For example, it describes the transition between the solid and liquid regions in random tiling models \cite{Johansson} and the largest parts of Young diagrams with respect to the Plancherel measure \cite{BaikDeiftRains, BOO}. This is a determinantal point process on $\mathbb{R}$ whose correlation kernel is given by
\begin{equation}\label{Airy kernel}
K^\Ai(u, v) = \frac{ \Ai(u) \Ai'(v) - \Ai'(u) \Ai(v) }{ u - v },
\end{equation}
where $\Ai$ denotes the Airy function. Let $\vec{x} = (x_{1},x_{2},x_{3})$ be such that $x_3<x_{2}<x_1<0$ and let $F(\vec{x}) = F(x_{1},x_{2},x_{3})$ be the probability of finding no points on $(x_{3},x_{2})\cup(x_{1},+\infty)$. In other words, $F$ represents a gap probability on two disjoint intervals. It is well-known \cite{Soshnikov} that $F(\vec{x})$ can be written as a Fredholm determinant
\begin{equation}\label{generating}
F(\vec{x}) = \det\left( 1 - \mathcal{K}^{\Ai} \chi_{(x_{3},x_{2}) \cup (x_{1},+\infty)} \right),
\end{equation}
where $\mathcal{K}^{\Ai}$ is the integral operator whose kernel is $K^{\Ai}$. The purpose of this paper is to obtain asymptotics for $F(r\vec{x}) = F(rx_{1},rx_{2},rx_{3})$ as $r \to + \infty$. Such asymptotics are often referred to as large gap asymptotics.

\vspace{0.2cm}The history of gap expansions began with the work of Dyson \cite{Dyson} for the sine process. Many authors contributed over the years to the analysis of large gap asymptotics, including Widom \cite{Widom1995}, des Cloizeaux and Mehta \cite{CloizeauxMehta}, Krasovsky \cite{KrasovskySine}, Ehrhardt \cite{Ehr sine}, etc., and we refer to \cite{BotDeiItsKra1} and the references therein for a historical review. 

\vspace{0.2cm}Large gap asymptotics on a single interval for the Airy kernel have been analysed in \cite{BBD, DIK, TracyWidom} and are now well-understood. The probability of a gap on $(-r, +\infty)$ enjoys the following asymptotics
\begin{equation}\label{largegapAiry1}
 \det\left( 1 - \mathcal{K}^{\Ai} \chi_{(-r,+\infty)} \right) = \exp\Big(-\frac{|r|^{3}}{12} - \frac{1}{8} \log |r| + \zeta'(-1) + \frac{1}{24}\log 2 + \bigO(r^{-\frac{3}{2}})\Big), \qquad \mbox{as } r \to + \infty,
\end{equation}
where $\zeta'$ denotes the derivative of the Riemann zeta function. This expansion was first obtained in \cite{TracyWidom}, but without rigorously proving the term of order $1$. Note that the term of order $1$ in \eqref{largegapAiry1} is a constant (independent of $r$). In Theorem \ref{thm: main result} below, we show that this in no longer the case when $x_{3}<x_{2}<x_{1}<0$; instead the term of order $1$ presents some oscillations expressed in terms of a Jacobi $\theta$-function.

\vspace{0.2cm}We briefly introduce the necessary material to present our result. 
Let us define
\begin{align}\label{def of Q intro}
\sqrt{\mathcal{Q}(z)} = \sqrt{(z-x_{1})(z-x_{2})(z-x_{3})},
\end{align}
where the principal branch is taken for each square root. Thus $\sqrt{\mathcal{Q}(z)}$ is analytic on $\mathbb{C}\setminus \big( (-\infty,x_{3}] \cup [x_{2},x_{1}] \big)$ and $\sqrt{\mathcal{Q}(z)} \sim z^\frac{3}{2}$ as $z \to \infty$. Let us also define
\begin{align}\label{def of p intro}
q(z) = -z^{2} + \frac{x_{1}+x_{2}+x_{3}}{2}z -q_{0},
\end{align}
where $q_{0} \in \mathbb{R}$ is given by
\begin{align}\label{def of q0}
q_{0} = \left( \int_{x_{3}}^{x_{2}} \frac{ds}{\sqrt{\mathcal{Q}(s)}} \right)^{-1} \int_{x_{3}}^{x_{2}} \frac{-s^{2} + \frac{x_{1}+x_{2}+x_{3}}{2}s}{\sqrt{\mathcal{Q}(s)}}ds.
\end{align}
The constant $\Omega$, defined by
\begin{align*}
\Omega = \int_{x_{2}}^{x_{1}} \frac{2q(x)}{\big|\sqrt{\mathcal{Q}(x)}\big|}dx,
\end{align*}
also appears in our final result. We show in Proposition \ref{p_prop} (see also Remark \ref{remark: x/y plane}) that $q$ has one simple root in $(x_{3},x_{2})$ and another simple root in $(x_{1},+\infty)$, so $q(x)>0$ for $x \in (x_{2},x_{1})$, which implies that $\Omega > 0$. We define $c_0$ and $\tau$ by
\begin{align}\label{def of c0 and tau}
& c_{0} = \left( \int_{x_{3}}^{x_{2}} \frac{2dx}{\big|\sqrt{\mathcal{Q}(x)}\big|} \right)^{-1} \in \mathbb{R}^{+}, \qquad \tau = \int_{x_{2}}^{x_{1}} \frac{2ic_{0}dx}{\big| \sqrt{\mathcal{Q}(x)}\big|} \in i \mathbb{R}^{+}.
\end{align}
The $\theta$-function $\theta(z) = \theta(z;\tau)$ associated to the constant $\tau$ of \eqref{def of c0 and tau} is given by (see e.g. \cite{NIST})
\begin{align}
\theta(z) = \sum_{m=-\infty}^{\infty} e^{2\pi i m z}e^{\pi i m^{2} \tau}.
\end{align}
It is an entire function which satisfies
\begin{align}\label{periodicity property of theta function}
& \theta(z+1) = \theta(z), & &  \theta(z+\tau) = e^{-2\pi i z}e^{-\pi i \tau}\theta(z), & & \theta(-z) = \theta(z), & & \mbox{for all } z \in \mathbb{C}.
\end{align}
We are now ready to state our main result.
\begin{theorem}\label{thm: main result}
Let $x_3<x_{2}<x_1<0$ be fixed. As $r \to + \infty$, we have
\begin{align}
& F(r\vec{x}) = \exp \Big( c\, r^{3} - \frac{1}{2} \log r + \log \theta(\nu) + C + \bigO(r^{-1}) \Big), \label{F asymp} \\
& c = \frac{x_{1}^{3}+x_{2}^{3}+x_{3}^{3}-(x_{1}+x_{2})(x_{1}+x_{3})(x_{2}+x_{3})}{12}+\frac{q_{0}}{3}\big( x_{1}+x_{2}+x_{3}\big), \label{def of C1}
\end{align}
where $C=C(\vec{x})$ is independent of $r$ and $\nu = \nu(r,\vec{x})$ is given by
\begin{align}\label{def of a}
\nu = - \frac{\Omega}{2\pi}r^\frac{3}{2}.
\end{align}
%Furthermore, the error term in \eqref{F asymp} is uniform for $x_{1},x_{2},x_{3}$ in compact subsets of $(-\infty,0)$, as long as there exists a constant $\delta > 0$ independent of $r$ such that
%\begin{align*}
%\min \{x_{1}-x_{2},x_{2}-x_{3}\} \geq \delta.
%\end{align*}
\end{theorem}
\begin{remark}
Theorem \ref{thm: main result} has been verified numerically using the Linear Algebra package of Bornemann for Fredholm determinants \cite{Bornemann}.
\end{remark}
\begin{remark}\label{remark:consistency check}
The asymptotic formula \eqref{F asymp} is not valid if $r \to + \infty$ and simultaneously $x_{2}\uparrow x_{1}$ or $x_{2} \downarrow x_{3}$. However, it is possible to naively expand our expression for the constant $c$ as $x_{2}\uparrow x_{1}$ or as $x_{2} \downarrow x_{3}$ and see if it agrees with the constants for the leading term of \eqref{largegapAiry1}. The expression \eqref{def of C1} for $c$ is given in terms of $q_{0}$ defined in \eqref{def of q0}. We show in Lemma \ref{lemma:g1 asymptotics} that
\begin{align}
& q_{0} = \frac{x_{1}x_{3}}{2} + \bigO\big(x_{2}-x_{3}\big), & & \mbox{as } x_{2} \downarrow x_{3}, \label{asymp for q0 when x2 to x3} \\
& q_{0} = \frac{x_{1}x_{3}}{2}+\bigO\left(\frac{1}{\log(x_{1}-x_{2})}\right), & & \mbox{as } x_{2} \uparrow x_{1}. \label{asymp for q0 when x2 to x1}
\end{align}
Using these expansions, we obtain
\begin{align*}
\lim_{x_{2}\uparrow x_{1}} c = -\frac{|x_{3}|^{3}}{12}, \qquad \lim_{x_{2}\downarrow x_{3}} c = -\frac{|x_{1}|^{3}}{12},
\end{align*}
which is consistent with \eqref{largegapAiry1}. Interestingly, the coefficient $-\frac{1}{2}$ of the $\log r$ term in \eqref{F asymp} is independent of $x_{1},x_{2},x_{3}$ and different from the coefficient $-\frac{1}{8}$ of the $\log r$ term in \eqref{largegapAiry1}.
\end{remark}
\begin{remark}
In \cite{DIZ}, Deift, Its and Zhou have obtained large gap asymptotics for the sine point process on disjoint intervals. Their final asymptotic formula is similar to our formula \eqref{F asymp}, in the sense that the subleading term in \cite{DIZ} is also logarithmic, and the third term presents oscillations also described in terms of Jacobi $\theta$-functions. However, their constant associated to the logarithmic term is quite involved and expressed in terms of hyperelliptic integrals, while in our case the constant is simply $-\frac{1}{2}$. We were informed by B. Fahs at an early stage of our project that the prefactor of log s is also $-\frac{1}{2}$ for the sine kernel determinant on two large intervals \cite{FK20}.
\end{remark}

\subsection{Outline of the proof}
We first follow the procedure developed by Its-Izergin-Korepin-Slavnov \cite{IIKS} to obtain an identity for $\partial_{r} \log F(r\vec{x})$ in terms of a certain resolvent operator. Next, we use a result of Claeys and Doeraene \cite{ClaeysDoeraene} to express this resolvent in terms of a model Riemann-Hilbert (RH) problem, whose solution is denoted $\Psi$, see \eqref{final formula diff identity} for the exact expression.  We employ the Deift/Zhou \cite{DKMVZ1,DeiftZhou} steepest descent method to obtain asymptotics for $\Psi$ as $r \to + \infty$. This method is well-established, but its application in the present situation leads to rather involved analysis. By integrating the differential identity $\partial_{r} \log F(r\vec{x})$, we have
\begin{align}\label{integration diff identity general form}
\log F(r\vec{x}) = \log F(M\vec{x}) + \int_{M}^{r} \partial_{r'} \log F(r'\vec{x})dr',
\end{align}
where $M>0$ is a sufficiently large constant. We obtain our main result by substituting the large $r'$ asymptotics for $\partial_{r'} \log F(r'\vec{x})$ and then performing integration with respect to $r'$. 

\vspace{0.2cm}Our analysis of the RH problem for $\Psi$ involves several quantities related to a genus one Riemann surface, which is realized as two copies of the complex plane glued together along two disjoint intervals. In particular, Jacobi $\theta$-functions, the Abel map, see \eqref{def of Abel}, and the solution of a Jacobi inversion problem, all play an important role.  The simple fact that the large $r$ asymptotics for $F(r\vec{x})$ are of the form
\begin{align}
& F(r\vec{x}) = \exp \Big( c\, r^{3} +c_{2} \log r + c_{3}\log \theta(\nu) + C + \bigO(r^{-1}) \Big), \label{F asymp c1 c2 c3}
\end{align}
for certain constants $c$, $c_{2}$, $c_{3}$ and $C$ requires a non-trivial amount of work. In particular, it is not straightforward to show that there are no terms proportional to $r^\frac{3}{2}$, no oscillations of order $\log r$, and that all oscillations of order $1$ can be written in the form $c_{3}\log \theta(\nu)$. An important part of this work is devoted to understanding the algebra related to the $\theta$-functions and the Abel map (see \cite{DIZ} for pioneering work in this direction). We obtain very simple expressions for $c_{2}$ and $c_{3}$ only after considerable simplifications using many identities involving $\theta$-functions and the Abel map. Furthermore, the identity $c_{3} = 1$ requires the use of Riemann's bilinear identity, while $c_{2} = -\frac{1}{2}$ requires the explicit evaluation of certain elliptic integrals, see Proposition \ref{prop: c2 is -1/2}. The fact that we manage to prove $c_{2} = -\frac{1}{2}$ and $c_{3} = 1$ can be viewed as one of the main contributions of this paper. 

\vspace{0.2cm}The integration constant $\log F(M\vec{x})$ in \eqref{integration diff identity general form} is unknown and contributes to the expression $C$ of \eqref{F asymp}, so our method does not allow for an explicit expression for $C$. Such constants are notoriously complicated to obtain (see \cite{Krasovsky,FahsKrasovsky}), and we leave this for future works.

\begin{remark}
I. Krasovsky and T. Maroudas \cite{KraMarou} have simultaneously and independently analyzed the large $r$ asymptotics of the Fredholm determinant $F(r\vec{x})$. Their method is based on a different differential identity than the one considered here.
\end{remark}

\subsection{Organization of the paper}
First, we present the model RH problem for $\Psi$ and obtain the differential identity \eqref{final formula diff identity} in Section \ref{section:diffid}. We perform the first steps of the steepest descent analysis in Section \ref{section:RH1}. These steps are denoted by $\Psi \mapsto T \mapsto S$ and consists of constructing the so-called $g$-function and the opening of lenses. The second main part of the steepest descent method consists in constructing approximations to $S$ in four distinct regions of the complex plane. The global parametrix $P^{(\infty)}$ approximates $S$ everywhere in the complex plane except near the points $y_{j}:=x_{j}-x_{3}$, and the local parametrix $P^{(y_j)}$ approximates $S$ near the point $y_j$, $j=1,2,3$. In Section \ref{section: global parametrix}, we construct $P^{(\infty)}$ explicitly in terms of Jacobi $\theta$-functions. This construction is technically different from the construction in \cite{DKMVZ2}, but contains similar ideas \cite{DIKZ1999}. In Section \ref{subsec:Besselparametrix}, we use some results from \cite{DIZ} to build local parametrices $P^{(y_{j})}$, $j=1,2,3$, in terms of Bessel functions. The last step $S \mapsto R$ of the steepest descent method is done in Section \ref{section:smallnorm}; in particular we show that $P^{(y_j)}$ approximates $S$ in a small neighborhood of $y_{j}$, $j=1,2,3$, for sufficiently large $r$. In Section \ref{Section: theta identities}, we summarize certain $\theta$-function identities that are required to analyze the differential identity \eqref{final formula diff identity}.  In Section \ref{section:integration1}, we complete the proof of Theorem \ref{thm: main result} by obtaining large $r$ asymptotics for $\partial_{r} \log F(r \vec{x})$ via \eqref{final formula diff identity}, and subsequently large $r$ asymptotics for $\log F(r \vec{x})$ via \eqref{integration diff identity general form}.

\section{Differential identity for $F$}\label{section:diffid}
In this section, we use the method of Its-Izergin-Korepin-Slavnov \cite{IIKS} to obtain a differential identity for $\partial_{r} \log F(r\vec{x})$ in terms of a model RH problem $\Psi$ obtained in \cite{ClaeysDoeraene}, which is a simple generalization of the problem in \cite[p. 374]{CIK} and the corresponding formula \cite[eq (2.17)]{CIK}.

\vspace{0.2cm}A kernel $K : \mathbb{R}^{2} \to \mathbb{R}$ is said to be integrable if it can be written in the form $K(x,y)=\frac{f^T(x)g(y)}{x-y}$, where $f(x)$ and $g(y)$ are column vectors satisfying $f^T(x)g(x)=0$. It is well-known and easy to see from \eqref{Airy kernel} that the Airy kernel $K^{\Ai}$ is integrable. In our case, we need to consider the trace-class kernel
\begin{align}\label{def of Krx}
K_{\vec x}(x,y) =  K^\Ai(x,y)  \big(\chi_{(x_3, x_{2})}(y) + \chi_{(x_1, +\infty)}(y)\big).
\end{align}
This kernel is also integrable with $f$ and $g$ given by
\begin{align*}
f(x)=\begin{pmatrix}\Ai(x)\\
\Ai'(x)\end{pmatrix}, \qquad g(y)=\begin{pmatrix}
\Ai'(y)\big(\chi_{(x_3, x_{2})}(y) + \chi_{(x_1, +\infty)}(y)\big) \\
-\Ai(y)\big(\chi_{(x_3, x_{2})}(y) + \chi_{(x_1, +\infty)}(y)\big)
\end{pmatrix}.
\end{align*}
The associated integral operator $\mathcal K_{\vec x}:L^{2}((x_{3},x_{2})\cup(x_{1},+\infty)) \to L^{2}((x_{3},x_{2})\cup(x_{1},+\infty))$ is given by
\begin{equation}\label{eq:defoperator}
\mathcal{K}_{\vec x}f(x)=\int_{-\infty}^{+\infty}K^\Ai(x,y) \big(\chi_{(x_3, x_{2})}(y) + \chi_{(x_1, +\infty)}(y)\big) f(y)dy.
\end{equation}
Let us now scale the size of the intervals with a parameter $r>0$, i.e.\ we replace $\vec{x}$ by $r \vec{x}$. Since $\partial_{r} \log F(r\vec{x}) =\partial_r\log\det\left(I-\mathcal{K}_{r\vec{x}}\right) $, by using general theory of integrable kernel operators, we obtain
\begin{align}
    \partial_{r} \log F(r\vec{x})&=-\text{Tr}\left[\left(I-\mathcal{K}_{r\vec{x}}\right)^{-1}\partial_r\mathcal{K}_{r\vec{x}}\right]=\sum_{j=1}^{3} (-1)^{j+1}x_j\text{Tr}\left[\left(I-\mathcal{K}_{r\vec{x}}\right)^{-1}\mathcal{K}_{r\vec{x}} \delta_{rx_j}\right] \nonumber \\
    &=\sum_{j=1}^{3}(-1)^{j+1}x_j\text{Tr}\left[\mathcal{R}_{r\vec{x}} \delta_{rx_j}\right] =\sum_{j=1}^{3}(-1)^{j+1}x_j\lim_{u\to rx_j}R_{r\vec{x}}(u,u), \label{F_r identity}
\end{align}
where $\delta_{rx_j}$ is the Dirac delta operator, $\mathcal{R}_{r\vec{x}}$ is the resolvant operator given by
\begin{align*}
\mathcal{R}_{r\vec{x}} := \left(I-\mathcal{K}_{r\vec{x}}\right)^{-1}\mathcal{K}_{r\vec{x}} = \left(I-\mathcal{K}_{r\vec{x}}\right)^{-1} - I,
\end{align*}
$R_{r\vec{x}}$ is the kernel of $\mathcal{R}_{r\vec{x}}$, and the limits $u\to rx_j$, $j = 1,2,3$ in \eqref{F_r identity} are taken from the interior of $(rx_3,rx_2) \cup (rx_1,+\infty)$. 

Next, we use a result of \cite{ClaeysDoeraene} (based on \cite{IIKS}) which relates $\mathcal{R}_{r\vec{x}}$ to the model RH problem $\Psi$.
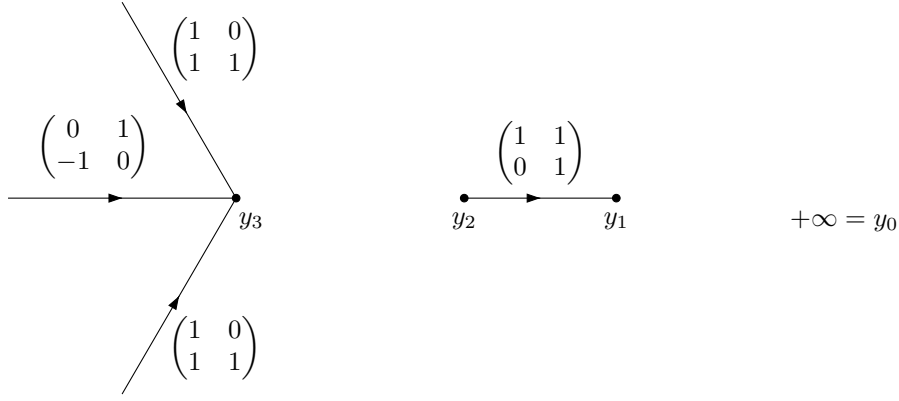
\begin{figure}
\centering
\begin{tikzpicture}
\draw[fill] (0,0) circle (0.05);
\draw (3,0) -- (5,0);
\draw (0,0) -- (120:3);
\draw (0,0) -- (-120:3);
\draw (0,0) -- (-3,0);
\draw[fill] (3,0) circle (0.05);
\draw[fill] (5,0) circle (0.05);
%\draw[fill] (8,0) circle (0.05);

\node at (0.2,-0.3) {$y_{3}$};
\node at (3,-0.3) {$y_{2}$};
\node at (5,-0.3) {$y_{1}$};
\node at (8,-0.3) {$+\infty=y_{0}$};

\node at (98:2) {$\begin{pmatrix} 1 & 0 \\ 1 & 1 \end{pmatrix}$};
\node at (160:2) {$\begin{pmatrix} 0 & 1 \\ -1 & 0 \end{pmatrix}$};
\node at (-98:2) {$\begin{pmatrix} 1 & 0 \\ 1 & 1 \end{pmatrix}$};

\node at (4,0.6) {$\begin{pmatrix} 1 & 1 \\ 0 & 1 \end{pmatrix}$};

\draw[black,arrows={-Triangle[length=0.18cm,width=0.12cm]}]
(-120:1.5) --  ++(60:0.001);
\draw[black,arrows={-Triangle[length=0.18cm,width=0.12cm]}]
(120:1.3) --  ++(-60:0.001);
\draw[black,arrows={-Triangle[length=0.18cm,width=0.12cm]}]
(180:1.5) --  ++(0:0.001);

\draw[black,arrows={-Triangle[length=0.18cm,width=0.12cm]}]
(0:4) --  ++(0:0.001);

\end{tikzpicture}
\caption{Jump contours for the model RH problem for $\Psi$.}
\label{fig:modelRHcontours}
\end{figure}
For convenience, we define
\begin{align*}
y_{j} = x_{j}-x_{3}, \quad j=1,2,3.
\end{align*}
Using the Its-Izergin-Korepin-Slavnov method, it was shown in \cite[proof of Proposition 1]{ClaeysDoeraene} that the resolvent kernel $R_{r\vec x}(u;u)$ can be expressed in terms of the solution of a naturally associated RH problem.
For $u \in (rx_3,rx_2) \cup (rx_1,+\infty)$, we have
\begin{equation}\label{link between resolvant and RHP}
R_{\vec x}(u,u)=\frac{1}{2\pi i}\left(\Psi_+^{-1}\Psi_+'\right)_{21}(u-x_3;x_{3},\vec{y}),
\end{equation}
where $\vec y=(y_1, y_{2},y_3)$ and $\Psi(z;x_{3},\vec{y})$ is the solution of the following RH problem. 

\subsubsection*{RH problem for $\Psi$}
\begin{enumerate}[label={(\alph*)}]
\item[(a)] $\Psi : \C \setminus \Gamma \rightarrow \C^{2\times 2}$ is analytic, with
\begin{equation}\label{eq:defGamma}
\Gamma=e^{\pm\frac{2\pi i}{3}} (0,+\infty)\cup (-\infty,0] \cup [y_{2},y_{1}]
\end{equation}
and $\Gamma$ oriented as in Figure \ref{fig:modelRHcontours}.
\item[(b)] $\Psi(z)$ has continuous boundary values as $z\in\Gamma\backslash \{y_{1},y_{2},y_{3}\}$ is approached from the left ($+$ side) or from the right ($-$ side) and they are related by
\begin{equation*}
\begin{array}{ll}
\Psi_+(z) = \Psi_-(z) \begin{pmatrix} 1 & 0 \\ 1 & 1 \end{pmatrix} & \textrm{for } z \in e^{\pm\frac{2\pi i}{3}} (0,+\infty), \\
\Psi_+(z) = \Psi_-(z) \begin{pmatrix} 0 & 1 \\ -1 & 0 \end{pmatrix} & \textrm{for } z \in (-\infty,0), \\
\Psi_+(z) = \Psi_-(z) \begin{pmatrix} 1 & 1 \\ 0 & 1 \end{pmatrix} & \textrm{for } z \in (y_{2}, y_{1}).
\end{array}
\end{equation*}
%where we write $y_m=0$ and $y_{0} = + \infty$.
\item[(c)] As $z \rightarrow \infty$, we have %there exist matrices $\Psi_1,\Psi_2$ depending on $x_{3},y_{1},y_{2}$ but not on $z$ such that $\Psi$ has the asymptotic behavior
\begin{equation}
\label{eq:psiasympinf}
\Psi(z) = \left( I + \bigO( z^{-1})\right) z^{\frac{1}{4} \sigma_3} M^{-1} e^{-(\frac{2}{3}z^\frac{3}{2} + x_{3}z^{\frac{1}{2}}) \sigma_3},
\end{equation}
where $M = (I + i \sigma_1) / \sqrt{2}$, $\sigma_1 = \begin{pmatrix} 0 & 1 \\ 1 & 0 \end{pmatrix}$ and $\sigma_3 = \begin{pmatrix} 1 & 0 \\ 0 & -1 \end{pmatrix}$, and where principal branches of $z^\frac{3}{2}$ and $z^{\frac{1}{2}}$ are taken.
\item[(d)] $\Psi(z) = \Or(\log|z-y_j|)$ as $z \rightarrow y_j$, $j = 1,2,3$.
\end{enumerate}
We conclude from the definition of $F$ \eqref{generating} together with \eqref{F_r identity} and \eqref{link between resolvant and RHP} that
\begin{align}\label{final formula diff identity}
\partial_{r} \log F(r\vec{x})&=\sum_{j=1}^3\frac{(-1)^{j+1}x_j}{2\pi i}\lim_{z \to y_{j}}\left(\Psi_+^{-1}\Psi_+'\right)_{21}(rz;rx_{3},r\vec y),
\end{align}
where the limits as $z\to y_j$, $j=1,2,3,$ are taken such that $z \in (0,y_2) \cup (y_1,+\infty)$. The differential identity \eqref{final formula diff identity} expresses $\partial_{r} \log F(r\vec{x})$ in terms of $\Psi$, which is the solution to a RH problem. In Sections \ref{section:RH1}-\ref{section:smallnorm}, we employ the Deift/Zhou steepest descent method \cite{DeiftZhou} to obtain the asymptotics of $\Psi(z), \Psi'(z)$ as $r\to+\infty$ for $z$ in small neighborhoods of $y_1,y_2,y_3$.

\section{Steepest descent for $\Psi$: first steps}\label{section:RH1}
In this section, we perform the first steps of the steepest descent method. We obtain the so-called $g$-function in Subsection \ref{subsection: g-function}. This function is used to normalize the RH problem at $\infty$ via the $\Psi \to T$ transformation of Subsection \ref{Subsection: Phi to T}. Finally, we proceed with the opening of the lenses $T \mapsto S$ in Subsection \ref{Subsection: T to S}. 
\subsection{$g$-function}\label{subsection: g-function}
The first step of the analysis consists of normalizing the RH problem at $\infty$ by means of an appropriate $g$-function; see \cite{DIZ} for a similar construction. Let us define
\begin{align*}
\sqrt{\mathcal{R}(z)} = \sqrt{(z-y_{1})(z-y_{2})z}
\end{align*}
where the principal branch is taken for each square root. More precisely, $\sqrt{\mathcal{R}(z)}$ is analytic on $\mathbb{C}\setminus \big( (-\infty,0] \cup [y_{2},y_{1}] \big)$ and $\sqrt{\mathcal{R}(z)} \sim z^\frac{3}{2}$ as $z \to \infty$. In particular it satisfies the jumps
\begin{align}\label{jump for R 2cuts}
& \sqrt{\mathcal{R}(z)}_{+} + \sqrt{\mathcal{R}(z)}_{-} = 0 & & \mbox{for } z \in (-\infty,0)\cup(y_{2},y_{1}).
\end{align}
We define the derivative of the $g$-function by
\begin{align}\label{gprime}
g'(z) = \frac{p(z)}{\sqrt{\mathcal{R}(z)}},
\end{align}
where
\begin{align}\label{p}
    p(z):=-z^{2}+\frac{y_1+y_2-x_3}{2}z+\frac{x_3(y_1+y_2)}{4}+\frac{(y_1-y_2)^2}{8}-g_{1}
\end{align}
and $g_{1} \in \mathbb{R}$ is given by
\begin{equation}\label{def of g1 2cuts}
g_{1} = \left(\int_{0}^{y_{2}}\frac{ds}{\sqrt{\mathcal{R}(s)}}\right)^{-1}\left(\int_{0}^{y_{2}}\frac{-s^{2}+\frac{y_1+y_2-x_3}{2}s+\frac{x_3(y_1+y_2)}{4}+\frac{(y_1-y_2)^2}{8}}{\sqrt{\mathcal{R}(s)}}ds\right).
\end{equation}
By definition of $g_{1}$, we have
\begin{align}\label{g' integrate to 0 2cuts}
\int_{0}^{y_{2}} g'(s)ds = 0.
\end{align}
\begin{remark}\label{remark: x/y plane}
The functions $\mathcal{Q}$ and $q$ defined in \eqref{def of Q intro} and \eqref{def of p intro} have the relations with $\mathcal{R}$ and $p$
\begin{align*}
\mathcal{Q}(z) = \mathcal{R}(z-x_{3}) \quad \mbox{ and } \quad q(z) = p(z-x_{3}).
\end{align*}
\end{remark}
\begin{proposition}\label{p_prop}
The degree two polynomial $p$ has two simple real zeros $z_\pm$, with $z_-\in(0,y_2)$ and $z_+\in(y_1,\infty)$. Moreover, we have
\begin{align*}
p(y_{3})<0, \qquad p(y_{2}) > 0, \qquad p(y_{1}) > 0.
\end{align*}
\end{proposition}
\begin{proof}
The polynomial $p$ has real coefficients, thus $p(z) \in \mathbb{R}$ for $z \in \mathbb{R}$. Then we immediately conclude from \eqref{g' integrate to 0 2cuts} that $p(z)$ must have two real simple zeros and that at least one of them lies on $(0,y_2)$, because $\sqrt{\mathcal{R}(z)}<0$ for all $z\in(0,y_2)$. However it is not straightforward to prove that $p$ has one zero in $(y_{1},+\infty)$, or equivalently that $p(y_{1}) > 0$ (since $p(z) \to -\infty$ as $z \to + \infty$). Note that $p(z)$ can be written as
\begin{align*}
p(z)=-z^2+\frac{z}{2}(y_1+y_2-x_3)+\alpha, ~ \mbox{ where } ~ \alpha:=\left(\int_{0}^{y_{2}}\frac{ds}{\sqrt{\mathcal{R}(s)}}\right)^{-1}\left(\int_{0}^{y_{2}}\frac{s^2-\frac{s}{2}(y_1+y_2-x_{3})}{\sqrt{\mathcal{R}(s)}}ds\right).
\end{align*}
The constant $\alpha$ can be written explicitly in terms of complete elliptic integrals of the first ($\mathbf{K}$) and second ($\mathbf{E}$) kind, which are defined as \cite[eq 8.112]{GRtable}
\begin{align*}
\mathbf{K}(k) = \int_{0}^{1} \frac{dx}{\sqrt{(1-x^{2})(1-k^{2}x^{2})}}, \qquad \mathbf{E}(k) = \int_{0}^{1} \frac{\sqrt{1-k^{2}x^{2}}}{\sqrt{1-x^{2}}}dx.
\end{align*}
Let us define $k_{\star} := \sqrt{\frac{y_{2}}{y_{1}}}$. Using \cite[eq 3.131.4, 3.141.15 and 3.141.27]{GRtable}, we obtain
\begin{align*}
\int_{0}^{y_{2}}\frac{ds}{\sqrt{\mathcal{R}(s)}}&=\frac{-2}{\sqrt{y_1}}\mathbf{K} (k_{\star} ), ~~~ \int_{0}^{y_{2}}\frac{s~ds}{\sqrt{\mathcal{R}(s)}}=-2\sqrt{y_1}\big( \mathbf{K} (k_{\star} )-\mathbf{E} (k_{\star} )\big), \\
\int_{0}^{y_{2}}\frac{s^2~ds}{\sqrt{\mathcal{R}(s)}}&=\frac{2}{3}\sqrt{y_1}\big( 2(y_1+y_2)\mathbf{E} (k_{\star} )-(2y_1+y_2)\mathbf{K} (k_{\star} )\big),
\end{align*}
from which we deduce 
\begin{align}\label{alpha as complete integral}
\alpha=\frac{y_1}{2}\left[\frac{y_1-y_2}{3}+x_3-\frac{\mathbf{E}(k_{\star})}{\mathbf{K} (k_{\star})}\left(\frac{y_1+y_2}{3}+x_3\right)\right].
\end{align}
From \cite[eq 19.9.8]{NIST}, we have the lower bound
\begin{align*}
\frac{\mathbf{E} (k_{\star})}{\mathbf{K} (k_{\star} )}>\sqrt{1-k_{\star}^{2}} = \sqrt{1-\frac{y_{2}}{y_{1}}} =\sqrt{\frac{x_1-x_2}{x_1-x_3}}.
\end{align*}
We weaken this bound for our purpose.  Notice
\begin{align*}
4(x_1-x_2)(x_1-x_3)&=(x_1+x_2+x_3)^2+3x_1(x_1-2x_2-2x_3)-(x_2-x_3)^2<(x_1+x_2+x_3)^2,
\end{align*}
because $x_1-2x_2-2x_3>0$ and $x_{1}<0$, so we have 
\begin{align*}
\frac{x_1-x_2}{x_1-x_3}=\frac{4(x_1-x_2)^2}{4(x_1-x_2)(x_1-x_3)}>\frac{4(x_1-x_2)^2}{(x_1+x_2+x_3)^2}.
\end{align*}
The new bound is
\begin{align*}
\frac{\mathbf{E} (k_{\star} )}{\mathbf{K} (k_{\star} )}>\frac{2(x_2-x_1)}{x_1+x_2+x_3},
\end{align*}
from which we have
\begin{align*}
p(y_1)&=\frac{y_1}{2}(y_2-y_1-x_3)+\alpha=\frac{y_1}{6}\left[2(x_2-x_1)-\frac{\mathbf{E} (k_{\star} )}{\mathbf{K} (k_{\star} )}\left(x_1+x_2+x_3\right)\right]>0,
% &=\frac{y_1}{2}\left[\frac{2}{3}(y_2-y_1)-\frac{\mathbf{E}\left(\sqrt{\frac{y_2}{y_1}}\right)}{\mathbf{K}\left(\sqrt{\frac{y_2}{y_1}}\right)}\left(\frac{y_1+y_2}{3}+x_3\right)\right] \\
\end{align*}
as desired.
\end{proof}
We define the $g$-function by
\begin{align}\label{g function def 2 cut}
g(z) = \int_{y_{1}}^{z} g'(s)ds,
\end{align}
where the path of integration does not cross $(-\infty,y_{1}]$. 
\begin{lemma}\label{g_lemma}
\begin{enumerate}
Let $x_{3}<x_{2}<x_{1}<0$.
\item The $g$-function is analytic in $\mathbb{C}\setminus(-\infty,y_1]$ and satisfies $g(z) = \overline{g(\overline{z})}$.
\item The $g$-function satisfies the jump conditions 
\begin{align}
& g_{+}(z) + g_{-}(z) = 0, & & z \in (-\infty,0)\cup(y_{2},y_{1}), \label{jump1 g 2cuts} \\
& g_{+}(z) - g_{-}(z) = i \Omega, & & z \in (0,y_{2}), \label{jump3 g 2cuts}
\end{align}
where $\Omega = 2i\int_{y_{2}}^{y_{1}} g_{+}'(s)ds = -2 i g_{+}(y_{2}) \in \mathbb{R}_+.$
\item As $z\to\infty$, $z \notin (-\infty,0)$, we have
\begin{equation}\label{asymp g 2cuts}
g(z) = -\frac{2}{3}z^\frac{3}{2} - x_{3} z^{\frac{1}{2}} +2g_{1}z^{-\frac{1}{2}} + \bigO(z^{-\frac{3}{2}}).
\end{equation}
\item Let $N_{0}$ be an open neighborhood of $(-\infty,y_{1}]$, such that 
\begin{align}\label{full sector at -inf}
\Big\{z \in \mathbb{C} : \arg z \in \Big(-\pi,-\pi+\frac{\pi}{3}+\epsilon\Big)\cup \Big(\pi-\frac{\pi}{3}-\epsilon,\pi\Big) \mbox{ and } |z| \geq M \Big\} \subset N_{0}
\end{align}
for some $\epsilon > 0$ and $M>0$.
We have
\begin{equation}\label{real g}
\Re\left[g(z)\right]\geq0, \qquad \mbox{for all }z\in N_0
\end{equation}
where equality holds only when $z\in(-\infty,0]\cup[y_2,y_1]$.
\end{enumerate}
\end{lemma}

\begin{proof}

\begin{enumerate}
\item Analyticity follows from the definition \eqref{g function def 2 cut}. It is straightforward to verify from \eqref{gprime} that $g'(z) = \overline{g'(\overline{z})}$, from which we deduce
\begin{align*}
\overline{g(\overline{z})} = \overline{\int_{y_{1}}^{\overline{z}} g'(s)ds} = \overline{\int_{0}^{1} g'(y_{1}+t(\overline{z}-y_{1}))(\overline{z}-y_{1})dt} = \int_{0}^{1} g'(y_{1}+t(z-y_{1}))(z-y_{1})dt = g(z).
\end{align*}
\item The jumps \eqref{jump1 g 2cuts} follow from \eqref{jump for R 2cuts} and \eqref{g' integrate to 0 2cuts}. Let $\gamma$ be a closed curve with positive orientation surrounding $[y_{2},y_{1}]$ and not intersecting $(-\infty,0]$. For $z \in (0,y_{2})$, we use \eqref{jump1 g 2cuts} to write
\begin{align*}
g_{+}(z) - g_{-}(z) = \int_{\gamma} g'(s)ds = -2 \int_{y_{2}}^{y_{1}}g_{+}'(s)ds,
\end{align*}
and this is \eqref{jump3 g 2cuts}. It follows from Proposition \ref{p_prop} that
\begin{align*}
\Omega = 2i \int_{y_{2}}^{y_{1}} g_{+}'(s)ds = 2 \int_{y_{2}}^{y_{1}} \frac{p(s)ds}{\sqrt{s} \sqrt{s-y_{2}}\sqrt{y_{1}-s}} > 0.
\end{align*}
\item From \eqref{gprime}, a computation shows that
\begin{align*}
g'(z) = -\sqrt{z} - \frac{x_{3}}{2}z^{-\frac{1}{2}} - g_{1}z^{-\frac{3}{2}} + \bigO(z^{-\frac{5}{2}}), & & \mbox{as } z \to \infty,
\end{align*}
and then after integration we have
\begin{align}\label{lol1}
g(z) = -\frac{2}{3}z^\frac{3}{2} - x_{3} z^{\frac{1}{2}} + g_{0} +2g_{1}z^{-\frac{1}{2}} + \bigO(z^{-\frac{3}{2}}), & & \mbox{as } z \to \infty,
\end{align}
for a certain $g_{0} \in \mathbb{C}$. We deduce from \eqref{lol1} that
\begin{align*}
g_{+}(z)+g_{-}(z) = 2g_{0} + \bigO(z^{-\frac{1}{2}}), & & \mbox{as } z \to -\infty, \; z \in \mathbb{R}^{-}.
\end{align*}
Comparing the above with \eqref{jump1 g 2cuts}, we conclude that $g_{0} = 0$.
\item It follows from part 2 that $\Re[g(z)]$ is single valued on $\mathbb{C}$ and furthermore $\Re[g(z)]\equiv 0$ for $z\in(-\infty,0)\cup(y_2,y_1)$. Due to Proposition \ref{p_prop}, we have
\begin{align*}
\Im[g'_+(x)]=\Im\bigg[\frac{p(x)}{\sqrt{\mathcal{R}(x)}_{\smash{+}}}\bigg]<0, \qquad \mbox{for } x\in(-\infty,0)\cup (y_{2},y_{1}).
\end{align*}
So $\Im[g_+(x)]$ decreases as $x\in(-\infty,0)\cup(y_{2},y_{1})$ increases, and it follows from the Cauchy-Riemann equations that $\Re[g(x+i\epsilon)]$ increases (and is therefore positive) for sufficiently small $\epsilon > 0$. We deduce from part 1 that $\Re[g(x-i\epsilon)]>0$ for sufficiently small $\epsilon>0$. Another immediate consequence of Proposition \ref{p_prop} is that $\Re[g(z)]>0$ for all $z \in (0,y_{2})$. Furthermore, from \eqref{g function def 2 cut} we have
\begin{align*}
\Re[g(z)] \sim c(z-y_{1})^{\frac{1}{2}},  \qquad c>0,
\end{align*}
(the exact value of $c$ is unimportant here) as $z\to y_1$ so there is an open disk $D_{y_{1}}$ around $y_{1}$ such that $\Re[g(z)]\geq 0$ for all $z \in D_{y_{1}}$ with equality $\Re[g(z)]=0$ if and only if $z \in D_{y_{1}} \cap (y_{2},y_{1})$. We conclude similarly that there also exists open disks $D_{y_{2}}$ and $D_{y_{3}}$ around $y_{2}$ and $y_{3}$, respectively, such that
\begin{align*}
\Re[g(z)]\geq 0, \qquad \mbox{for all } z \in D_{y_{2}}\cup D_{y_{3}}
\end{align*}
with equality only when $z\in(-\infty,0]\cup[y_2,y_1]$. The fact that we can choose $\epsilon > 0$ and $M>0$ such that \eqref{full sector at -inf} holds follows directly from \eqref{asymp g 2cuts}.
\end{enumerate}
\end{proof}
In Lemma \ref{lemma:g1 asymptotics} below, we obtain asymptotics for $g_{1}$ as $y_2\downarrow0$ and as $y_2\uparrow y_1$. This is useful to compare our results with \cite{DIK}, see Remark \ref{remark:consistency check}. Indeed, we deduce \eqref{asymp for q0 when x2 to x3}-\eqref{asymp for q0 when x2 to x1} from \eqref{g1 as y2 tends to 0}-\eqref{g1 as y2 tends to y1} by noting that
\begin{align}
q_{0} = g_{1}+\frac{2x_{1}x_{2}+2x_{1}x_{3}+2x_{2}x_{3}-x_{1}^{2}-x_{2}^{2}}{8}.
\end{align}
\begin{lemma}\label{lemma:g1 asymptotics}
Fix $y_{1}$ and $x_{3}$. We have
\begin{align}
& g_{1} =\frac{y_1}{4}\left(x_3+\frac{y_1}{2}\right)+\bigO(y_2^2), & & \mbox{as } y_2 \downarrow 0 \label{g1 as y2 tends to 0}\\
& g_1=-\frac{y_1(2y_1+3x_3)}{3\log(y_1-y_2)}+\bigO\left(\frac{1}{\log^{2}(y_1-y_2)}\right), & & \mbox{as } y_{2}\uparrow y_{1}. \label{g1 as y2 tends to y1}
\end{align}
\end{lemma}
\begin{proof}
As $y_2\downarrow0$, we have
\begin{align*}
\int_0^{y_2}\frac{ds}{\sqrt{\mathcal{R}(s)}}&=\int_0^{y_2}\frac{-ds}{\sqrt{s(y_2-s)}}\frac{1}{\sqrt{y_1}}\left(1+ \frac{s}{2y_{1}}+\frac{3s^{2}}{8y_{1}^{2}} + \bigO\Big( \frac{s^{3}}{y_{1}^{3}} \Big) \right) \\
&=\frac{-\pi}{\sqrt{y_1}}\left(1+\frac{y_2}{4y_1}+\frac{9y_2^2}{64y_1^2}+\bigO(y_2^3)\right).
\end{align*}
Similar computations show that
\begin{align*}
\int_0^{y_2}\frac{s~ds}{\sqrt{\mathcal{R}(s)}}&=\frac{-\pi y_2}{2\sqrt{y_1}}\left(1+\frac{3y_2}{8y_1}+\bigO(y_2^2)\right), ~~~ \int_0^{y_2}\frac{s^2~ds}{\sqrt{\mathcal{R}(s)}}=\frac{-3\pi y_2^2}{8\sqrt{y_1}}+\bigO(y_2^3), \quad \mbox{as } y_{2} \downarrow 0,
\end{align*}
from which we obtain \eqref{g1 as y2 tends to 0}. Now we turn to the proof of \eqref{g1 as y2 tends to y1}. Using \eqref{alpha as complete integral}, we express the constant $g_1$ explicitly in terms of complete elliptic integrals:
\begin{align*}
g_1&=\frac{y_1}{2}\frac{\mathbf{E} (k_{\star})}{\mathbf{K} (k_{\star})}\left(\frac{y_1+y_2}{3}+x_3\right)-\frac{x_3}{4}(y_1-y_2)-\frac{1}{24}(y_1-y_2)(y_{1}+3y_{2}).
\end{align*}
We then use \cite[eq 19.12.1 and 19.12.2]{NIST} for the asymptotics as $y_2\uparrow y_1$ to obtain \eqref{g1 as y2 tends to y1}.
\end{proof}

\subsection{Rescaling of the RH problem}\label{Subsection: Phi to T}
Define the function $T(z)=T(z;x_3,\vec{y})$ as follows:
\begin{equation}\label{def of T}
T(z) = \begin{pmatrix}
r^{-\frac{1}{4}} & 2 i r^{\frac{7}{4}}g_{1} \\
0 & r^{\frac{1}{4}}
\end{pmatrix} \Psi(rz;r x_3,r\vec{y})e^{-r^\frac{3}{2}g(z)\sigma_{3}}, \qquad z \in \mathbb{C}\setminus \Gamma.
\end{equation}
The asymptotics \eqref{eq:psiasympinf} of $\Psi$ then imply after a small calculation that $T$ behaves as
\begin{equation}
\label{eq:Tasympinf}
T(z) = \left( I + \Or\left(z^{-1}\right) \right) z^{\frac{1}{4} \sigma_3} M^{-1},
\end{equation}
as $z \to\infty$, $z \notin (-\infty,0)$, where the principal branches of the roots are chosen.
%The entries of $T_1$ and $T_2$ are related to those of $\Psi_1$ and $\Psi_2$ in \eqref{eq:psiasympinf}. 
Using \eqref{jump1 g 2cuts}, the jumps $T_{-}(z)^{-1}T_{+}(z)$ for $z \in (y_{2},y_{1})$ can be factorized as
\begin{align}\label{factorization jumps}
\begin{pmatrix}
e^{-2 r^\frac{3}{2}g_{+}(z)} & 1 \\
0 & e^{-2 r^\frac{3}{2}g_{-}(z)}
\end{pmatrix} = \begin{pmatrix}
1 & 0 \\
e^{-2 r^\frac{3}{2}g_{-}(z)} & 1
\end{pmatrix} \begin{pmatrix}
0 & 1 \\
-1 & 0
\end{pmatrix} \begin{pmatrix}
1 & 0 \\
e^{-2r^\frac{3}{2}g_{+}(z)} & 1
\end{pmatrix}.
\end{align}
\subsection{Opening of lenses}\label{Subsection: T to S}
Around the interval $(y_{2},y_{1})$, we will split the jump contour in three parts using \eqref{factorization jumps}. This transformation is traditionally called opening of lenses. Let us consider lens-shaped contours $\gamma_{2,+}$ and $\gamma_{2,-}$ lying in the upper and lower half plane respectively, as shown in Figure \ref{fig:contour for S}. 
\begin{figure}
\centering
\begin{tikzpicture}
\draw[fill] (0,0) circle (0.05);
\draw (0,0) -- (5,0);
\draw (0,0) -- (120:3);
\draw (0,0) -- (-120:3);
\draw (0,0) -- (-3,0);

\draw (3,0) .. controls (3.5,1) and (4.5,1) .. (5,0);
\draw (3,0) .. controls (3.5,-1) and (4.5,-1) .. (5,0);
%\draw (0,0) .. controls (1.5,1) .. (3,0);

\draw[fill] (3,0) circle (0.05);
\draw[fill] (5,0) circle (0.05);
%\draw[fill] (8,0) circle (0.05);

\node at (0.15,-0.3) {$0$};
\node at (3,-0.3) {$y_{2}$};
\node at (5.2,-0.3) {$y_{1}$};
\node at (-3,-0.3) {$-\infty$};

\draw[black,arrows={-Triangle[length=0.18cm,width=0.12cm]}]
(-120:1.5) --  ++(60:0.001);
\draw[black,arrows={-Triangle[length=0.18cm,width=0.12cm]}]
(120:1.3) --  ++(-60:0.001);
\draw[black,arrows={-Triangle[length=0.18cm,width=0.12cm]}]
(180:1.5) --  ++(0:0.001);

\draw[black,arrows={-Triangle[length=0.18cm,width=0.12cm]}]
(0:4.08) --  ++(0:0.001);
\draw[black,arrows={-Triangle[length=0.18cm,width=0.12cm]}]
(0:1.6) --  ++(0:0.001);

\draw[black,arrows={-Triangle[length=0.18cm,width=0.12cm]}]
(4.08,0.75) --  ++(0:0.001);
\draw[black,arrows={-Triangle[length=0.18cm,width=0.12cm]}]
(4.08,-0.75) --  ++(0:0.001);

\node at (4,1) {$\gamma_{2,+}$};
\node at (4,-1) {$\gamma_{2,-}$};
\end{tikzpicture}
\caption{Jump contours $\Sigma_{S}$ for  $S$.}
\label{fig:contour for S}
\end{figure}
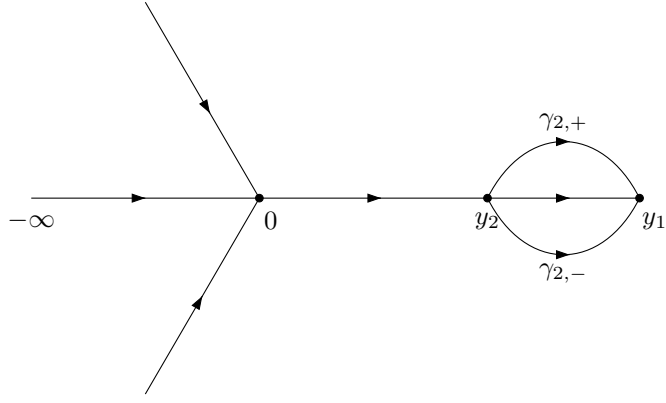
We define $S$ by
\begin{equation}\label{def:S}
S(z)=T(z)\left\{ \hspace{-0.1cm} \begin{array}{l l}
\begin{pmatrix}
1 & 0 \\
-e^{-2r^\frac{3}{2}g(z)} & 1
\end{pmatrix}, & z \mbox{ is inside the lenses around } (y_{2},y_{1}) \mbox{ and } \Im[z] > 0, \\
\begin{pmatrix}
1 & 0 \\
e^{-2r^\frac{3}{2}g(z)} & 1
\end{pmatrix}, & z \mbox{ is inside the lenses  around } (y_{2},y_{1}) \mbox{ and } \Im[z]< 0, \\
I, & \mbox{otherwise}.
\end{array} \right.
\end{equation}
We use Lemma \ref{g_lemma}, as well as the RH problem for $\Psi$ and the definitions \eqref{def of T} of $T$ and \eqref{def:S} of $S$ to conclude that $S$ satisfies the following RH problem. 
\subsubsection*{RH problem for $S$}
\begin{enumerate}[label={(\alph*)}]
\item[(a)] $S : \C \backslash \Sigma_{S} \rightarrow \C^{2\times 2}$ is analytic, with
\begin{equation}\label{eq:defGammaS}
\Sigma_{S}=(-\infty,y_{1}]\cup \gamma_{+}\cup \gamma_{-}, \qquad \gamma_{\pm} = \gamma_{2,\pm} \cup \big(e^{\pm \frac{2\pi i}{3}} (0,+\infty)\big)
\end{equation}
and $\Sigma_{S}$ oriented as in Figure \ref{fig:contour for S}.
\item[(b)] The jumps for $S$ are given by
\begin{align}
& S_{+}(z) = S_{-}(z)\begin{pmatrix}
0 & 1 \\ -1 & 0
\end{pmatrix}, & & z \in (-\infty,0)\cup (y_{2},y_{1}), \label{SInf0 jump} \\
& S_{+}(z) = S_{-}(z)\begin{pmatrix}
1 & 0 \\
e^{-2r^\frac{3}{2}g(z)} & 1
\end{pmatrix}, & & z \in \gamma_{+}\cup \gamma_{-}, \label{Slense jump} \\
& S_{+}(z) = S_{-}(z)e^{-i \Omega r^\frac{3}{2}\sigma_{3}}, & & z \in (0,y_{2}). \label{S0y2 jump}
\end{align}
\item[(c)] As $z \rightarrow \infty$, $z \notin (-\infty,0)$, we have
\begin{equation}
\label{eq:Sasympinf}
S(z) = \left( I + \Or\left(z^{-1}\right) \right) z^{\frac{1}{4} \sigma_3} M^{-1}.
\end{equation}
\item[(d)] As $z \to y_{j}$, $j = 1, 2,3$, we have
\begin{align}\label{S asymp at yj}
S(z) = \Or( \log|z-y_j|).
\end{align}
\end{enumerate}
Deforming $\gamma_{+}$ and $\gamma_{-}$ if necessary, we assume that $\gamma_{+},\gamma_{-} \subset N_{0}$ where $N_{0}$ has the properties described in Lemma \ref{g_lemma}. In particular, $\Re[g(z)]> 0$ for all $z \in \gamma_{+}\cup \gamma_{-}$, so that the jumps for $S$ are exponentially close to the identity as $r \to + \infty$ on the lens boundaries. This convergence is uniform outside small fixed neighborhoods of $y_{1},y_{2},y_{3}$, but is not uniform as $r \to + \infty$ and simultaneously $z \to y_{j}$, $j=1,2,3$.

\section{Global parametrix}\label{section: global parametrix}
In this section we construct the global parametrix $P^{(\infty)}$. We will show in Section \ref{section:smallnorm} that $P^{(\infty)}$ approximates $S$ outside of neighborhoods of $y_{1},y_{2},y_{3}$.

\subsubsection*{RH problem for $P^{(\infty)}$}
\begin{enumerate}[label={(\alph*)}]
\item[(a)] $P^{(\infty)} : \C \backslash (-\infty,y_{1}] \rightarrow \C^{2\times 2}$ is analytic.
\item[(b)] The jumps for $P^{(\infty)}$ are given by
\begin{align}
& P^{(\infty)}_{+}(z) = P^{(\infty)}_{-}(z)\begin{pmatrix}
0 & 1 \\ -1 & 0
\end{pmatrix}, & & z \in (-\infty,0)\cup (y_{2},y_{1}), \nonumber \\
& P^{(\infty)}_{+}(z) = P^{(\infty)}_{-}(z)e^{-i\Omega r^\frac{3}{2} \sigma_{3}}, & & z \in (0,y_{2}). \label{jumps for Pinf on (0,y2)}
\end{align}
\item[(c)] As $z \rightarrow \infty$, we have
\begin{equation}
\label{eq:Pinf asympinf}
P^{(\infty)}(z) = \left( I + \Or\left(z^{-1}\right) \right) z^{\frac{1}{4} \sigma_3} M^{-1}.
\end{equation}
\item[(d)] As $z \to y_j$, we have $P^{(\infty)}(z) = \bigO((z-y_{j})^{-\frac{1}{4}})$ for $j=1,2,3$.
\end{enumerate}
Conditions (a)-(c) for the RH problem for $P^{(\infty)}$ are obtained from the RH problem $S$ by ignoring the (pointwise) exponentially small jumps of $S$. Condition (d) does not come from the RH problem for $S$ and has been added to ensure uniqueness of the solution. The solution $P^{(\infty)}$ can be constructed in terms of Jacobi $\theta$-functions \cite{DIKZ1999}. This construction is technically different from the construction in \cite{DKMVZ2}, but contains similar ideas. Let $X$ be the two-sheeted Riemann surface of genus one associated to $\sqrt{\mathcal{R}(z)}$, with
\begin{align*}
\sqrt{\mathcal{R}(z)} = \sqrt{z(z-y_{2})(z-y_{1})},
\end{align*}
and we let $\sqrt{\mathcal{R}(z)} > 0$ for $z \in (y_{1},+\infty)$ on the first sheet. We also define $A$ and $B$ cycles such that they form a canonical homology basis of $X$. The $A$ cycle surrounds $(0,y_{2})$ with counterclockwise orientation. The upper part of the $A$ cycle (the solid line in Figure \ref{fig: cycles Riemann surface}) lies on the first sheet, and the lower part (the dashed line in Figure \ref{fig: cycles Riemann surface}) lies on the second sheet. The $B$ cycle surrounds $(-\infty,0)$ with clockwise orientation. 
\begin{figure}[]
\begin{center}
\begin{tikzpicture}
\draw   (0,0)--(1.5,0);
\node at (-0.5,0) {$-\infty$};
\node [above] at (7.5,-0.05) {$0$};
\node [above] at (9.5,-0.05) {$y_{2}$};
\node [above] at (11,-0.05) {$y_{1}$};

\draw[black,fill=black]  (7.5,0) circle [radius=0.04];
\draw[black,fill=black]  (9.5,0) circle [radius=0.04];
\draw[black,fill=black]  (11,0) circle [radius=0.04];

\node [below] at (9,1.5) {$A$};
\node [below] at  (1.6,1.55){$B$};

\draw  (1.5,0)--(7.5,0);
\draw  (9.5,0)--(11,0);

\draw[-<-=0.5] (7,0) to [out=90,in=180] (8.5,1) to [out=0,in=90]  (10,0);
\draw[dashed,->-=0.5] (7,0) to [out=-90,in=180] (8.5,-1) to [out=0,in=-90]  (10,0);

\draw[->-=0.5] (1,1) to [out=0,in=180] (7,1) to [out=0,in=90]  (8,0);
\draw[->-=0.5] (8,0) to [out=270,in=0] (7,-1) to [out=180,in=0]  (1,-1);                 
\end{tikzpicture} 
\caption{\label{fig: cycles Riemann surface}The selection of a canonical homology basis for $X$. The solid parts are on the first sheet and the dashed parts are on the second sheet.}\end{center}
\end{figure}
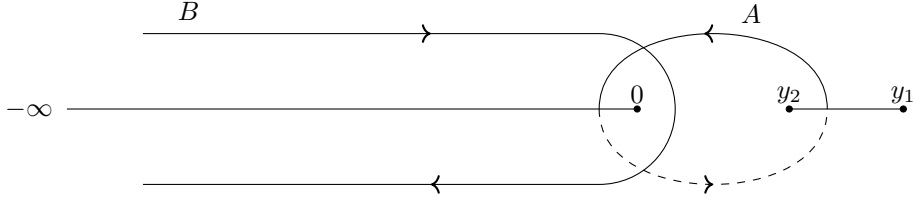
The unique $A$-normalized holomorphic one-form $\omega$ on $X$ is given by
\begin{align}
\omega = \frac{c_{0}~dz}{\sqrt{\mathcal{R}(z)}}, \qquad c_{0} = \left( \int_{A} \frac{dz}{\sqrt{\mathcal{R}(z)}} \right)^{-1}.
\end{align}
By construction $\int_{A}\omega = 1$ and the lattice parameter is given by $\tau = \int_{B}\omega \in i \mathbb{R}^{+}$. The associated $\theta$-function of the third kind $\theta(z) = \theta(z;\tau)$ is given by
\begin{align}
\theta(z) = \sum_{m=-\infty}^{\infty} e^{2\pi i m z}e^{\pi i m^{2} \tau}.
\end{align}
It is an entire function which satisfies (\ref{periodicity property of theta function}).
The zeros of $\theta(z)$ are the points $\frac{1}{2}+m_{1} + \frac{\tau}{2}+m_{2}\tau$, with $m_{1},m_{2} \in \mathbb{Z}$.
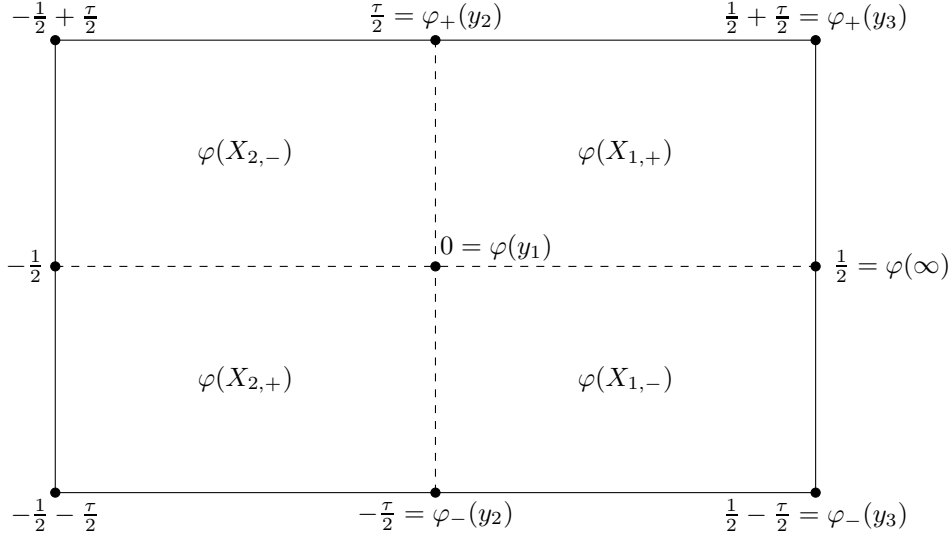
\begin{figure}
\begin{center}
\begin{tikzpicture}
\node at (0,0) {};
\draw (-5,-3) -- (5,-3);
\draw (-5,3) -- (5,3);
\draw (-5,-3) -- (-5,3);
\draw (5,-3) -- (5,3);

\draw[black,fill=black]  (-5,-3) circle [radius=0.06];
\node at (-5,-3.3) {$-\frac{1}{2}-\frac{\tau}{2}$};
\draw[black,fill=black]  (5,-3) circle [radius=0.06];
\node at (5,-3.3) {$\frac{1}{2}-\frac{\tau}{2}=\varphi_{-}(y_{3})$};
\draw[black,fill=black]  (-5,3) circle [radius=0.06];
\node at (-5,3.3) {$-\frac{1}{2}+\frac{\tau}{2}$};
\draw[black,fill=black]  (5,3) circle [radius=0.06];
\node at (5,3.3) {$\frac{1}{2}+\frac{\tau}{2}=\varphi_{+}(y_{3})$};

\draw[dashed] (0,-3) -- (0,3);
\draw[dashed] (-5,0) -- (5,0);
\draw[black,fill=black]  (-5,0) circle [radius=0.06];
\node at (-5.4,0) {$-\frac{1}{2}$};
\draw[black,fill=black]  (5,0) circle [radius=0.06];
\node at (6,0) {$\frac{1}{2}=\varphi(\infty)$};
\draw[black,fill=black]  (0,-3) circle [radius=0.06];
\node at (0,-3.3) {$-\frac{\tau}{2}=\varphi_{-}(y_{2})$};
\draw[black,fill=black]  (0,3) circle [radius=0.06];
\node at (0,3.3) {$\frac{\tau}{2}=\varphi_{+}(y_{2})$};
\draw[black,fill=black]  (0,0) circle [radius=0.06];
\node at (0.8,0.25) {$0=\varphi(y_{1})$};

\node at (2.5,1.5) {$\varphi(X_{1,+})$};
\node at (2.5,-1.5) {$\varphi(X_{1,-})$};
\node at (-2.5,1.5) {$\varphi(X_{2,-})$};
\node at (-2.5,-1.5) {$\varphi(X_{2,+})$};
\end{tikzpicture}
\end{center}
\caption{\label{fig: mapping Abel map}The image set of the function $\varphi$.}
\end{figure}
Let $X_{1}$ and $X_{2}$ denote the upper and lower sheet of $X$, respectively, and consider the set
\begin{align*}
\Lambda := \{u \in \mathbb{C}: -\tfrac{1}{2}\leq \Re[u] \leq  \tfrac{1}{2} \, \mbox{ and } \, -\tfrac{|\tau|}{2}\leq \Im[u] \leq  \tfrac{|\tau|}{2} \},
\end{align*}
together with the function
\begin{align}\label{def of u}
\varphi: \big( X_{1} \setminus (-\infty,y_{1}] \big) \cup \big( X_{2} \setminus (-\infty,y_{1}] \big) \to \Lambda, \qquad z \mapsto \varphi(z) = \int_{y_{1}}^{z} \omega,
\end{align}
where the path does not cross $(-\infty,y_{1}]$ and lies in the same sheet as $z$. We denote $X_{j,+}$ and $X_{j,-}$ for the strict upper and lower half plane of the $j$-th sheet. The function $\varphi$ maps $X_{1,+}, X_{1,-},X_{2,+}, X_{2,-}$ into the four quadrants of $\Lambda$ as shown in Figure \ref{fig: mapping Abel map}. For $z$ on the first sheet, $\varphi(z)$ has the following jumps
\begin{align}
& \varphi_{+}(z) + \varphi_{-}(z) = 0, & & z \in (y_{2},y_{1}), \label{jump1 for u} \\
& \varphi_{+}(z) + \varphi_{-}(z) = 1, & & z \in (-\infty,0), \label{jump2 for u} \\
& \varphi_{+}(z) - \varphi_{-}(z) = \tau, & & z \in (0,y_{2}). \label{jump3 for u}
\end{align}
At the four branch points $y_{1}$, $y_{2}$, $0$, $\infty$ on the first sheet, we have
\begin{align}\label{special values of u}
\varphi(y_{1}) = 0, \qquad \varphi_{\pm}(y_{2}) = \pm \frac{\tau}{2}, \qquad \varphi_{\pm}(0) = \frac{1}{2} \pm \frac{\tau}{2}, \qquad \varphi(\infty) = \frac{1}{2},
\end{align}
see Figure \ref{fig: mapping Abel map}. The Abel map is closely related to $\varphi$, so we denote it by $\varphi_{A}$, and is given by
\begin{align}\label{def of Abel}
\varphi_{A}: X \to \mathbb{C}/(\mathbb{Z} + \tau \mathbb{Z}), \qquad z \mapsto \varphi_{A}(z) := \varphi(z) \mod (\mathbb{Z} + \tau \mathbb{Z}).
\end{align}
The Abel map is a bijection with the torus $\mathbb{C}/ (\mathbb{Z} + \tau \mathbb{Z})$ and is analytic on the full Riemann surface (as opposed to $\varphi$ which presents discontinuities, as can be seen in \eqref{jump1 for u}-\eqref{jump3 for u}). There is an explicit expression for its inverse $\varphi_{A}^{-1}(u)$ in terms of the Jacobi elliptic function $\mathrm{sn}$ \cite[eq 22.2.4]{NIST}:
\begin{align*}
\varphi_{A}^{-1}(u) = y_{1} + (y_{2}-y_{1}) \,  \mathrm{sn}^{2}\left( \frac{i u \sqrt{y_{1}}}{2 c_{0}}, 1-\frac{y_{2}}{y_{1}} \right).
\end{align*} 
We introduce the following ratio of $\theta$-functions (cf. Baker-Akhiezer functions)
\begin{align*}
\mathcal{F}(z;\widehat{d},\nu) = \frac{\theta(z+\widehat{d}+\nu)}{ \theta(z+\widehat{d})}, \qquad z,\widehat{d}\in \mathbb{C}, \quad \nu \in \mathbb{R}.
\end{align*}
This function has a simple pole at $z = -\widehat{d}+\frac{1}{2}+\frac{\tau}{2} \mod (\mathbb{Z} + \tau \mathbb{Z})$. From \eqref{periodicity property of theta function} and \eqref{jump1 for u}-\eqref{jump3 for u}, for $z \in X_{1}$ we have 
\begin{align*}
& \mathcal{F}(\varphi_{+}(z);\widehat{d},\nu) = \mathcal{F}(\varphi_{-}(z);-\widehat{d},-\nu), & & z \in (-\infty,0) \cup (y_{2},y_{1}), \\
& \mathcal{F}(\varphi_{+}(z);\widehat{d},\nu) = e^{-2\pi i \nu}\mathcal{F}(\varphi_{-}(z);\widehat{d},\nu), & & z \in (0,y_{2}).
\end{align*}
Let us consider the matrix
\begin{equation}\label{Q1inftydef}
Q_{1}^{(\infty)}(z) := \begin{pmatrix}
\mathcal{F}(\varphi(z);-\widehat{d}_{1},-\nu) & \mathcal{F}(\varphi(z);\widehat{d}_{1},\nu) \\
\mathcal{F}(\varphi(z);\widehat{d}_{2},-\nu) & \mathcal{F}(\varphi(z);-\widehat{d}_{2},\nu)
\end{pmatrix}, \qquad z \in \mathbb{C}\setminus (-\infty,y_{1}],
\end{equation}
where $\varphi(z)$ is interpreted as the value of $\varphi$ at the point $z$ on the upper sheet $X_{1}$. The matrix $Q_{1}^{(\infty)}(z)$ satisfies the jump conditions
\begin{align*}
& Q_{1,+}^{(\infty)}(z) = Q_{1,-}^{(\infty)}(z)\begin{pmatrix}
0 & 1 \\ 1 & 0
\end{pmatrix}, & & z \in (-\infty,0)\cup (y_{2},y_{1}), \\
& Q_{1,+}^{(\infty)}(z) = Q_{1,-}^{(\infty)}(z)e^{2\pi i \nu \sigma_{3}}, & & z \in (0,y_{2}).
\end{align*}
Note that $\det Q_{1}^{(\infty)}$ is not identically $1$. To ensure that $Q_{1}^{(\infty)}$ has no pole at $\infty$, we assume from now that $\widehat{d}_{1},\widehat{d}_{2} \neq \frac{\tau}{2} \mod (\mathbb{Z} + \tau \mathbb{Z})$. We define 
\begin{align*}
\beta(z) = \sqrt[4]{\frac{z(z-y_{1})}{z-y_{2}}}, \qquad z \in \mathbb{C}\setminus \big((-\infty,0)\cup [y_{2},y_{1})\big),
\end{align*}
where the principal branch is chosen for each branch so $\beta(z) > 0$ for $z > y_{1}$. It can be verified that the matrix
\begin{equation}
Q_{2}^{(\infty)}(z) := \beta(z)^{\sigma_3}M^{-1}
\end{equation}
is analytic in $\mathbb{C}\setminus \big(  (-\infty,0]\cup [y_{2},y_{1}] \big)$, satisfies the jumps
\begin{equation}
Q_{2,+}^{(\infty)}(z) = Q_{2,-}^{(\infty)}(z)\begin{pmatrix}
0 & 1 \\ -1 & 0
\end{pmatrix}, \qquad z \in (-\infty,0) \cup (y_{2},y_{1}),
\end{equation}
with asymptotic behavior at $\infty$ given by
\begin{equation}
Q_{2}^{(\infty)}(z) = \big( I+\bigO(z^{-1}) \big)z^{\frac{\sigma_{3}}{4}}M^{-1}, \qquad \mbox{ as } z \to \infty.
\end{equation}
Therefore, the matrix
\begin{equation}\label{Q3}
Q_{3}^{(\infty)}(z) := \frac{\beta(z)^{\sigma_3}}{\sqrt{2}}\begin{pmatrix}
\ds \mathcal{F}(\varphi(z);-\widehat{d}_{1},-\nu) & \ds -i\mathcal{F}(\varphi(z);\widehat{d}_{1},\nu) \\
\ds -i\mathcal{F}(\varphi(z);\widehat{d}_{2},-\nu) & \ds \mathcal{F}(\varphi(z);-\widehat{d}_{2},\nu)
\end{pmatrix}
\end{equation}
has the jumps
\begin{align*}
& Q_{3,+}^{(\infty)}(z) = Q_{3,-}^{(\infty)}(z)\begin{pmatrix}
0 & 1 \\ -1 & 0
\end{pmatrix}, & & z \in (-\infty,0)\cup (y_{2},y_{1}), \\
& Q_{3,+}^{(\infty)}(z) = Q_{3,-}^{(\infty)}(z)e^{2\pi i \nu \sigma_{3}}, & & z \in (0,y_{2}).
\end{align*}
However, $Q_{3}^{(\infty)}(z)$ is not a solution to the RH problem for $P^{(\infty)}(z)$ because it does not have the required behavior at $\infty$. Indeed, as $z \to \infty$, $z\notin(-\infty,0]$, we have
\begin{align}
Q_{3}^{(\infty)}(z)Mz^{-\frac{\sigma_{3}}{4}} & = \frac{1}{2}\left(
\begin{array}{l}
\big[ \mathcal{F}(\varphi(z);-\widehat{d}_{1},-\nu)+\mathcal{F}(\varphi(z);\widehat{d}_{1},\nu) \big] \beta(z)z^{-\frac{1}{4}}  \\
\big[ \mathcal{F}(\varphi(z);-\widehat{d}_{2},\nu)-\mathcal{F}(\varphi(z);\widehat{d}_{2},-\nu) \big] i\beta(z)^{-1}z^{-\frac{1}{4}} 
\end{array}
\right. \cdots \nonumber \\ 
 & \cdots \left. \begin{array}{l}
\big[ \mathcal{F}(\varphi(z);-\widehat{d}_{1},-\nu)-\mathcal{F}(\varphi(z);\widehat{d}_{1},\nu) \big] i\beta(z)z^{\frac{1}{4}} \\
\big[ \mathcal{F}(\varphi(z);\widehat{d}_{2},-\nu)+\mathcal{F}(\varphi(z);-\widehat{d}_{2},\nu) \big] \beta(z)^{-1}z^{\frac{1}{4}}
\end{array} \right) = \tilde{Q} + \bigO(z^{-1})
\end{align}
where the leading coefficient $\tilde{Q}$ is given by
\begin{align*}
\tilde{Q} = \begin{pmatrix}
\mathcal{F}(\frac{1}{2};\widehat{d}_{1},\nu) & 2ic_{0}\mathcal{F}^{\prime}(\frac{1}{2};\widehat{d}_{1},\nu) \\
0 & \mathcal{F}(\frac{1}{2};-\widehat{d}_{2},\nu)
\end{pmatrix},
\end{align*}
and where we have used the expansions
\begin{align*}
& \varphi(z) = \frac{1}{2} - 2c_{0}z^{-\frac{1}{2}}+\bigO(z^{-\frac{3}{2}}), \\
& \mathcal{F}(\varphi(z);\widehat{d},\nu) = \mathcal{F}(\tfrac{1}{2};\widehat{d},\nu) - 2 c_{0} \mathcal{F}'(\tfrac{1}{2};\widehat{d},\nu) z^{-\frac{1}{2}} + \bigO(z^{-1}), \\
& \mathcal{F}(\varphi(z);-\widehat{d},-\nu) = \mathcal{F}(-\varphi(z)+1;\widehat{d},\nu) = \mathcal{F}(\tfrac{1}{2};\widehat{d},\nu) + 2 c_{0} \mathcal{F}'(\tfrac{1}{2};\widehat{d},\nu) z^{-\frac{1}{2}} + \bigO(z^{-1}),
\end{align*}
as $z \to \infty$, $z \notin (-\infty,0]$. Since $\widehat{d}_{1},\widehat{d}_{2} \neq \frac{\tau}{2} \mod (\mathbb{Z} + \tau \mathbb{Z})$, the quantities $\mathcal{F}(\frac{1}{2};\widehat{d}_{1},\nu)$ and $\mathcal{F}(\frac{1}{2};-\widehat{d}_{2},\nu)$ are well-defined. We also assume that 
$\Im[\widehat{d}_{1}], \Im[\widehat{d}_{2}] \neq (\frac{1}{2} + n)\Im[\tau], n \in \mathbb{Z}$, 
so that $\mathcal{F}(\frac{1}{2};\widehat{d}_{1},\nu)$ and $\mathcal{F}(\frac{1}{2};-\widehat{d}_{2},\nu)$ are both different from $0$ for all values of $\nu \in \mathbb{R}$. This implies that $\tilde{Q}$ is invertible, therefore we can normalize $Q_{3}^{(\infty)}(z)$ at $\infty$. Thus we define
\begin{align}\label{def of P inf hat two cuts}
&P^{(\infty)}(z) := \tilde{Q}^{-1}Q_3^{(\infty)}(z)=\begin{pmatrix}
\ds \frac{1}{\mathcal{G}(\frac{1}{2})} & \ds \frac{-ic_{\mathcal{F}}}{\mathcal{H}(\frac{1}{2})} \\
0 & \ds \frac{1}{\mathcal{H}(\frac{1}{2})}
\end{pmatrix}\frac{\beta(z)^{\sigma_3}}{\sqrt{2}}\begin{pmatrix}
\ds \mathcal{G}(-\varphi(z)) & \ds -i\mathcal{G}(\varphi(z)) \\
\ds -i\mathcal{H}(-\varphi(z)) & \ds \mathcal{H}(\varphi(z))
\end{pmatrix}
\end{align}
where
\begin{align*}
& \mathcal{G}(z) = \mathcal{F}(z;\widehat{d}_{1},\nu),  & & \mathcal{H}(z) = \mathcal{F}(z;-\widehat{d}_{2},\nu), & & c_{\mathcal{F}} = 2c_{0} (\log\mathcal{G})'(\tfrac{1}{2}).
\end{align*}
The above matrix $P^{(\infty)}(z)$ has the same jumps as $Q_{3}^{(\infty)}(z)$ (since the multiplication by $\tilde{Q}^{-1}$ is taken on the left), and has the prescribed asymptotic behavior as $z\to\infty$ given by \eqref{eq:Pinf asympinf}. It remains to choose $\widehat{d}_{1}$ and $\widehat{d}_{2}$ appropriately so that $P^{(\infty)}(z)$ has no poles. Since $\beta^{-1}$ vanishes at $z = y_{2}$, we choose $\widehat{d}_{2} = \frac{1}{2}$ so that the pole of $\mathcal{H}(\pm \varphi(z))$ also lies at $y_{2}$. Thus at the branch point $y_{2}$, we have
\begin{equation}
\mathcal{H}(\pm \varphi(z))\beta(z)^{-1} = \bigO((z-y_{2})^{-\frac{1}{4}}), \qquad \mbox{ as } z \to y_{2}.
\end{equation}
Let us now inspect the first line of the right-most matrix in \eqref{def of P inf hat two cuts}. Since $\beta$ vanishes at $0$ and $y_{1}$, we have two choices for $\widehat{d}_{1}$ (see \eqref{special values of u}):
\begin{equation}
\widehat{d}_{1} = 0 \qquad \mbox{ or } \qquad \widehat{d}_{1} = \frac{1}{2}+\frac{\tau}{2}.
\end{equation}
However, the parameter $\nu$ will be such that $\nu \to -\infty$ as $r \to + \infty$, so we need $\Im[\widehat{d}_{1}] \neq \frac{\Im[\tau]}{2}$ in order to ensure that $\tilde{Q}$ is invertible, see the discussion above \eqref{def of P inf hat two cuts}. Therefore we choose $\widehat{d}_{1} = 0$. The choice for $\nu$ is imposed from the jumps on $(0,y_{2})$, see \eqref{jumps for Pinf on (0,y2)}, therefore we must have
\begin{align}\label{nu_r_relation}
\nu = - \frac{\Omega}{2\pi}r^\frac{3}{2}.
\end{align}
To summarize, the matrix-valued function $P^{(\infty)}(z)$ given by \eqref{def of P inf hat two cuts} with $\widehat{d}_{1} = 0$, $\widehat{d}_{2} = \frac{1}{2}$ and $\nu = - \frac{\Omega}{2\pi}r^\frac{3}{2}$ satisfies the RH problem for $P^{(\infty)}$. Furthermore, we have $\mathcal{H}(z) = \mathcal{G}(z-\frac{1}{2})$, so that \eqref{def of P inf hat two cuts} can be rewritten as
\begin{align}
&P^{(\infty)}(z) := \tilde{Q}^{-1}Q_3^{(\infty)}(z)=\begin{pmatrix}
\ds \frac{1}{\mathcal{G}(\frac{1}{2})} & \ds \frac{-ic_{\mathcal{F}}}{\mathcal{G}(0)} \\
0 & \ds \frac{1}{\mathcal{G}(0)}
\end{pmatrix}\frac{\beta(z)^{\sigma_3}}{\sqrt{2}}\begin{pmatrix}
\ds \mathcal{G}(-\varphi(z)) & \ds -i\mathcal{G}(\varphi(z)) \\
\ds -i\mathcal{G}(-\varphi(z)-\tfrac{1}{2}) & \ds \mathcal{G}(\varphi(z)-\tfrac{1}{2}) 
\end{pmatrix} \label{def of P inf hat two cuts FINAL}
\end{align}
where
\begin{align*}
& \mathcal{G}(z) = \mathcal{F}(z;0,\nu) = \frac{\theta(z + \nu)}{\theta(z)}.
\end{align*}
The function $\mathcal{G}$ is analytic in $\mathbb{C}\setminus \{\frac{1+\tau}{2}+n+m\tau:n,m\in \mathbb{Z}\}$ and satisfies
\begin{align}\label{periodicity of G}
\mathcal{G}(z+1) = \mathcal{G}(z), \qquad  \mathcal{G}(z+\tau) = e^{-2\pi i \nu}\mathcal{G}(z).
\end{align}
In the following subsections, we obtain expressions for the first four terms of the asymptotics of $P^{(\infty)}(z)$ as $z \to y_{j}$, $\Im[z]> 0$, $j=1,2,3$. These coefficients will be needed in the computations of Section \ref{section:smallnorm}.
\subsection{Asymptotics of $P^{(\infty)}(z)$ as $z \to y_{1}$}\label{section: expansion Pinf at y1}

The asymptotics for $\beta(z)$ and $\varphi(z)$ as $z \to y_{1}$, $z \in \mathbb{C}\setminus (-\infty,y_{1}]$, are given by
\begin{align}
& \beta(z) = \beta_{y_{1}}^{(\frac{1}{4})}(z-y_{1})^{\frac{1}{4}} + \beta_{y_{1}}^{(\frac{5}{4})}(z-y_{1})^{\frac{5}{4}} + \bigO\big((z-y_{1})^{\frac{9}{4}}\big), \nonumber \\
& \beta_{y_{1}}^{(\frac{1}{4})} = \frac{y_{1}^{\frac{1}{4}}}{(y_{1}-y_{2})^{\frac{1}{4}}}, \qquad \beta_{y_{1}}^{(\frac{5}{4})} = \frac{-y_{2}}{4 y_{1}^{\frac{3}{4}}(y_{1}-y_{2})^{\frac{5}{4}}}, \nonumber \\
& \varphi(z) = \varphi_{y_{1}}^{(\frac{1}{2})} (z-y_{1})^{\frac{1}{2}} + \varphi_{y_{1}}^{(\frac{3}{2})} (z-y_{1})^\frac{3}{2} + \bigO((z-y_{1})^{\frac{5}{2}}), \nonumber \\
& \varphi_{y_{1}}^{(\frac{1}{2})} = \frac{2c_{0}}{\sqrt{y_{1}}\sqrt{y_{1}-y_{2}}}, \qquad \varphi_{y_{1}}^{(\frac{3}{2})} = \frac{-c_{0}(2y_{1}-y_{2})}{3 y_{1}^\frac{3}{2}(y_{1}-y_{2})^\frac{3}{2}}. \label{asymp u and beta at y1}
\end{align}
The asymptotics for $P^{(\infty)}(z)$ as $z \to y_{1}$ are given by
\begin{align*}
& P^{(\infty)}(z) = \sum_{j=0}^{3} (P^{(\infty)})_{y_{1}}^{(-\frac{1}{4}+\frac{j}{2})}(z-y_{1})^{-\frac{1}{4}+\frac{j}{2}} + \bigO((z-y_{1})^{\frac{7}{4}}), \\
& (P^{(\infty)})_{y_{1}}^{(-\frac{1}{4})} =  \frac{-\mathcal{G}(\frac{1}{2})}{\sqrt{2}\beta_{y_{1}}^{(\frac{1}{4})}\mathcal{G}(0)} \begin{pmatrix}
c_{\mathcal{F}} & ic_{\mathcal{F}} \\ i & -1 \end{pmatrix}, \\
& (P^{(\infty)})_{y_{1}}^{(\frac{1}{4})} = \frac{\beta_{y_1}^{(\frac{1}{4})}\mathcal{G}(0)}{\sqrt{2}\mathcal{G}(\frac{1}{2})}\begin{pmatrix} 1 & -i \\ 0 & 0 \end{pmatrix}+\frac{\varphi_{y_1}^{(\frac{1}{2})}\mathcal{G}'(\frac{1}{2})}{\sqrt{2}\beta_{y_1}^{(\frac{1}{4})}\mathcal{G}(0)}\begin{pmatrix} c_{\mathcal{F}} & -i c_{\mathcal{F}} \\ i & 1 \end{pmatrix}, \\
&(P^{(\infty)})_{y_1}^{(\frac{3}{4})}=\frac{-\varphi_{y_1}^{(\frac{1}{2})}\beta_{y_1}^{(\frac{1}{4})}\mathcal{G}'(0)}{\sqrt{2}\mathcal{G}(\frac{1}{2})}\begin{pmatrix} 1 & i \\ 0 & 0 \end{pmatrix}+\frac{\frac{\beta_{y_1}^{(\frac{5}{4})}}{\beta_{y_1}^{(\frac{1}{4})}}\mathcal{G}(\frac{1}{2})-\frac{1}{2}(\varphi_{y_1}^{(\frac{1}{2})})^2\mathcal{G}''(\frac{1}{2})}{\sqrt{2}\beta_{y_1}^{(\frac{1}{4})}\mathcal{G}(0)}\begin{pmatrix} c_{\mathcal{F}} & ic_{\mathcal{F}} \\ i & -1 \end{pmatrix}, \\
&(P^{(\infty)})_{y_1}^{(\frac{5}{4})}=\frac{\beta_{y_1}^{(\frac{5}{4})}\mathcal{G}(0)+\frac{1}{2}\beta_{y_1}^{(\frac{1}{4})}(\varphi_{y_1}^{(\frac{1}{2})})^2\mathcal{G}''(0)}{\sqrt{2}\mathcal{G}(\frac{1}{2})}\begin{pmatrix} 1 & -i \\ 0 & 0 \end{pmatrix} \\
&+\frac{\varphi_{y_1}^{(\frac{3}{2})}\mathcal{G}'(\frac{1}{2})+\frac{1}{6}(\varphi_{y_1}^{(\frac{1}{2})})^3\mathcal{G}'''(\frac{1}{2})-\frac{\beta_{y_1}^{(\frac{5}{4})}}{\beta_{y_1}^{(\frac{1}{4})}}\varphi_{y_1}^{(\frac{1}{2})}\mathcal{G}'(\frac{1}{2})}{\sqrt{2}\beta_{y_1}^{(\frac{1}{4})}\mathcal{G}(0)}\begin{pmatrix} c_{\mathcal{F}} & -ic_{\mathcal{F}} \\ i & 1 \end{pmatrix}.
\end{align*}

% & (P^{(\infty)})_{y_{1}}^{(\frac{1}{4})} = \frac{\beta_{y_{1}}^{(\frac{1}{4})}\mathcal{F}(0;\widehat{d}_{1},\nu)\mathcal{F}(\frac{1}{2}; -\widehat{d}_{2},\nu)}{\sqrt{2}\mathcal{F}(\frac{1}{2}; \widehat{d}_{1},\nu)} \begin{pmatrix}
% 1 & -i \\ 0 & 0 \end{pmatrix} \\
% & + \frac{u_{y_{1}}^{(\frac{1}{2})}\mathcal{F}'(0;-\widehat{d}_{2},\nu)}{\sqrt{2}\beta_{y_{1}}^{(\frac{1}{4})}} \begin{pmatrix} 2c_{0} (\log \mathcal{F})'(\frac{1}{2}; \widehat{d}_{1},\nu) & -2ic_{0}(\log \mathcal{F})'(\frac{1}{2}; \widehat{d}_{1},\nu) \\ i & 1 \end{pmatrix},

\subsection{Asymptotics of $P^{(\infty)}(z)$ as $z \to y_{2}$ from $\Im[z]> 0$}\label{section: expansion Pinf at y2}
%First, we note the relations
%\begin{align*}
%\mathcal{G}(z - \tfrac{\tau}{2}) = e^{2 \pi i \nu}\mathcal{G}(z + \tfrac{\tau}{2}),
%\end{align*}
%which imply
%\begin{align*}
%(\partial_{z}^{j}\mathcal{G})( -\tfrac{\tau}{2}) = e^{2 \pi i \nu}(\partial_{z}^{j}\mathcal{G})( \tfrac{\tau}{2}), \qquad j \geq 0.
%\end{align*}
We recall that $\frac{1+\tau}{2}$ is a simple zero of $z \mapsto \theta(z)$. Thus the function $z \mapsto \mathcal{G}(z)$ has a simple pole at $z = \frac{1+\tau}{2}$. Since $\varphi_{+}(y_{2}) = \frac{\tau}{2}$, the entries on the second row of $P^{(\infty)}(z)$ blow up as $z \to y_{2}$, $\Im[z]> 0$. Let us define
\begin{align}\label{Gtilde}
z \mapsto \widetilde{\mathcal{G}}(z) := \mathcal{G}(z) (z- \tfrac{1+\tau}{2}).
\end{align}
Using the relations \eqref{periodicity of G} and \eqref{Gtilde}, before expanding $P^{(\infty)}(z)$ as $z \to y_{2}, \Im[z]>0$, we rewrite $P^{(\infty)}(z)$ as follows
\begin{align}\label{rewriting of Pinf at y2}
&P^{(\infty)}(z) =\begin{pmatrix}
\ds \frac{1}{\mathcal{G}(\frac{1}{2})} & \ds \frac{-ic_{\mathcal{F}}}{\mathcal{G}(0)} \\
0 & \ds \frac{1}{\mathcal{G}(0)}
\end{pmatrix}\frac{\beta(z)^{\sigma_3}}{\sqrt{2}}\begin{pmatrix}
\ds e^{2\pi i \nu}\mathcal{G}(\tau-\varphi(z)) & \ds -i\mathcal{G}(\varphi(z)) \\
\ds -ie^{2\pi i \nu}\frac{\widetilde{\mathcal{G}}(-\varphi(z)+\tfrac{1}{2}+\tau)}{\frac{\tau}{2}-\varphi(z)} & \ds \frac{\widetilde{\mathcal{G}}(\varphi(z)+\tfrac{1}{2})}{\varphi(z)-\frac{\tau}{2}}
\end{pmatrix}
\end{align}
The asymptotics for $\beta(z)$ and $\varphi(z)$ as $z \to y_{2}$, $\Im[z]> 0$, are of the form
\begin{align*}
& \beta(z) = \beta_{y_{2}}^{(-\frac{1}{4})}(z-y_{2})^{-\frac{1}{4}} + \beta_{y_{2}}^{(\frac{3}{4})}(z-y_{2})^{\frac{3}{4}} + \bigO\big((z-y_{2})^{\frac{7}{4}}\big), \\
& \varphi(z) = \frac{\tau}{2} +  \varphi_{y_{2}}^{(\frac{1}{2})} (z-y_{2})^{\frac{1}{2}} + \varphi_{y_{2}}^{(\frac{3}{2})} (z-y_{2})^\frac{3}{2} + \bigO((z-y_{2})^{\frac{5}{2}}),
\end{align*}
and the asymptotics for $P^{(\infty)}(z)$ as $z \to y_{2}$, $\Im[z]>0$, are of the form
\begin{align*}
& P^{(\infty)}(z) = \sum_{j=0}^{3} (P^{(\infty)})_{y_{2}}^{(-\frac{1}{4}+\frac{j}{2})}(z-y_{2})^{-\frac{1}{4}+\frac{j}{2}} + \bigO((z-y_{2})^{\frac{7}{4}}).
\end{align*}
The explicit values for $(P^{(\infty)})_{y_{2}}^{(-\frac{1}{4}+\frac{j}{2})}$, $j=0,...,3$ are all needed to prove Theorem \ref{thm: main result}. They can be obtained after long but straightforward computations in terms of $\beta_{y_{2}}^{(-\frac{1}{4})}$, $\beta_{y_{2}}^{(\frac{3}{4})}$, $\varphi_{y_{2}}^{(\frac{1}{2})}$, $\varphi_{y_{2}}^{(\frac{3}{2})}$, $\mathcal{G}(0)$, $\mathcal{G}(\frac{1}{2})$, $\mathcal{G}^{(k)}(\frac{\tau}{2})$, $\widetilde{\mathcal{G}}^{(k)}(\frac{1+\tau}{2})$, $k=0,1,2,3$, but for conciseness we do not write them down.

\subsection{Asymptotics of $P^{(\infty)}(z)$ as $z \to y_{3}=0$ from $\Im[z]> 0$}\label{section: expansion Pinf at y3}
Using the relations \eqref{periodicity of G} and \eqref{Gtilde}, before expanding $P^{(\infty)}(z)$ as $z \to y_{3}, \Im[z]>0$, we rewrite $P^{(\infty)}(z)$ as follows
\begin{align*}
&P^{(\infty)}(z) =\begin{pmatrix}
\ds \frac{1}{\mathcal{G}(\frac{1}{2})} & \ds \frac{-ic_{\mathcal{F}}}{\mathcal{G}(0)} \\
0 & \ds \frac{1}{\mathcal{G}(0)}
\end{pmatrix}\frac{\beta(z)^{\sigma_3}}{\sqrt{2}}\begin{pmatrix}
\ds e^{2\pi i \nu}\frac{\widetilde{\mathcal{G}}(1+\tau-\varphi(z))}{\frac{1+\tau}{2}-\varphi(z)} & \ds -i\frac{\widetilde{\mathcal{G}}(\varphi(z))}{\varphi(z)-\frac{1+\tau}{2}} \\
\ds -ie^{2\pi i \nu}\mathcal{G}(\tfrac{1}{2}+\tau-\varphi(z)) & \ds \mathcal{G}(\varphi(z)-\tfrac{1}{2})
\end{pmatrix}
\end{align*}
The asymptotics for $\beta(z)$ and $\varphi(z)$ as $z \to y_{3}=0$, $\Im[z]> 0$, are of the form
\begin{align*}
& \beta(z) = \beta_{y_{3}}^{(\frac{1}{4})}z^{\frac{1}{4}} + \beta_{y_{3}}^{(\frac{5}{4})}z^{\frac{5}{4}} + \bigO\big(z^{\frac{9}{4}}\big), \\
& \varphi(z) = \frac{1+\tau}{2} +  \varphi_{y_{3}}^{(\frac{1}{2})} z^{\frac{1}{2}} + \varphi_{y_{3}}^{(\frac{3}{2})} z^\frac{3}{2} + \bigO(z^{\frac{5}{2}}),
\end{align*}
and the asymptotics for $P^{(\infty)}(z)$ as $z \to y_{3}=0$, $\Im[z]>0$, are of the form
\begin{align*}
& P^{(\infty)}(z) = \sum_{j=0}^{3} (P^{(\infty)})_{y_{3}}^{(-\frac{1}{4}+\frac{j}{2})}z^{-\frac{1}{4}+\frac{j}{2}} + \bigO(z^{\frac{7}{4}}).
\end{align*}
The explicit values for $(P^{(\infty)})_{y_{3}}^{(-\frac{1}{4}+\frac{j}{2})}$, $j=0,...,3$ are also all needed to prove Theorem \ref{thm: main result}. They can be expressed after long but straightforward computations in terms of $\beta_{y_{3}}^{(\frac{1}{4})}$, $\beta_{y_{3}}^{(\frac{5}{4})}$, $\varphi_{y_{3}}^{(\frac{1}{2})}$, $\varphi_{y_{3}}^{(\frac{3}{2})}$, $\mathcal{G}(0)$, $\mathcal{G}(\frac{1}{2})$, $\mathcal{G}^{(k)}(\frac{\tau}{2})$, $\widetilde{\mathcal{G}}^{(k)}(\frac{1+\tau}{2})$, $k=0,1,2,3$, and we do not write them down for conciseness.
\subsection{Asymptotics of $P^{(\infty)}(z)^{-1}$ as $z \to y_{k}$ from $\Im[z]> 0$, $k=1,2,3$}
Since $\det P^{(\infty)} \equiv 1$, we have
\begin{align*}
P^{(\infty)}(z)^{-1} = \begin{pmatrix}
P_{22}^{(\infty)}(z) & -P_{12}^{(\infty)}(z) \\
- P_{21}^{(\infty)}(z) & P_{11}^{(\infty)}(z)
\end{pmatrix},
\end{align*}
and we deduce from Subsections \ref{section: expansion Pinf at y1}, \ref{section: expansion Pinf at y2} and \ref{section: expansion Pinf at y3} the following asymptotic expansions
\begin{align*}
& P^{(\infty)}(z)^{-1} = \sum_{j=0}^{3} (P^{(\infty)}_{\mathrm{inv}})_{y_{k}}^{(-\frac{1}{4}+\frac{j}{2})}(z-y_{k})^{-\frac{1}{4}+\frac{j}{2}} + \bigO((z-y_{k})^{\frac{7}{4}}), \qquad k=1,2,3, \\
& (P^{(\infty)}_{\mathrm{inv}})_{y_{k}}^{(-\frac{1}{4}+\frac{j}{2})} = \begin{pmatrix}
(P^{(\infty)})_{y_{k},22}^{(-\frac{1}{4}+\frac{j}{2})} & -(P^{(\infty)})_{y_{k},12}^{(-\frac{1}{4}+\frac{j}{2})} \\
-(P^{(\infty)})_{y_{k},21}^{(-\frac{1}{4}+\frac{j}{2})} & (P^{(\infty)})_{y_{k},11}^{(-\frac{1}{4}+\frac{j}{2})}
\end{pmatrix}, \qquad k=1,2,3, \quad j=0,1,2,3.
\end{align*}

\section{Local parametrices}\label{subsec:Besselparametrix}
Let us choose $\delta > 0$ such that
\begin{align*}
\delta \leq \frac{\min \{-x_{1},x_{1}-x_{2},x_{2}-x_{3}\}}{3} \quad \mbox{ and } \quad \mathbb{D}_{y_j}:=\{z \in \mathbb{C}: |z-y_{j}| < \delta \} \subset N_{0}, \qquad j=1,2,3,
\end{align*}
where $N_{0}$ is as described in Lemma \ref{g_lemma}. The local parametrix $P^{(y_{j})}(z)$, $j=1,2,3$, is defined for $z \in \mathbb{D}_{y_{j}}$ as the solution to a RH problem with the same jumps as $S$ inside $\mathbb{D}_{y_{j}}$. Furthermore, we require that $P^{(y_{j})}(z) = \bigO(\log|z-y_j|)$ as $z \to y_{j}$ and that $P^{(y_{j})}$ matches with $P^{(\infty)}$ on the boundary of $\mathbb{D}_{y_j}$, in the sense that
\begin{equation}\label{lol2}
P^{(y_{j})}(z) = (I + o(1))P^{(\infty)}(z), \qquad \mbox{as } r \to + \infty
\end{equation} 
uniformly for $z \in \partial \mathbb{D}_{y_j}$. These constructions are standard (so we give relatively short explanations) and given in terms of an explicitly solvable model RH problem $\Phi_{\mathrm{Be}}$ which is expressed in terms of Bessel functions. This model RH problem was first derived in \cite{DIZ} and we present it in Appendix \ref{Section:Appendix} for convenience. 
%For an in depth discussion on how to construct a parametrix using the Bessel RH problem, we refer the reader to \cite[Section 4.3]{BKT15}.

\subsection{Parametrix at $y_1$}

Let $f_{y_1}$ be the conformal map from $\mathbb{D}_{y_1}$ to a neighborhood of $0$ defined by
\begin{equation}\label{y1coord}
f_{y_{1}}(z)=\frac{g^2(z)}{4}.
\end{equation}
The expansion of $f_{y_1}(z)$ as $z \to y_{1}$ is given by
\begin{align}
& f_{y_{1}}(z) = c_{y_{1}}(z-y_{1})\big( 1+ c_{y_{1}}^{(2)}(z-y_{1}) + c_{y_{1}}^{(3)}(z-y_{1})^{2} + \bigO((z-y_{1})^{3}) \big), \label{xiy1_expansion} \\
& c_{y_{1}} = \frac{p^2(y_1)}{y_1(y_1-y_2)}, ~~~ c_{y_{1}}^{(2)} = \frac{2p'(y_1)}{3p(y_1)}-\frac{2y_1-y_2}{3y_1(y_1-y_2)}, \nonumber \\
& c_{y_{1}}^{(3)} = \frac{1}{45}\left( \frac{23 y_{1}^{2} - 23 y_{1}y_{2} + 8y_{2}^{2}}{y_{1}^{2}(y_{1}-y_{2})^{2}} + \frac{5p'(y_{1})^{2}}{p(y_{1})^{2}}-\frac{14(2y_{1}-y_{2})p'(y_{1})}{y_{1}(y_{1}-y_{2})p(y_{1})}+\frac{9p''(y_{1})}{p(y_{1})} \right). \nonumber
\end{align}
We have freedom in the choice of $\gamma_{\pm}\subset N_{0}$. We choose the lenses such that $f_{y_{1}}(\gamma_{+}\cap \mathbb{D}_{y_{1}}) \subset e^{\frac{2\pi i}{3}}\mathbb{R}^{+}$ and $f_{y_{1}}(\gamma_{-}\cap \mathbb{D}_{y_{1}}) \subset e^{\frac{2\pi i}{3}}\mathbb{R}^{-}$. We conclude from the RH problem for $\Phi_{\mathrm{Be}}$ presented in Section \ref{Section:Appendix} that the matrix
\begin{equation}
\Phi_{\text{Be}}\left(r^{3} f_{y_1}(z)\right)e^{-r^\frac{3}{2}g(z)\sigma_3}
\end{equation}
has the same jumps as $S(z)$ inside the disk $z\in\mathbb{D}_{y_1}$, see \eqref{SInf0 jump}-\eqref{S0y2 jump}. Let us define $P^{(y_{1})}$ by
\begin{align}
& P^{(y_1)}(z):=E_{y_{1}}(z)\Phi_{\text{Be}}\left(r^{3} f_{y_1}(z) \right)e^{-r^\frac{3}{2}g(z)\sigma_3}, \label{Py1} \\
& E_{y_{1}}(z) = P^{(\infty)}(z)M^{-1}\left(2\pi r^{\frac{3}{2}}f_{y_1}(z)^{\frac{1}{2}}\right)^{\frac{\sigma_3}{2}}. \label{def of Ey1}
\end{align}
It can be verified from the jumps of $P^{(\infty)}(z)$ that the prefactor $E_{y_{1}}(z)$ is analytic for $z\in\mathbb{D}_{y_1}$. Thus, $P^{(y_1)}(z)$ has the (exact) same jumps as and endpoint behavior, see \eqref{local behavior near 0 of P_Be}, as $S(z)$ for $z\in\mathbb{D}_{y_1}$. Furthermore, due to \eqref{xiy1_expansion} and \eqref{large z asymptotics Bessel}, we have 
\begin{equation}\label{Py1asymp}
P^{(y_1)}(z)P^{(\infty)}(z)^{-1}=I + \frac{P^{(\infty)}(z)\Phi_{\mathrm{Be,1}} P^{(\infty)}(z)^{-1}}{r^\frac{3}{2} f_{y_{1}}(z)^{\frac{1}{2}}} + \bigO(r^{-3})
\end{equation}
as $r\to +\infty$ uniformly for $z\in\partial\mathbb{D}_{y_1}$. In Section \ref{section:integration1} we will need the first columns of $E_{y_{1}}(y_{1})$ and $E_{y_{1}}'(y_{1})$. After some computations, we find
\begin{align}
E_{y_{1}}(y_{1}) &= \frac{-\sqrt{2\pi}c_{y_1}^{\frac{1}{4}}\mathcal{G}(\frac{1}{2})}{\beta_{y_{1}}^{(\frac{1}{4})}\mathcal{G}(0)}\begin{pmatrix} c_{\mathcal{F}} & 0 \\ i & 0 \end{pmatrix}r^{\frac{3}{4}} +\begin{pmatrix}
0 & \star  \\
0 & \star
\end{pmatrix}r^{-\frac{3}{4}}, \label{Ey1y1} \\
E_{y_{1}}'(y_{1})&=\sqrt{2\pi}c_{y_1}^{\frac{1}{4}}r^{\frac{3}{4}}\Bigg\{\frac{\Big(\tfrac{\beta_{y_{1}}^{(\frac{5}{4})}}{\beta_{y_{1}}^{(\frac{1}{4})}} -\frac{1}{4}c_{y_{1}}^{(2)}\Big)\mathcal{G}(\frac{1}{2}) - \frac{1}{2} (\varphi_{y_{1}}^{(\frac{1}{2})})^{2} \mathcal{G}''(\frac{1}{2})}{\beta_{y_{1}}^{(\frac{1}{4})} \mathcal{G}(0)} \begin{pmatrix}
    c_{\mathcal{F}} & 0 \\ i & 0
    \end{pmatrix} \nonumber \\
    & \quad -\frac{\varphi_{y_1}^{(\frac{1}{2})}\beta_{y_1}^{(\frac{1}{4})}\mathcal{G}'(0)}{\mathcal{G}(\frac{1}{2})}\begin{pmatrix} 1 & 0 \\ 0 & 0 \end{pmatrix} \Bigg\} + \begin{pmatrix}
0 & \star  \\
0 & \star
\end{pmatrix}r^{-\frac{3}{4}}     \nonumber
\end{align}
where $\star$ denotes unnecessary constants.
\subsection{Parametrix at $y_2$}

Let $f_{y_2}$ be the conformal map from $\mathbb{D}_{y_2}$ to a neighborhood of $0$ defined by
\begin{equation}\label{y2coord}
f_{y_2}(z) =-\frac{1}{4} \left(g(z)\mp\frac{i\Omega}{2}\right)^2,
\end{equation}
where we take $-/+$ when $z$ is above/below the real axis. Its expansion as $z \to y_{2}$ is given by
\begin{align}
& f_{y_{2}}(z) = c_{y_{2}}(z-y_{2})\big( 1+ c_{y_{2}}^{(2)}(z-y_{2}) + c_{y_{2}}^{(3)}(z-y_{2})^{2}+ \bigO((z-y_{2})^{3}) \big), \label{xiy2_expansion} \\
& c_{y_{2}} = \frac{p^2(y_2)}{y_2(y_1-y_2)}>0, ~~~ c_{y_{2}}^{(2)} = \frac{2p'(y_2)}{3p(y_2)}-\frac{y_1-2y_2}{3y_2(y_1-y_2)}, \nonumber \\
& c_{y_{2}}^{(3)} = \frac{1}{45}\left( \frac{8y_{1}^{2}-23 y_{1}y_{2}+23y_{2}^{2}}{(y_{1}-y_{2})^{2}y_{2}^{2}}+\frac{5p'(y_{2})^{2}}{p(y_{2})^{2}}-\frac{14(y_{1}-2y_{2})p'(y_{2})}{(y_{1}-y_{2}) y_{2} p(y_{2})} + \frac{9p''(y_{2})}{p(y_{2})} \right) . \nonumber
\end{align}
We choose the lenses $\gamma_{+} \subset N_{0}$ and $\gamma_{-}\subset N_{0}$ such that $-f_{y_{2}}(\gamma_{\pm} \cap \mathbb{D}_{y_{2}}) \subset e^{\pm\frac{2\pi i}{3}}\mathbb{R}^{+}$, and we define 
\begin{align}
& P^{(y_{2})}(z) = E_{y_{2}}(z)  \sigma_{3} \Phi_{\mathrm{Be}}(-r^{3} f_{y_{2}}(z))\sigma_{3} e^{-r^\frac{3}{2}g(z)\sigma_{3}}, \label{Py2} \\
& E_{y_{2}}(z) = P^{(\infty)}(z)e^{\pm\frac{i \Omega}{2} r^\frac{3}{2}\sigma_{3}}M\Big( 2\pi r^{\frac{3}{2}}(-f_{y_{2}}(z))^{\frac{1}{2}} \Big)^{\frac{\sigma_{3}}{2}}. \nonumber
\end{align}
It can be verified that $E_{y_{2}}(z)$ is analytic in the disk $\mathbb{D}_{y_{2}}$ and that $P^{(y_2)}(z)=\Or(\log|z-y_2|)$ as $z\to y_2$, see \eqref{local behavior near 0 of P_Be}. From \eqref{y2coord} and \eqref{large z asymptotics Bessel}, we also conclude that 
\begin{equation}\label{Py2asymp}
P^{(y_2)}(z)P^{(\infty)}(z)^{-1}= I + \frac{P^{(\infty)}(z)e^{\pm \frac{i \Omega}{2}r^\frac{3}{2}\sigma_{3}}\sigma_{3} \Phi_{\mathrm{Be},1}\sigma_{3} e^{\mp \frac{i\Omega}{2}r^\frac{3}{2}\sigma_3}P^{(\infty)}(z)^{-1} }{r^\frac{3}{2}(-f_{y_{2}}(z))^{\frac{1}{2}}} + \bigO(r^{-3})
\end{equation}
as $r\to +\infty$ uniformly for $z\in\partial\mathbb{D}_{y_2}$. Furthermore, we also verify that
\begin{align}
E_{y_{2}}(y_{2}) &= \, e^{-\frac{\pi i}{4}}\sqrt{2\pi}  c_{y_{2}}^{\frac{1}{4}}e^{\pi i \nu} r^{\frac{3}{4}} \Bigg\{ \frac{\widetilde{\mathcal{G}}(\frac{1+\tau}{2})}{\varphi_{y_{2}}^{(\frac{1}{2})} \beta_{y_{2}}^{(-\frac{1}{4})}\mathcal{G}(0)} \begin{pmatrix}
c_{\mathcal{F}} & 0 \\
i & 0
\end{pmatrix} + \frac{\beta_{y_{2}}^{(-\frac{1}{4})} \mathcal{G}(\frac{\tau}{2})}{\mathcal{G}(\frac{1}{2})} \begin{pmatrix}
1 & 0 \\
0 & 0
\end{pmatrix} \Bigg\} + \begin{pmatrix}
0 & \star  \\
0 & \star
\end{pmatrix}r^{-\frac{3}{4}}, \label{Ey2y2} \\
E_{y_{2}}'(y_{2}) &= \, e^{-\frac{\pi i}{4}} \sqrt{2\pi} c_{y_{2}}^{\frac{1}{4}} e^{ \pi i \nu} r^{\frac{3}{4}} \Bigg\{ \frac{\big( \beta_{y_{2}}^{(\frac{3}{4})} + \frac{1}{4} c_{y_{2}}^{(2)} \beta_{y_{2}}^{(-\frac{1}{4})} \big) \mathcal{G}(\frac{\tau}{2}) + \frac{1}{2}(\varphi_{y_{2}}^{(\frac{1}{2})})^{2} \beta_{y_{2}}^{(-\frac{1}{4})} \mathcal{G}''(\frac{\tau}{2})}{\mathcal{G}(\frac{1}{2})} \begin{pmatrix}
1 & 0 \\
0 & 0
\end{pmatrix} \nonumber \\
& \hspace{-1cm} + \frac{\big(\frac{1}{4} c_{y_{2}}^{(2)} \varphi_{y_{2}}^{(\frac{1}{2})} \beta_{y_{2}}^{(-\frac{1}{4})} - \varphi_{y_{2}}^{(\frac{1}{2})} \beta_{y_{2}}^{(\frac{3}{4})} - \varphi_{y_{2}}^{(\frac{3}{2})} \beta_{y_{2}}^{(-\frac{1}{4})}\big) \widetilde{\mathcal{G}}(\frac{1+\tau}{2}) + \frac{1}{2}(\varphi_{y_{2}}^{(\frac{1}{2})})^{3} \beta_{y_{2}}^{(-\frac{1}{4})} \widetilde{\mathcal{G}}''(\frac{1+\tau}{2})}{(\varphi_{y_{2}}^{(\frac{1}{2})})^{2} (\beta_{y_{2}}^{(-\frac{1}{4})})^{2} \mathcal{G}(0)} \begin{pmatrix}
c_{\mathcal{F}} & 0 \\
i & 0
\end{pmatrix} \Bigg\}  \nonumber \\
& + \begin{pmatrix}
0 & \star  \\
0 & \star
\end{pmatrix}r^{-\frac{3}{4}}, \nonumber
\end{align}
as $r\to+\infty$, where $\star$ denotes unnecessary constants.

\subsection{Parametrix at $y_3=0$}

Let $f_{y_3}$ be the conformal map from $\mathbb{D}_{y_3}$ to a neighborhood of $0$ defined by
\begin{equation}\label{0coord}
f_{y_{3}}(z)=\frac{1}{4}\left(g(z)\mp\frac{i\Omega}{2}\right)^2,
\end{equation}
where we take the $-/+$ sign when $z$ is above/below the real axis. The expansion of $f_{y_{3}}(z)$ as $z \to 0$ is given by
\begin{align}
& f_{y_{3}}(z) = c_{y_{3}}z\big( 1+ c_{y_{3}}^{(2)}z+ c_{y_{3}}^{(3)}z^{2} + \bigO(z^{3}) \big), \label{xiy3_expansion} \\
& c_{y_{3}} = \frac{p^2(0)}{y_1y_2}, ~~~ c_{y_{3}}^{(2)} = \frac{2p'(0)}{3p(0)}+\frac{y_1+y_2}{3y_1y_2}, \nonumber \\
& c_{y_{3}}^{(3)} = \frac{1}{45}\left( \frac{8y_{1}^{2}+7y_{1}y_{2}+8y_{2}^{2}}{y_{1}^{2}y_{2}^{2}} + \frac{5p'(0)^{2}}{p(0)^{2}} + \frac{14(y_{1}+y_{2})p'(0)}{y_{1}y_{2}p(0)}+\frac{9p''(0)}{p(0)} \right)  . \nonumber
\end{align}
We choose the lenses $\gamma_{+} \subset N_{0}$ and $\gamma_{-}\subset N_{0}$ such that $f_{y_{3}}(\gamma_{\pm} \cap \mathbb{D}_{y_{3}}) \subset e^{\pm\frac{2\pi i}{3}}\mathbb{R}^{+}$, and we define 
\begin{align}
& P^{(y_{3})}(z) = E_{y_{3}}(z)  \Phi_{\mathrm{Be}}(r^{3} f_{y_{3}}(z)) e^{-r^\frac{3}{2}g(z)\sigma_{3}}, \label{P0} \\
& E_{y_{3}}(z) = P^{(\infty)}(z) e^{\pm \frac{i \Omega}{2} r^\frac{3}{2}\sigma_{3}} M^{-1} \Big( 2\pi r^{\frac{3}{2}}(f_{y_{3}}(z))^{\frac{1}{2}} \Big)^{\frac{\sigma_{3}}{2}}. \nonumber
\end{align}
It can be verified that $E_{y_{3}}(z)$ is analytic inside the disk $\mathbb{D}_{y_{3}}$, $P^{(y_3)}(z)=\Or(\log|z|)$ as $z\to 0$, see \eqref{local behavior near 0 of P_Be}, and that 
\begin{align}\label{P0asymp}
P^{(y_3)}(z)P^{(\infty)}(z)^{-1}= I + \frac{P^{(\infty)}(z)e^{\pm \frac{i \Omega}{2}r^\frac{3}{2}\sigma_{3}}
 \Phi_{\mathrm{Be},1} e^{\mp \frac{i\Omega}{2}r^\frac{3}{2}\sigma_3}P^{(\infty)}(z)^{-1} }{r^\frac{3}{2}f_{y_{3}}(z)^{\frac{1}{2}}} + \bigO(r^{-3})
\end{align}
as $r\to +\infty$ uniformly for $z\in\partial\mathbb{D}_{y_3}$. Furthermore, we verify that
\begin{align}
& E_{y_{3}}(y_{3}) = - \sqrt{2\pi} c_{y_{3}}^{\frac{1}{4}}e^{\pi i \nu}r^{\frac{3}{4}} \Bigg\{ \frac{\mathcal{G}(\frac{\tau}{2})}{\beta_{y_{3}}^{(\frac{1}{4})} \mathcal{G}(0)} \begin{pmatrix}
c_{\mathcal{F}} & 0 \\ i & 0
\end{pmatrix} + \frac{\beta_{y_{3}}^{(\frac{1}{4})} \widetilde{\mathcal{G}}(\frac{1+\tau}{2})}{\varphi_{y_{3}}^{(\frac{1}{2})} \mathcal{G}(\frac{1}{2})} \begin{pmatrix}
1 & 0 \\ 0 & 0
\end{pmatrix} \Bigg\} + \begin{pmatrix}
0 & \star  \\
0 & \star
\end{pmatrix}r^{-\frac{3}{4}}, \label{Ey3y3} \\
& E_{y_{3}}'(y_{3}) = - \sqrt{2\pi} c_{y_{3}}^{\frac{1}{4}}e^{\pi i \nu} r^{\frac{3}{4}} \Bigg\{ \frac{\Big( \frac{1}{4}c_{y_{3}}^{(2)} - \frac{\beta_{y_{3}}^{(\frac{5}{4})}}{\beta_{y_{3}}^{(\frac{1}{4})}} \Big) \mathcal{G}(\frac{\tau}{2}) + \frac{1}{2}(\varphi_{y_{3}}^{(\frac{1}{2})})^{2} \mathcal{G}''(\frac{\tau}{2})}{\beta_{y_{3}}^{(\frac{1}{4})} \mathcal{G}(0)} \begin{pmatrix}
c_{\mathcal{F}} & 0 \\
i & 0
\end{pmatrix} \nonumber \\
& + \frac{\big( \frac{1}{4}c_{y_{3}}^{(2)} \beta_{y_{3}}^{(\frac{1}{4})} - \frac{\varphi_{y_{3}}^{(\frac{3}{2})}}{\varphi_{y_{3}}^{(\frac{1}{2})}} \beta_{y_{3}}^{(\frac{1}{4})} + \beta_{y_{3}}^{(\frac{5}{4})} \big) \widetilde{\mathcal{G}}(\frac{1+\tau}{2}) + \frac{1}{2} (\varphi_{y_{3}}^{(\frac{1}{2})})^{2} \beta_{y_{3}}^{(\frac{1}{4})} \widetilde{\mathcal{G}}''(\frac{1+\tau}{2})}{\varphi_{y_{3}}^{(\frac{1}{2})} \mathcal{G}(\frac{1}{2})} \Bigg\} + \begin{pmatrix}
0 & \star  \\
0 & \star
\end{pmatrix}r^{-\frac{3}{4}} \nonumber
\end{align}
as $r\to+\infty$, where $\star$ denotes unnecessary constants.

\section{Small norm problem}\label{section:smallnorm}
In this section we show that for sufficiently large $r$, $P^{(\infty)}(z)$ approximates $S(z)$ for $z \in \mathbb{C}\setminus \cup_{j=1}^3\mathbb{D}_{y_j}$ and $P^{(y_j)}(z)$ approximates $S(z)$ for $z \in \mathbb{D}_{y_j}$, $j=1,2,3$. We define
\begin{equation}\label{errorMatrix}
R(z)=\begin{cases}
        S(z)P^{(y_j)}(z)^{-1}, &z\in\mathbb{D}_{y_j}, ~ j=1,2,3, \\
        S(z)P^{(\infty)}(z)^{-1}, &z\in\mathbb{C}\setminus\bigcup_{j=1}^3\mathbb{D}_{y_j}.
    \end{cases}
\end{equation}
Since $P^{(y_j)}(z)$, $j=1,2,3$, has the exact same jumps as $S$ inside the disks, $R(z)$ is analytic for $z\in \cup_{j=1}^{3}\mathbb{D}_{y_j}\setminus\{y_{j}\}$. Furthermore, we verify from \eqref{S asymp at yj}, \eqref{Py1}, \eqref{Py2}, \eqref{P0} and \eqref{local behavior near 0 of P_Be} that $S(z)P^{(y_j)}(z)^{-1} = \bigO(\log(z-y_{j}))$ as $z \to y_{j}$, $j=1,2,3$. This means that the singularities of $R$ at $y_{1},y_{2},y_{3}$ are in fact removable and $R$ is analytic in all three open disks. Let $\Sigma_{R}$ denote the jump contour for $R$ which is explicitly given by
\begin{align*}
\Sigma_{R} = \bigg(\gamma_{+} \cup \gamma_{-}  \cup \bigcup_{j=1}^{3} \partial \mathbb{D}_{y_{j}} \bigg)  \setminus \bigcup_{j=1}^{3}  \mathbb{D}_{y_{j}},
\end{align*}
with orientation as shown in Figure \ref{fig:SigmaR}. In particular, note that we orient the boundaries of the disks in the clockwise direction.
\begin{figure}
\centering
\begin{tikzpicture}
%dot at 0,y_2,y_1, respectively
\draw[fill] (0,0) circle (0.05);
\draw[fill] (3,0) circle (0.05);
\draw[fill] (5,0) circle (0.05);

%upper, lower lense, respectively from negative infinity to parametrix boundary
\draw (-0.3,0.52) -- (120:2.4);
\draw (-0.3,-0.52) -- (-120:2.4);

%boundary of parametrix at 0, y_2, y_1, respectively
\draw (0,0) circle (0.6cm);
\draw (3,0) circle (0.6cm);
\draw (5,0) circle (0.6cm);

%upper,lower lense, respectively, from y_2 to y_1
\draw (3.36,0.48) .. controls (3.65,0.85) and (4.35,0.85) .. (4.64,0.48);
\draw (3.36,-0.48) .. controls (3.65,-0.85) and (4.35,-0.85) .. (4.64,-0.48);

%label 0, y_2, y_1, -\infty, respectively
\node at (0.15,-0.3) {$0$};
\node at (3.15,-0.3) {$y_{2}$};
\node at (5.2,-0.3) {$y_{1}$};

%arrow for lower/upper lense, respectively, at 0
\draw[black,arrows={-Triangle[length=0.18cm,width=0.12cm]}]
(-120:1.5) --  ++(60:0.001);
\draw[black,arrows={-Triangle[length=0.18cm,width=0.12cm]}]
(120:1.3) --  ++(-60:0.001);

%arrow for parametrix boundary at 0,y_2,y_1, respectively
\draw[black,arrows={-Triangle[length=0.18cm,width=0.12cm]}]
(0.6,-0.09) --  ++(-90:0.001);
\draw[black,arrows={-Triangle[length=0.18cm,width=0.12cm]}]
(3.6,-0.09) --  ++(-90:0.001);
\draw[black,arrows={-Triangle[length=0.18cm,width=0.12cm]}]
(5.6,-0.09) --  ++(-90:0.001);

%arrow for upper, lower lense, respectively, from y_2 to y_1
\draw[black,arrows={-Triangle[length=0.18cm,width=0.12cm]}]
(4.08,0.75) --  ++(0:0.001);
\draw[black,arrows={-Triangle[length=0.18cm,width=0.12cm]}]
(4.08,-0.75) --  ++(0:0.001);

\end{tikzpicture}
\caption{The jump contour $\Sigma_R$.}
\label{fig:SigmaR}
\end{figure}
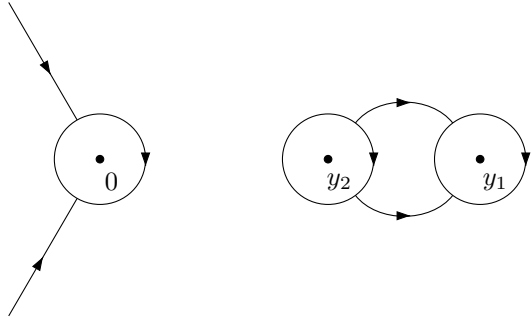
The jumps for $R$ are denoted by $J_{R}$,
\begin{align*}
J_{R} : \Sigma_{R} \to \mathbb{C}^{2 \times 2}, \quad z \mapsto J_{R}(z):= R_{-}(z)^{-1}R_{+}(z),
\end{align*}
and are explicitly given by
\begin{align}\label{Ejump}
J_{R}(z)=\begin{cases}
P^{(\infty)}(z)\begin{pmatrix} 1 & 0 \\ e^{-2r^\frac{3}{2}g(z)} & 1 \end{pmatrix}P^{(\infty)}(z)^{-1}, &z\in\gamma_+\cup\gamma_-\setminus\bigcup_{j=1}^3\mathbb{D}_{y_j} \\
        P^{(y_j)}(z)P^{(\infty)}(z)^{-1}, &z\in\partial\mathbb{D}_{y_j}, ~ j=1,2,3.
    \end{cases}
\end{align}
From Lemma \ref{g_lemma} and \eqref{Py1asymp}, \eqref{Py2asymp}, \eqref{P0asymp}, as $r\to\infty$ we have 
\begin{equation}\label{JR expansion}
J_{R}(z) = \begin{cases}
I + \bigO(e^{-\tilde{c} |rz|^\frac{3}{2}}), & \mbox{uniformly for } z \in \Sigma_{R}\setminus \cup_{j=1}^{3} \partial \mathbb{D}_{y_{j}}, \\
I + \frac{J_{R}^{(1)}(z)}{r^\frac{3}{2}} + \bigO(r^{-3}), & \mbox{uniformly for } z \in \cup_{j=1}^{3} \partial \mathbb{D}_{y_{j}},
\end{cases}
\end{equation}
where $\tilde{c}>0$ is a sufficiently small constant, and $J_{R}^{(1)}(z) = \bigO(1)$ as $r \to + \infty$ uniformly for $z \in \cup_{j=1}^{3} \partial \mathbb{D}_{y_{j}}$.  It is important to recall that $P^{(\infty)}(z)$ has $r$ dependence but this is of no consequence due to the periodicity of the $\theta$-function in the real direction.  The matrix $J_{R}^{(1)}(z)$ has been computed in \eqref{Py1asymp}, \eqref{Py2asymp}, \eqref{P0asymp}, and is given by
\begin{align}
J_{R}^{(1)}(z) = \begin{cases}
\frac{P^{(\infty)}(z)\Phi_{\mathrm{Be,1}} P^{(\infty)}(z)^{-1}}{f_{y_{1}}(z)^{\frac{1}{2}}}, & \mbox{if } z  \in \partial \mathbb{D}_{y_{1}}, \\[0.2cm]
\frac{P^{(\infty)}(z)e^{\pm \frac{i \Omega}{2}r^\frac{3}{2}\sigma_{3}}\sigma_{3} \Phi_{\mathrm{Be},1}\sigma_{3} e^{\mp \frac{i\Omega}{2}r^\frac{3}{2}\sigma_3}P^{(\infty)}(z)^{-1} }{(-f_{y_{2}}(z))^{\frac{1}{2}}}, & \mbox{if } z  \in \partial \mathbb{D}_{y_{2}}, \\[0.2cm]
\frac{P^{(\infty)}(z)e^{\pm \frac{i \Omega}{2}r^\frac{3}{2}\sigma_{3}}
 \Phi_{\mathrm{Be},1} e^{\mp \frac{i\Omega}{2}r^\frac{3}{2}\sigma_3}P^{(\infty)}(z)^{-1} }{f_{y_{3}}(z)^{\frac{1}{2}}}, & \mbox{if } z  \in \partial \mathbb{D}_{y_{3}}.
\end{cases} \label{explicit expression for the jumps Jrp1p}
\end{align}
By the standard theory for RH problems \cite{DeiftZhou}, $R(z)$ exists for sufficiently large $r$ and satisfies
\begin{equation}\label{large r asymptotics for R}
R(z)=I+\frac{R^{(1)}(z)}{r^\frac{3}{2}}+\bigO(r^{-3}),
\end{equation}
as $r\to\infty$, uniformly for $z\in\mathbb{C}\setminus\Sigma_{R}$. The goal for the remainder of this section is to obtain explicit expressions for $R^{(1)\prime}(y_{j})$, $j=1,2,3$. Since $R(z)$ satisfies the equation
\begin{equation}\label{integral formula for R}
R(z)=I+\frac{1}{2\pi i}\int_{\Sigma_{R}}\frac{R_-(\xi)\left(J_{R}(\xi)-I\right)}{\xi-z}~d\xi, \qquad z \in \mathbb{C}\setminus \Sigma_{R},
\end{equation}
we deduce from \eqref{JR expansion} that
\begin{equation} \label{expression for Rp1p as integral}
R^{(1)}(z)=\frac{1}{2\pi i}\int_{\bigcup_{j=1}^3\partial\mathbb{D}_{y_j}}\frac{J_R^{(1)}(\xi)}{\xi-z}~d\xi, \qquad z \in \mathbb{C}\setminus \bigcup_{j=1}^3\partial\mathbb{D}_{y_j},
\end{equation}
where we recall that the circles $\partial\mathbb{D}_{y_j}$, $j=1,2,3,$ have clockwise orientation. From \eqref{explicit expression for the jumps Jrp1p}, we note that the jumps $J_{R}^{(1)}$ can be analytically continued from $\cup_{j=1}^{3} \partial \mathbb{D}_{y_{j}}$ to $\cup_{j=1}^{3} \big(\overline{\mathbb{D}}_{y_{j}} \setminus \{y_{j}\}\big)$, so we can evaluate \eqref{expression for Rp1p as integral} by a residue calculation. The asymptotic expansion of $J_{R}^{(1)}(z)$ as $z \to y_{j}$, $j=1,2,3$, is of the form
\begin{align*}
J_{R}^{(1)}(z) = \sum_{k=-1}^{N} (J_{R}^{(1)})_{y_{j}}^{(k)} (z-y_{j})^{k} + \bigO((z-y_{j})^{N+1}),
\end{align*}
for any $N \in \mathbb{N}_{>0}$, and therefore \eqref{expression for Rp1p as integral} can be rewritten as
\begin{align}
& R^{(1)}(z)=\sum_{j=1}^3\frac{1}{z-y_j} (J_{R}^{(1)})_{y_{j}}^{(-1)}, & & z\in\mathbb{C}\setminus\cup_{j=1}^3\mathbb{D}_{y_j}, \\
& R^{(1)}(z)=-\sum_{k=0}^{\infty}(J_{R}^{(1)})_{y_{j}}^{(k)}(z-y_{j})^{k}  +  \sum_{\substack{j'=1 \\ j'\neq j}}^3\frac{1}{z-y_{j'}}(J_{R}^{(1)})_{y_{j'}}^{(-1)}, & & z \in \mathbb{D}_{y_{j}}, j=1,2,3.
\end{align}
In particular, we have
\begin{align}
& R^{(1)\prime}(y_{j})= -(J_{R}^{(1)})_{y_{j}}^{(1)} - \sum_{\substack{j'=1 \\ j'\neq j}}^3\frac{1}{(y_{j}-y_{j'})^{2}}(J_{R}^{(1)})_{y_{j'}}^{(-1)}, & & j=1,2,3. \label{der of Rp1p at yj}
\end{align}
An explicit expression for $R^{(1)\prime}(y_{j})$ requires the evaluation of the matrix-coefficients appearing in \eqref{der of Rp1p at yj}. The coefficients $(J_{R}^{(1)})_{y_{j}}^{(1)}$ are particularly hard to compute and require more coefficients in the expansions of $P^{(\infty)}(z)$ as $z \to y_{j}$ than those already computed in Subsections \ref{section: expansion Pinf at y1}, \ref{section: expansion Pinf at y2} and \ref{section: expansion Pinf at y3}. 
%For example, we see from the following formula that the coefficient $(P^{(\infty)})_{y_{1}}^{(\frac{7}{4})}$ is needed to evaluate $(J_{R}^{(1)})_{y_{1}}^{(1)}$:
%\begin{align*}
%(J_{R}^{(1)})_{y_{1}}^{(1)} = \frac{1}{\sqrt{c_{y_{1}}}} \bigg\{ \frac{3 (c_{y_{1}}^{(2)})^{2} - 4 c_{y_{1}}^{(3)}}{8 } (P^{(\infty)})_{y_{1}}^{(-\frac{1}{4})} \Phi_{\mathrm{Be},1} (P^{(\infty)}_{\mathrm{inv}})_{y_{1}}^{(-\frac{1}{4})} - \frac{c_{y_{1}}^{(2)}}{2} (P^{(\infty)})_{y_{1}}^{(-\frac{1}{4})} \Phi_{\mathrm{Be},1} (P^{(\infty)}_{\mathrm{inv}})_{y_{1}}^{(\frac{3}{4})} \\
%+ (P^{(\infty)})_{y_{1}}^{(-\frac{1}{4})} \Phi_{\mathrm{Be},1} (P^{(\infty)}_{\mathrm{inv}})_{y_{1}}^{(\frac{7}{4})} - \frac{c_{y_{1}}^{(2)}}{2} (P^{(\infty)})_{y_{1}}^{(\frac{1}{4})} \Phi_{\mathrm{Be},1} (P^{(\infty)}_{\mathrm{inv}})_{y_{1}}^{(\frac{1}{4})} + (P^{(\infty)})_{y_{1}}^{(\frac{1}{4})} \Phi_{\mathrm{Be},1} (P^{(\infty)}_{\mathrm{inv}})_{y_{1}}^{(\frac{5}{4})} \\
%+ (P^{(\infty)})_{y_{1}}^{(\frac{3}{4})} \Phi_{\mathrm{Be},1} (P^{(\infty)}_{\mathrm{inv}})_{y_{1}}^{(\frac{3}{4})} - \frac{c_{y_{1}}^{(2)}}{2}(P^{(\infty)})_{y_{1}}^{(\frac{3}{4})} \Phi_{\mathrm{Be},1} (P^{(\infty)}_{\mathrm{inv}})_{y_{1}}^{(-\frac{1}{4})} + (P^{(\infty)})_{y_{1}}^{(\frac{5}{4})} \Phi_{\mathrm{Be},1} (P^{(\infty)}_{\mathrm{inv}})_{y_{1}}^{(\frac{1}{4})} \\
%+ (P^{(\infty)})_{y_{1}}^{(\frac{7}{4})} \Phi_{\mathrm{Be},1} (P^{(\infty)}_{\mathrm{inv}})_{y_{1}}^{(-\frac{1}{4})} \bigg\}.
%\end{align*}
%The formulas for $(J_{R}^{(1)})_{y_{j}}^{(1)}$, $j=2,3$ are similar. 
However, it turns out that what will be needed in Section \ref{section:integration1} is to obtain explicit expressions for
\begin{align}\label{sandwich j}
\Big[E_{y_j}(y_j)^{-1} R^{(1)\prime}(y_j)E_{y_j}(y_j)\Big]_{21}.
\end{align}
The quantities \eqref{sandwich j} are much easier to evaluate than the $R^{(1)\prime}(y_j)$ themselves, due to heavy cancellations between $R^{(1)\prime}(y_j)$ and the factors $E_{y_j}(y_j)^{\pm 1}$. In particular, \eqref{sandwich j} can be explicitly written in terms of the first four coefficients appearing in the expansions of $P^{(\infty)}(z)$ as $z \to y_{j}$ (these coefficients were computed in Subsections \ref{section: expansion Pinf at y1}, \ref{section: expansion Pinf at y2} and \ref{section: expansion Pinf at y3}). Let us define
\begin{align*}
\mathcal{J}_{y_{j}}(z)=\Big[E_{y_j}(y_j)^{-1}  J_{R}^{(1)}(z)  E_{y_j}(y_j)\Big]_{21}.
\end{align*}
As $z \to y_{j'}$, $j'=1,2,3$, we have
\begin{align}
& \mathcal{J}_{y_{j}}(z) = \sum_{k=-1}^{1} (\mathcal{J}_{y_{j}})_{y_{j'}}^{(k)}(z-y_{j'})^{k} + \bigO((z-y_{j'})^{2}), \nonumber \\
& (\mathcal{J}_{y_{j}})_{y_{j'}}^{(k)} = \Big[E_{y_j}(y_j)^{-1}  (J_{R}^{(1)})_{y_{j'}}^{(k)}  E_{y_j}(y_j)\Big]_{21}, \label{Jcal coeff def}
\end{align}
so that \eqref{sandwich j} can be written as
\begin{align*}
\Big[E_{y_j}(y_j)^{-1} R^{(1)\prime}(y_j)E_{y_j}(y_j)\Big]_{21} = -(\mathcal{J}_{y_{j}})_{y_{j}}^{(1)} - \sum_{\substack{j'=1 \\ j'\neq j}}^3\frac{1}{(y_{j}-y_{j'})^{2}}(\mathcal{J}_{y_{j}})_{y_{j'}}^{(-1)}.
\end{align*}
After a standard computation using the expansions derived in Subsections \ref{section: expansion Pinf at y1}, \ref{section: expansion Pinf at y2} and \ref{section: expansion Pinf at y3}, and using the expressions for $E_{y_{j}}(y_{j})$ given in \eqref{Ey1y1}, \eqref{Ey2y2} and \eqref{Ey3y3}, we obtain
\begin{align*}
& (\mathcal{J}_{y_{1}})_{y_{2}}^{(-1)} = - \frac{\sqrt{c_{y_{1}}}e^{2\pi i \nu} \pi (\beta_{y_{2}}^{(-\frac{1}{4})})^{2} \mathcal{G}(\frac{\tau}{2})^{2}}{8 \sqrt{c_{y_{2}}} (\beta_{y_{1}}^{(\frac{1}{4})})^{2} \mathcal{G}(0)^{2}}r^\frac{3}{2}, \qquad (\mathcal{J}_{y_{1}})_{y_{3}}^{(-1)} = - \frac{i\sqrt{c_{y_{1}}}e^{2\pi i \nu} \pi (\beta_{y_{3}}^{(\frac{1}{4})})^{2} \widetilde{\mathcal{G}}(\frac{1+\tau}{2})^{2}}{8 \sqrt{c_{y_{3}}} (\varphi_{y_{3}}^{(\frac{1}{2})})^{2}(\beta_{y_{1}}^{(\frac{1}{4})})^{2} \mathcal{G}(0)^{2}}r^\frac{3}{2}, \\
& (\mathcal{J}_{y_{2}})_{y_{1}}^{(-1)} = - \frac{\sqrt{c_{y_{2}}}e^{2\pi i \nu} \pi (\beta_{y_{2}}^{(-\frac{1}{4})})^{2} \mathcal{G}(\frac{\tau}{2})^{2}}{8 \sqrt{c_{y_{1}}} (\beta_{y_{1}}^{(\frac{1}{4})})^{2} \mathcal{G}(0)^{2}}r^\frac{3}{2}, \qquad (\mathcal{J}_{y_{3}})_{y_{1}}^{(-1)} = - \frac{i\sqrt{c_{y_{3}}}e^{2\pi i \nu}\pi (\beta_{y_{3}}^{(\frac{1}{4})})^{2} \widetilde{\mathcal{G}}\big( \frac{1+\tau}{2} \big)^{2}}{8\sqrt{c_{y_{1}}} (\varphi_{y_{3}}^{(\frac{1}{2})})^{2} (\beta_{y_{1}}^{(\frac{1}{4})})^{2} \mathcal{G}(0)^{2}}r^\frac{3}{2}, \\
& (\mathcal{J}_{y_{2}})_{y_{3}}^{(-1)} = - \frac{\sqrt{c_{y_{2}}}e^{4\pi i \nu} \pi \Big( \varphi_{y_{2}}^{(\frac{1}{2})} \varphi_{y_{3}}^{(\frac{1}{2})} (\beta_{y_{2}}^{(-\frac{1}{4})})^{2} \mathcal{G}(\frac{\tau}{2})^{2} - (\beta_{y_{3}}^{(\frac{1}{4})})^{2} \widetilde{\mathcal{G}} \big( \frac{1+\tau}{2} \big)^{2} \Big)^{2}}{8\sqrt{c_{y_{3}}}(\varphi_{y_{2}}^{(\frac{1}{2})})^{2}(\varphi_{y_{3}}^{(\frac{1}{2})})^{2} (\beta_{y_{2}}^{(-\frac{1}{4})})^{2}(\beta_{y_{3}}^{(\frac{1}{4})})^{2} \mathcal{G}(0)^{2}\mathcal{G}(\frac{1}{2})^{2}}r^\frac{3}{2}, \\
& (\mathcal{J}_{y_{3}})_{y_{2}}^{(-1)} = - \frac{\sqrt{c_{y_{3}}}e^{4\pi i \nu} \pi \Big( \varphi_{y_{2}}^{(\frac{1}{2})} \varphi_{y_{3}}^{(\frac{1}{2})} (\beta_{y_{2}}^{(-\frac{1}{4})})^{2} \mathcal{G}(\frac{\tau}{2})^{2} - (\beta_{y_{3}}^{(\frac{1}{4})})^{2} \widetilde{\mathcal{G}} \big( \frac{1+\tau}{2} \big)^{2} \Big)^{2}}{8\sqrt{c_{y_{2}}}(\varphi_{y_{2}}^{(\frac{1}{2})})^{2}(\varphi_{y_{3}}^{(\frac{1}{2})})^{2} (\beta_{y_{2}}^{(-\frac{1}{4})})^{2}(\beta_{y_{3}}^{(\frac{1}{4})})^{2} \mathcal{G}(0)^{2}\mathcal{G}(\frac{1}{2})^{2}}r^\frac{3}{2}.
\end{align*}
The coefficient $(\mathcal{J}_{y_{1}})_{y_{1}}^{(1)}$ requires more involved computations and is given by
\begin{align*}
(\mathcal{J}_{y_{1}})_{y_{1}}^{(1)} = \frac{i\pi r^\frac{3}{2}}{16 \beta_{y_{1}}^{(\frac{1}{4})}\mathcal{G}(0)^{2}} \bigg\{ 3 c_{y_{1}}^{(2)} \beta_{y_{1}}^{(\frac{1}{4})} \mathcal{G}(0)^{2} - 2 \Big( 6 \beta_{y_{1}}^{(\frac{5}{4})} \mathcal{G}(0)^{2} + (\varphi_{y_{1}}^{(\frac{1}{2})})^{2} \beta_{y_{1}}^{(\frac{1}{4})} \Big[ \mathcal{G}'(0)^{2} + 3 \mathcal{G}(0) \mathcal{G}''(0) \Big] \Big) \bigg\}.
\end{align*}
The expressions for $(\mathcal{J}_{y_{2}})_{y_{2}}^{(1)}$ and $(\mathcal{J}_{y_{3}})_{y_{3}}^{(1)}$ are more complicated and signicantly longer than $(\mathcal{J}_{y_{1}})_{y_{1}}^{(1)}$, so we do not write them down. In Section \ref{Section: theta identities}, we derive several identities involving the $\theta$-function which allows one to simplify the formulas for $(\mathcal{J}_{y_{j}})_{y_{j'}}^{(k)}$, see proof of Proposition \ref{prop: log r term second term}.

\section{Some $\theta$-function identities}\label{Section: theta identities}
The coefficients $(P^{(\infty)})_{y_{j}}^{(-\frac{1}{4}+\frac{k}{2})}$ obtained in Subsections \ref{section: expansion Pinf at y1}, \ref{section: expansion Pinf at y2}, \ref{section: expansion Pinf at y3}, the expressions for $E_{y_{j}}(y_{j})$ and $E'_{y_{j}}(y_{j})$ given by \eqref{Ey1y1}, \eqref{Ey2y2}, \eqref{Ey3y3}, and the matrices $(\mathcal{J}_{y_{j}})_{y_{j'}}^{(k)}$ computed in Section \ref{section:smallnorm}, all involve the $\theta$-function and its derivatives evaluated at the $8$ points
\begin{align*}
0, \qquad \frac{1}{2}, \qquad \frac{\tau}{2}, \qquad \frac{1+\tau}{2}, \qquad \nu, \qquad \nu+\frac{1}{2}, \qquad \nu+\frac{\tau}{2}, \qquad \nu+\frac{1+\tau}{2}.
\end{align*}
In this section, we provide a systematic way of simplifying these expressions by showing that
\begin{align}\label{get rid of theta 1}
\theta^{(j)}\Big( \frac{1}{2} \Big), \qquad \theta^{(j)}\Big( \frac{\tau}{2} \Big), \qquad \theta^{(j)}\Big( \frac{1+\tau}{2} \Big), \qquad j \geq 0,
\end{align}
where $\theta^{(j)}(z) := \frac{d^j}{dz^j}\theta(z)$, can be expressed in terms of $\theta(0)$, $\theta^{(2)}(0)$, $\theta^{(4)}(0)$,..., and similarly that 
\begin{align}\label{get rid of theta 2}
\theta^{(j)}\Big( \nu + \frac{1}{2} \Big), \qquad \theta^{(j)}\Big( \nu +\frac{\tau}{2} \Big), \qquad \theta^{(j)}\Big( \nu +\frac{1+\tau}{2} \Big), \qquad j \geq 0,
\end{align}
can be expressed in terms of $\theta(\nu)$, $\theta'(\nu)$, $\theta^{(2)}(\nu)$,$\cdots$. We start by recalling the basic symmetries of the $\theta$-function, which can be found in \eqref{periodicity property of theta function}. By differentiating the relations $\theta(-u) = \theta(u)$ and $\theta(\frac{1}{2}-u) = \theta(\frac{1}{2}+u)$, we immediately obtain 
\begin{align*}
\theta'\left(0\right)= 0, \qquad \theta^{(3)}\left(0\right) = 0, \qquad \theta'\Big(\frac{1}{2}\Big)= 0, \qquad \theta^{(3)}\Big(\frac{1}{2}\Big) = 0.
\end{align*}
We also recall that $\theta(z)$ has a simple zero at $z = \frac{1 + \tau}{2}$. Let us define $\widetilde{\theta}$ by
\begin{align}\label{def of theta tilde}
\widetilde{\theta}(z) = \frac{\theta(z)}{z - \frac{1 + \tau}{2}},
\end{align}
so that $\widetilde{\mathcal{G}}$ defined in \eqref{Gtilde} is given by $\widetilde{\mathcal{G}}(z)=\frac{\theta(z+\nu)}{\widetilde{\theta}(z)}$. By differentiating \eqref{def of theta tilde}, we obtain
\begin{align}
& \widetilde{\theta}^{(j-1)}\Big(\frac{1+\tau}{2}\Big) = \frac{1}{j} \theta^{(j)}\Big(\frac{1+\tau}{2}\Big), \qquad j=1,2,\cdots. \label{theta tilde to theta}
\end{align}

\begin{remark}
It can be shown by standard arguments that $\det P^{(\infty)}(z) \equiv 1$. This gives a non-trivial relation between the $\theta$-function evaluated at several points. However, it turns out that this is not enough for our needs. Proposition \ref{prop: Abel map and theta function} below states three fundamental relations between the Abel map and the $\theta$-function. These identities are of central importance for us and will be used extensively in the remainder of this paper.
\end{remark}
\begin{proposition}\label{prop: Abel map and theta function}
Let $z \mapsto \varphi_{A}(z)$ denote the Abel map \eqref{def of Abel}. Then, for all $z$ on the Riemann surface, 
\begin{subequations}\label{thetaD123}
\begin{align}\label{thetaD123a}
& \frac{e^{2i \pi \varphi_{A}(z)} \theta (\varphi_{A}(z) +\frac{\tau}{2})^2}{\theta (\varphi_{A}(z))^2}
= \frac{D_1^2}{z},
	\\\label{thetaD123b}
& \frac{e^{2i \pi \varphi_{A}(z)} \theta(\varphi_{A}(z) +\frac{1+\tau}{2})^2}{\theta (\varphi_{A}(z))^2}
= \frac{D_2^2(z-y_{1})}{z},
	\\\label{thetaD123c}
& \frac{\theta(\varphi_{A}(z) +\frac{1}{2})^2}{\theta (\varphi_{A}(z))^2}
= \frac{D_3^2(z-y_2)}{z},
\end{align}
\end{subequations}
where we have identified $z \in X$ with its projection onto the complex plane on the right-hand sides, and
\begin{align}\label{D123def optimal}
& D_1 = (y_{1} y_{2})^{\frac{1}{4}}e^{- \frac{\pi i \tau}{4}}, & & D_2 = i \frac{y_{2}^{\frac{1}{4}}e^{-\frac{\pi i \tau}{4}}}{(y_{1}-y_{2})^{\frac{1}{4}}} , & & D_3 = \frac{y_{1}^{\frac{1}{4}}}{(y_{1}-y_{2})^{\frac{1}{4}}}.
\end{align}
\end{proposition}
\begin{proof}
An easy computation using the transformation properties (\ref{periodicity property of theta function}) of $\theta$ show that the three quotients 
$$\frac{e^{2\pi i u} \theta(u +\frac{\tau}{2})^2}{\theta(u)^2}, \qquad
\frac{e^{2i \pi u} \theta(u +\frac{1+\tau}{2})^2}{\theta (u)^2}, \qquad
\frac{\theta(u +\frac{1}{2})^2}{\theta (u)^2}$$
are invariant under the shifts $u \mapsto u + 1$ and $u\mapsto u+ \tau$, so that they are well-defined on the torus $\C / (\mathbb{Z} + \tau \mathbb{Z})$. Let us first prove (\ref{thetaD123a}). Since $\theta(u)$ has a simple zero at $u = \frac{1+\tau}{2}$ and no other zeros, it follows that the quotient $\frac{e^{2\pi i u} \theta(u +\frac{\tau}{2})^2}{\theta(u)^2}$ defines a meromorphic function on $\C / (\mathbb{Z} + \tau \mathbb{Z})$ with a double zero at $u = \frac{1}{2}$ and a double pole at $u = \frac{1+\tau}{2}$ and no other zeros or poles. 
Recalling that
$$\varphi_{A}(y_1) = 0, \qquad \varphi_{A}(y_2) = \frac{\tau}{2}, \qquad \varphi_{A}(y_3) = \frac{1}{2} + \frac{\tau}{2}, \qquad \varphi_{A}(\infty) = \frac{1}{2},$$
this implies that the left-hand side of (\ref{thetaD123a}) is a meromorphic  function on the Riemann surface with a double zero at $z=\infty$ and a double pole at $z = y_3$. Since the torus is of genus $1$, it follows from the Riemann-Roch theorem that the identity (\ref{thetaD123a}) holds for some constant $D_1^2 \in \C$. 
Analogous arguments apply to (\ref{thetaD123b}) and (\ref{thetaD123c}). By evaluating the identities in (\ref{thetaD123}) at $y_{1}$, $\infty$ and $\infty$, respectively, we obtain
\begin{align}\label{D123def}
& D_1 = \sqrt{y_1}\frac{ \theta(\frac{\tau}{2})}{\theta(0)} \in \R, & & D_2 = i\frac{\theta(\frac{\tau}{2})}{\theta(\frac{1}{2})} \in i\R, & & D_3 = \frac{\theta(0)}{ \theta(\frac{1}{2})} \in \R.
\end{align}
On the other hand, by expanding \eqref{thetaD123a} as $z \to y_{2}$ and \eqref{thetaD123b}, \eqref{thetaD123c} as $z \to y_{1}$, we obtain
\begin{align}\label{D123defa old}
& D_1 = \sqrt{y_{2}}\frac{e^{-\frac{\pi i \tau}{2}}\theta(0)}{\theta(\frac{\tau}{2})}, & & D_2 = \frac{\sqrt{y_1} \varphi_{y_1}^{(\frac{1}{2})} \theta'(\frac{\tau +1}{2})}{\theta(0)}, & & D_3 = \frac{\sqrt{y_1} \theta(\frac{1}{2})}{\sqrt{y_1-y_2} \theta(0)}.
\end{align}
Comparing \eqref{D123def} with \eqref{D123defa old}, we obtain \eqref{D123def optimal}.
\end{proof}
We immediately obtain the following Corollary.
\begin{corollary}\label{coro: abel map identity 1}
Let $z \mapsto \varphi_{A}(z)$ be the Abel map defined in \eqref{def of Abel}. Then
\begin{align}\label{abelthetaexplicit}
\frac{e^{2 i \pi \varphi_{A}(z)} \theta (\varphi_{A}(z) +\frac{\tau}{2}) \theta (\varphi_{A}(z) +\frac{1+\tau}{2})}{\theta (\varphi_{A}(z) ) \theta (\varphi_{A}(z) +\frac{1}{2})}
= \frac{D_{1} D_{2}}{D_{3}} \sqrt{\frac{z-y_1}{(z-y_2)z}}
\end{align}
for all $z$ on the Riemann surface, where the branch is such that $\sqrt{\frac{z-y_1}{(z-y_2)z}} > 0$ for $z > y_1$ on the upper sheet, and
\begin{align}\label{C0def}
\frac{D_{1} D_{2}}{D_{3}} = i \sqrt{y_{2}} e^{-\frac{\pi i \tau}{2}}.
\end{align}
\end{corollary}
\begin{proof}
By multiplying \eqref{thetaD123a} with \eqref{thetaD123b}, and then dividing by \eqref{thetaD123c}, we obtain 
\begin{align}\label{abelthetaexplicit square}
\bigg(\frac{e^{2 i \pi \varphi_{A}(z)} \theta (\varphi_{A}(z) +\frac{\tau}{2}) \theta (\varphi_{A}(z) +\frac{1+\tau}{2})}{\theta (\varphi_{A}(z) ) \theta (\varphi_{A}(z) +\frac{1}{2})}\bigg)^{2}
= \frac{D_{1}^{2} D_{2}^{2}}{D_{3}^{2}} \frac{z-y_1}{(z-y_2)z}.
\end{align}
By taking the square root of \eqref{abelthetaexplicit square}, we obtain \eqref{abelthetaexplicit}, up to a multiplicative sign that can be determined from an expansion as $z \to y_{1}$, see \cite[Chapter 20]{NIST}.
\end{proof}
By taking derivatives of \eqref{thetaD123a}-\eqref{thetaD123c}, we can obtain expressions for \eqref{get rid of theta 1} in terms of only $\theta(0),\theta''(0),\ldots$. We summarize the relations that will be used in Section \ref{section:integration1} in the following Corollary.
\begin{corollary}\label{coro:identities without nu}
We have the relations
\begin{align*}
& \theta\Big(\frac{\tau}{2}\Big) = e^{-\frac{\pi i \tau}{4}}\frac{y_{2}^{\frac{1}{4}}}{y_{1}^{\frac{1}{4}}}\theta(0), & & \theta\Big(\frac{1}{2}\Big) = \frac{(y_{1}-y_{2})^{\frac{1}{4}}}{y_{1}^{\frac{1}{4}}}\theta(0), & & \theta'\Big(\frac{\tau}{2}\Big) = -i \pi e^{-\frac{\pi i \tau}{4}}\frac{y_{2}^{\frac{1}{4}}}{y_{1}^{\frac{1}{4}}}\theta(0),
\end{align*}
as well as
\begin{align*}
& \theta'\Big(\frac{1+\tau}{2}\Big) = i e^{-\frac{\pi i \tau}{4}}\frac{(y_{1}-y_{2})^{\frac{1}{4}}y_{2}^{\frac{1}{4}}}{2c_{0}}\theta(0), \\
& \theta''\Big(\frac{\tau}{2}\Big) = e^{-\frac{\pi i \tau}{4}}\frac{y_{2}^{\frac{1}{4}}}{y_{1}^{\frac{1}{4}}}\bigg( \theta''(0)-\bigg[ \pi^{2} + \frac{y_{1}-y_{2}}{4 c_{0}^{2}} \bigg]\theta(0) \bigg), \\
& \theta''\Big(\frac{1+\tau}{2}\Big) = \pi e^{-\frac{\pi i \tau}{4}} \frac{(y_{1}-y_{2})^{\frac{1}{4}}y_{2}^{\frac{1}{4}}}{c_{0}}\theta(0), \\
& \theta''\Big(\frac{1}{2}\Big) = \frac{(y_{1}-y_{2})^{\frac{1}{4}}}{4 c_{0}^{2} y_{1}^{\frac{1}{4}}}\Big( y_{2} \theta(0) + 4 c_{0}^{2} \theta''(0) \Big), \\
& \theta^{(3)}\Big( \frac{1+\tau}{2}\Big) = -ie^{-\frac{\pi i \tau}{4}} \frac{(y_{1}-y_{2})^{\frac{1}{4}}y_{2}^{\frac{1}{4}}}{8c_{0}^{3}} \Big( (y_{1}-2y_{2})\theta(0) + 12c_{0}^{2} \big( \pi^{2}\theta(0)-\theta''(0) \big) \Big).
\end{align*}
%\begin{align*}
%& \theta^{(3)}\Big( \frac{\tau}{2} \Big) = \pi i e^{-\frac{\pi i  \tau}{4}} \frac{y_{2}^{\frac{1}{4}}}{y_{1}^{\frac{1}{4}}} \frac{3(y_{1}-y_{2})\theta(0)+4c_{0}^{2}(\pi^{2} \theta(0) - 3\theta''(0))}{4c_{0}^{2}}, \\
%& \theta^{(4)}\Big( \frac{1+\tau}{2} \Big) = - \pi e^{-\frac{\pi i  \tau}{4}} (y_{1}-y_{2})^{\frac{1}{4}} y_{2}^{\frac{1}{4}} \frac{(y_{1}-2y_{2})\theta(0)+4c_{0}^{2}(\pi^{2} \theta(0)-3\theta''(0))}{2c_{0}^{3}}
%\end{align*}
\end{corollary}
Finally, we note that the relations \eqref{thetaD123a}-\eqref{thetaD123c} allows also to express \eqref{get rid of theta 2} in terms of $\theta(\nu)$, $\theta'(\nu)$, $\theta^{(2)}(\nu)$,$\cdots$, via the following identities
\begin{subequations}\label{thetaD123 nu}
\begin{align}
& \theta (\nu +\tfrac{\tau}{2})^2
= e^{-2i \pi \nu}\frac{D_1^2}{a}\theta (\nu)^2, \label{thetaD123a nu}
	\\
& \theta(\nu +\tfrac{1+\tau}{2})^2
= e^{-2i \pi \nu}\frac{D_2^2(a-y_{1})}{a}\theta (\nu)^2, \label{thetaD123b nu}
	\\ 
& \theta(\nu +\tfrac{1}{2})^2
= \frac{D_3^2(a-y_2)}{a}\theta (\nu)^2, \label{thetaD123c nu}
\end{align}
\end{subequations}
where $a = \varphi_{A}^{-1}(\nu)$.
\section{Proof of Theorem \ref{thm: main result}}\label{section:integration1}
In Section \ref{section:diffid} we obtained the following differential identity:
\begin{align}\label{final formula diff identity int section}
& \partial_{r} \log F(r\vec{x}) = \sum_{j=1}^{3}K_{j}, & & \mbox{with } \quad K_{j}:= \frac{(-1)^{j+1}x_j}{2\pi i}\lim_{z \to y_{j}}\left(\Psi_+^{-1}\Psi_+'\right)_{21}(rz;rx_{3},r\vec y),
\end{align}
where the limits as $z\to y_j$, $j = 1,2,3,$ are taken such that $z \in (0,y_2) \cup (y_1,+\infty)$. In this section, we use the analysis from Sections \ref{section:RH1}-\ref{section:smallnorm} to obtain large $r$ asymptotics for $K_{j}$, $j=1,2,3$. These asymptotics can be simplified in two steps: first, we use the numerous identities involving the $\theta$-function and the Abel map of Section \ref{Section: theta identities}, and second, we use Riemann's bilinear identity and some other identities for elliptic integrals. Then, we integrate these asymptotics as shown in \eqref{integration diff identity general form} and prove the asymptotic formula \eqref{F asymp} for $\log F(r\vec{x})$. 
\begin{proposition}\label{prop: asymp splitting}
Let $\vec{x} = (x_1, x_2, x_3)$ be fixed and such that $x_3 < x_2 < x_1 < 0$. We have
\begin{align}\nonumber
\partial_r\log F(r \vec{x}) 
    = &\; r^2 \sum_{j=1}^3 (-1)^{j+1} x_j c_{y_j} 
     +  \frac{1}{2\pi i r^{\frac{5}{2}}}\sum_{j=1}^3 (-1)^{j+1}x_j\Big[E_{y_j}(y_j)^{-1} R^{(1)\prime}(y_j)E_{y_j}(y_j)\Big]_{21}
 	\\ 
&     + \frac{1}{2\pi i r}  \sum_{j=1}^3 (-1)^{j+1}x_j \Big[ E_{y_j}(y_j)^{-1} E_{y_j}'(y_j)\Big]_{21} 
      + \bigO \big(r^{-\frac{5}{2}}\big)\label{partialrlogdet}
\end{align}
as $r \to +\infty$, where $c_{y_j}$, $j=1,2,3,$ are given by \eqref{xiy1_expansion}, \eqref{xiy2_expansion} and \eqref{xiy3_expansion}.
\end{proposition}
%\begin{remark}
%Quite surprisingly, it turns out that the terms 
%\begin{align*}
%\frac{1}{2\pi i r^{\frac{5}{2}}}\sum_{j=1}^3 (-1)^{j+1}x_j\Big[E_{y_j}(y_j)^{-1} R^{(1)\prime}(y_j)E_{y_j}(y_j)\Big]_{21}
%\end{align*}
%contributes, after intergration
%\begin{align*}
%\frac{1}{2\pi i r}  \sum_{j=1}^3 (-1)^{j+1}x_j \Big[ E_{y_j}(y_j)^{-1} E_{y_j}'(y_j)\Big]_{21}
%\end{align*}
%contributes for the oscillations of order $1$ in \eqref{}
%\end{remark}
\begin{proof}
Using the transformations $\Psi \to T \to S \to R$ given by \eqref{def of T}, \eqref{def:S}, \eqref{errorMatrix}, we see that, for $z$ inside $\mathbb{D}_j$, $j = 1,2,3$, but outside the lenses, we have
\begin{align*}
\Psi(rz) = \begin{pmatrix}
r^{\frac{1}{4}} & -2ir^{\frac{7}{4}}g_{1} \\
0 & r^{-\frac{1}{4}}
\end{pmatrix} R(z)P^{(y_{j})}(z) e^{r^{\frac{3}{2}}g(z)\sigma_{3}}.
\end{align*}
Thus, for $z$ inside $\mathbb{D}_j$ but outside the lenses, we obtain
\begin{align*}
(\Psi^{-1}\Psi')(rz) = &\; \frac{1}{r}e^{-r^{\frac{3}{2}}g(z)\sigma_{3}}\Big\{ r^{\frac{3}{2}}g'(z)\sigma_{3} + P^{(y_j)}(z)^{-1}P^{(y_j)\prime}(z) 
	\\
& + P^{(y_j)}(z)^{-1}R(z)^{-1}R'(z)P^{(y_j)}(z) \Big\}e^{r^{\frac{3}{2}}g(z)\sigma_{3}}.
\end{align*}
Consequently, for $j=1,2,3,$ we have
\begin{align}\label{Kj expression 1}
& K_j = \frac{(-1)^{j+1}x_j}{2\pi i r} \lim_{z \to y_j} 
e^{2r^{\frac{3}{2}}g_{+}(z)}\Big[P_{+}^{(y_j)}(z)^{-1}P_{+}^{(y_j)\prime}(z) 
 + P_{+}^{(y_j)}(z)^{-1}R(z)^{-1}R'(z)P_{+}^{(y_j)}(z) \Big]_{21},
\end{align}
where again the limits are taken such that $z \in (0,y_{2})\cup (y_{1},+\infty)$, and we write $R(z)$ instead of $R_+(z)$ because $R(z)$ has no jumps on $(0,y_{2})\cup (y_{1},+\infty)$. 
The expressions \eqref{Py1}, \eqref{Py2}, \eqref{P0} for $P^{(y_{j})}(z)$, $j=1,2,3,$ are of the form
\begin{align*}
& P^{(y_j)}(z)=E_{y_{j}}(z)\Phi_{\text{Be}}\left(r^{3} f_{y_j}(z) \right)e^{-r^{\frac{3}{2}}g(z)\sigma_3}, \qquad j=1,3 \\
& P^{(y_{2})}(z) = E_{y_{2}}(z)  \sigma_{3} \Phi_{\mathrm{Be}}(-r^{3} f_{y_{2}}(z))\sigma_{3} e^{-r^{\frac{3}{2}}g(z)\sigma_{3}},
\end{align*}
so we can write \eqref{Kj expression 1} as
\begin{align*}
K_j = &\; \frac{x_j}{2\pi i r} \lim_{z \downarrow y_j} 
\Big[   
\Phi_{\text{Be}}(r^{3}f_{y_{j}}(z))^{-1}
E_{y_j}(z)^{-1}
E_{y_j}'(z)\Phi_{\text{Be}}(r^{3}f_{y_{j}}(z))
	\\
& +  
\Phi_{\text{Be}}(r^{3}f_{y_{j}}(z))^{-1}
\Phi_{\text{Be}}'(r^{3}f_{y_{j}}(z)) r^{3}f_{y_j}'(z)
- r^{\frac{3}{2}}g'(z)\sigma_3
	\\
& + 
\Phi_{\text{Be}}(r^{3}f_{y_{j}}(z))^{-1}
E_{y_j}(z)^{-1}
R(z)^{-1}R'(z)E_{y_j}(z)\Phi_{\text{Be}}(r^{3}f_{y_{j}}(z)) \Big]_{21},  \qquad j=1,3,
	\\
K_2 = & \frac{-x_2}{2\pi i r}\lim_{z \uparrow y_2} 
\Big[   
\sigma_3\Phi_{\text{Be}}(-r^{3}f_{y_{2}}(z))^{-1}\sigma_3
E_{y_2}(z)^{-1} E_{y_2}'(z) \sigma_3 \Phi_{\text{Be}}(-r^{3}f_{y_{2}}(z)) \sigma_3 
	\\
& + 
\sigma_3\Phi_{\text{Be}}(-r^{3}f_{y_{2}}(z))^{-1}
\Phi_{\text{Be}}'(-r^{3}f_{y_{2}}(z))\sigma_3 (-r^{3}f_{y_2}'(z))
- r^{\frac{3}{2}}g'(z)\sigma_3
	\\
& +  
\sigma_3\Phi_{\text{Be}}(-r^{3}f_{y_{2}}(z))^{-1}
\sigma_3E_{y_2}(z)^{-1}
R(z)^{-1}R'(z)E_{y_2}(z)\sigma_3\Phi_{\text{Be}}(r^{3}f_{y_{2}}(z))\sigma_3\Big]_{21}.
\end{align*}
Note that the term $- r^{\frac{3}{2}}g'(z)\sigma_3$ does not contribute for $K_{j}$, since $(\sigma_{3})_{21} = 0$. As $z \to y_j$ with $z$ outside the lenses, we have $(-1)^{j+1}f_{y_j}(z) \to 0$ with $|\arg [ (-1)^{j+1}f_{y_j}(z)]| < \frac{2\pi}{3}$ for each $j = 1,2,3$. Thus, using \eqref{asymp Bessel at 0 3}, we have
\begin{align*}
\lim_{z \to y_j} \Big[\Phi_{\text{Be}}\big((-1)^{j+1}r^{3}f_{y_{j}}(z)\big)^{-1} \Phi_{\text{Be}}'\big((-1)^{j+1}r^{3}f_{y_{j}}(z)\big)\Big]_{21} = 2\pi i, \quad j=1,2,3.
\end{align*}
Moreover, $E_{y_j}(z)$ and $R(z)$ are analytic for $z \in \mathbb{D}_{y_j}$. Hence, using \eqref{asymp Bessel at 0 1}-\eqref{asymp Bessel at 0 3} we have
 \begin{align*}
K_j = &\; \frac{(-1)^{j+1}x_j}{2\pi i r} \Big\{
\Big[ E_{y_j}(y_j)^{-1} E_{y_j}'(y_j)\Big]_{21}
 +  2\pi i r^{3} f_{y_j}'(y_j)
+ \Big[E_{y_j}(y_j)^{-1} R(y_j)^{-1}R'(y_j)E_{y_j}(y_j)\Big]_{21}\Big\}
\end{align*}
for $j =1,2,3$. We also note that $f_{y_j}'(y_j) =  c_{y_j}$ for $j = 1,2,3$. Furthermore, by \eqref{large r asymptotics for R} and \eqref{integral formula for R}, we have
\begin{align*}
R^{-1}(z)R'(z)=\left(I+\bigO(r^{-\frac{3}{2}})\right)\left(\frac{R^{(1)\prime}(z)}{r^\frac{3}{2}}+\bigO(r^{-3})\right), \qquad \mbox{as } r \to + \infty,
\end{align*}
and by \eqref{Ey1y1}, \eqref{Ey2y2} and \eqref{Ey3y3}, we have
\begin{align*}
& E_{y_j}(y_j) = \begin{pmatrix} \bigO(r^{\frac{3}{4}}) & \bigO(r^{-\frac{3}{4}}) \\ \bigO(r^{\frac{3}{4}}) & \bigO(r^{-\frac{3}{4}}) \end{pmatrix}, \qquad r \to +\infty, \ j = 1,2,3,
	\\
& E_{y_j}(y_j)^{-1} = \begin{pmatrix} \bigO(r^{-\frac{3}{4}}) & \bigO(r^{-\frac{3}{4}}) \\\bigO(r^{\frac{3}{4}}) & \bigO(r^{\frac{3}{4}}) \end{pmatrix}, \qquad r \to +\infty, \ j = 1,2,3.
\end{align*}
We infer that, for $j =1,2,3$,
\begin{align*}
K_j =  \frac{(-1)^{j+1}x_j}{2\pi i r} \Big\{
 2\pi i c_{y_j} r^3 
+\Big[ E_{y_j}(y_j)^{-1} E_{y_j}'(y_j)\Big]_{21}
 +  r^{-\frac{3}{2}} \Big[E_{y_j}(y_j)^{-1} R^{(1)\prime}(y_j)E_{y_j}(y_j)\Big]_{21}
+ \bigO(r^{-\frac{3}{2}})\Big\}.
\end{align*}
The proposition follows by summing from $j= 1$ to $j =3$.
\end{proof}

\subsection{Evaluation of the first term}

\begin{proposition}\label{prop: leading term}
The first term on the right-hand side of (\ref{partialrlogdet}) is given by
\begin{align} \label{leading term in prop}
r^2 \sum_{j=1}^3 (-1)^{j+1} x_j c_{y_j}
& = c \, \partial_{r} r^{3},
\end{align}
where $c$ is defined in \eqref{def of C1}.
\end{proposition}
\begin{proof}
By direct computation using the explicit expressions
\begin{align*}
c_{y_1} = \frac{p(y_1)^2}{y_1(y_1-y_2)}, \qquad
c_{y_2} = \frac{p(y_2)^2}{y_2(y_1-y_2)}, \qquad
c_{y_3} = \frac{p(y_3)^2}{y_1 y_2}, 
\end{align*}
and the definition of $p(z)$ given by \eqref{p}, we obtain
\begin{align}
r^2 \sum_{j=1}^3 (-1)^{j+1} x_j c_{y_j} & =  \frac{1}{3}\left( \frac{x_{1}p(y_{1})^{2}}{y_{1}(y_{1}-y_{2})} - \frac{x_{2}p(y_{2})^{2}}{y_{2}(y_{1}-y_{2})} + \frac{x_{3} p(y_{3})^{2}}{y_{1}y_{2}}  \right)\partial_{r} r^{3} \nonumber \\
& = \frac{1}{24} \Big((x_1-x_2)^2 (x_1+x_2-x_3)+2 x_3^3 + 8 g_1 (x_1+x_2+x_3)\Big)\partial_{r} r^{3}. \label{leading term proof}
\end{align}
We note that $\mathcal{Q}, q$ defined in \eqref{def of Q intro}, \eqref{def of p intro} have the relations with $\mathcal{R}, p$
\begin{align*}
\mathcal{Q}(z) = \mathcal{R}(z-x_{3}) \quad \mbox{ and } \quad q(z) = p(z-x_{3}).
\end{align*}
Therefore, $g_{1}$ given by \eqref{def of g1 2cuts} and $q_{0}$ given by \eqref{def of q0} are related by
\begin{align}\label{g1 to q0 rel}
g_{1} = q_{0}-\frac{2x_{1}x_{2}+2x_{1}x_{3}+2x_{2}x_{3}-x_{1}^{2}-x_{2}^{2}}{8}.
\end{align}
We obtain \eqref{leading term in prop} after substituting \eqref{g1 to q0 rel} into \eqref{leading term proof}.
\end{proof}

\subsection{Evaluation of the second term}

\begin{proposition}\label{prop: decomposition of part with Rp1p}
The second term on the right-hand side of (\ref{partialrlogdet}) is given by
\begin{multline}\label{secontermexpression}
 \frac{1}{2\pi i r^{\frac{5}{2}}}\sum_{j=1}^3 (-1)^{j+1} x_j\Big[E_{y_j}(y_j)^{-1} R^{(1)\prime}(y_j) E_{y_j}(y_j)\Big]_{21}   \\
 = \frac{1}{2\pi i r^{\frac{5}{2}}}\sum_{j=1}^3 (-1)^{j} x_j (\mathcal{J}_{y_{j}})_{y_{j}}^{(1)} + \frac{1}{2\pi i r^{\frac{5}{2}}}\sum_{j=1}^3 (-1)^{j} x_j\sum_{\substack{j'=1 \\ j'\neq j}}^3 \frac{1}{(y_{j}-y_{j'})^{2}} (\mathcal{J}_{y_{j}})_{y_{j'}}^{(-1)}.
\end{multline}
\end{proposition}
\begin{proof}
The claim follows after substituting the expression \eqref{der of Rp1p at yj} for $R^{(1)\prime}(y_j)$ into the left-hand side of \eqref{secontermexpression} and by using the definition \eqref{Jcal coeff def} of $(\mathcal{J}_{y_{j}})_{y_{j'}}^{(k)}$.
\end{proof}

\begin{proposition}\label{prop: log r term first term}
The first term on the right-hand side of (\ref{secontermexpression}) is given by
\begin{align}\label{Part 1 of log r term in prop}
\frac{1}{2\pi i r^{\frac{5}{2}}}\sum_{j=1}^3 (-1)^{j} x_j (\mathcal{J}_{y_{j}})_{y_{j}}^{(1)} = & -\frac{1}{16}\sum_{j=1}^{3} \frac{x_{j} p'(y_{j})}{p(y_{j})} \partial_{r}\log{r}.
\end{align}
\end{proposition}
\begin{proof}
This follows from involved computations which are quite long and thus will be omitted here. We use the identities of Corollary \ref{coro:identities without nu} to express $\theta^{(j)}(\frac{1}{2})$, $\theta^{(j)}(\frac{\tau}{2})$ and $\theta^{(j)}(\frac{1+\tau}{2})$, $j=0,1,2,3,$ in terms of only $\theta^{(j)}(0)$, $j=0,2$. We also use the $j$-th derivative with respect to $\nu$ of \eqref{thetaD123a nu}-\eqref{thetaD123c nu}, $j=0,1,2,3,$ to express $\theta^{(j)}(\nu+\frac{1}{2})$, $\theta^{(j)}(\nu+\frac{\tau}{2})$ and $\theta^{(j)}(\nu+\frac{1+\tau}{2})$, $j=0,1,2,3,$ in terms of only $\theta^{(j)}(\nu)$, $j=0,1,2,3$. Remarkably, the final expression on the right-hand side of \eqref{Part 1 of log r term in prop} does not contain the $\theta$-function.
\end{proof}
\begin{definition}
Let $\varphi_{\mathbb{C}}^{-1}(u)$ be the projection into the complex plane of $\varphi_{A}^{-1}(u)$. We define $a = a(r)$ by
\begin{align}\label{def of a in prop}
a := \varphi_{\mathbb{C}}^{-1}(\nu) = \varphi_{\mathbb{C}}^{-1}\Big(-\frac{\Omega }{2\pi}r^\frac{3}{2} \Big).
\end{align}
From the relations $\varphi_{\mathbb{C}}^{-1}(u)=\varphi_{\mathbb{C}}^{-1}(-u)$, $\varphi_{\mathbb{C}}^{-1}(u+1)=\varphi_{\mathbb{C}}^{-1}(u)$ and $\Omega > 0$, we conclude that the function $r \mapsto a(r)$ is oscillatory as $r \to + \infty$. Furthermore, we see that $\varphi_{A}(x)$ is monotone for $x \in (y_{1},+\infty)$, which follows from direct inspection of \eqref{def of u}. Since $\varphi_{A}(y_{1}) = 0$ and $\varphi_{A}(\infty) = \frac{1}{2}$, we conclude that $a(r)$ is bigger than $y_{1}$ for any $r$ and oscillates between $y_{1}$ and $+\infty$.
\end{definition}
\begin{proposition}\label{prop: log r term second term}
The second term on the right-hand side of (\ref{secontermexpression}) is given by
\begin{align}\label{log r term second term in prop}
\frac{1}{2\pi i r^{\frac{5}{2}}}\sum_{j=1}^3 (-1)^{j} x_j\sum_{\substack{j'=1 \\ j'\neq j}}^3 \frac{1}{(y_{j}-y_{j'})^{2}} (\mathcal{J}_{y_{j}})_{y_{j'}}^{(-1)} = n^{(0)}\partial_{r} \log r + n^{(-1)} \partial_{r}\int_{M}^{r} \frac{d\tilde{r}}{\tilde{r}a(\tilde{r})},
\end{align}
where $a=a(r)$ is oscillatory and given by \eqref{def of a in prop}, $M>0$ is a large constant, and $n^{(-1)}, n^{(0)}$ are given by
\begin{align}
& n^{(-1)} = -\frac{1}{32}\left( \frac{1}{p(y_{1})}+\frac{1}{p(y_{2})}+\frac{1}{p(y_{3})} \right)(x_{1}-x_{3})(x_{2}-x_{3})(x_{1}+x_{2}+x_{3}), \label{npm1p prop} \\ 
& n^{(0)} = - \frac{x_{1} \frac{y_{2}}{y_{1}}p(y_{1})-x_{2} \frac{y_{1}}{y_{2}}p(y_{2})}{16 p(y_{3})(x_{1}-x_{2})} - \frac{x_{3} p(y_{3})}{16}\left( \frac{1}{y_{1}p(y_{1})}+\frac{1}{y_{2}p(y_{2})} \right). \label{npop prop}
\end{align}
\end{proposition}
\begin{proof}
By using Proposition \ref{p_prop}, Corollary \ref{coro:identities without nu} together with the derivatives with respect to $\nu$ of \eqref{thetaD123a nu}-\eqref{thetaD123c nu}, we can simplify the quantities obtained at the end of Section \ref{section:smallnorm} for $(\mathcal{J}_{y_{j}})_{y_{j'}}^{(-1)}$, $1 \leq j \neq j' \leq 3,$ as follows
\begin{align*}
& (\mathcal{J}_{y_{1}})_{y_{2}}^{(-1)} = - \frac{\pi i}{8} \frac{(y_{1}-y_{2})y_{2} p(y_{1})}{p(y_{2})}\frac{r^\frac{3}{2}}{a}, & & (\mathcal{J}_{y_{2}})_{y_{1}}^{(-1)} = - \frac{\pi i}{8} \frac{(y_{1}-y_{2})y_{1} p(y_{2})}{p(y_{1})}\frac{r^\frac{3}{2}}{a}, \\
& (\mathcal{J}_{y_{1}})_{y_{3}}^{(-1)} = \frac{\pi i}{8}\frac{y_{1} y_{2} p(y_{1})}{(y_{1}-y_{2})p(y_{3})} \frac{r^\frac{3}{2}(a-y_{1})}{a}, & & (\mathcal{J}_{y_{3}})_{y_{1}}^{(-1)} = \frac{\pi i}{8} \frac{y_{1}p(y_3)}{p(y_{1})}\frac{r^\frac{3}{2}(a-y_{1})}{a}, \\
& (\mathcal{J}_{y_{2}})_{y_{3}}^{(-1)} = \frac{\pi i}{8}\frac{y_{1} y_{2} p(y_{2})}{(y_{1}-y_{2})p(y_{3})} \frac{r^\frac{3}{2}(a-y_{2})}{a}, & & (\mathcal{J}_{y_{3}})_{y_{2}}^{(-1)} = \frac{\pi i}{8}\frac{y_{2}p(y_{3})}{p(y_{2})}\frac{r^\frac{3}{2}(a-y_ {2})}{a},
\end{align*}
where $a = a(r)$ depends on $r$ and is given by \eqref{def of a in prop}. Then, a long but straightforward computation gives
\begin{align*}
\frac{1}{2\pi i r^{\frac{5}{2}}}\sum_{j=1}^3 (-1)^{j} x_j\sum_{\substack{j'=1 \\ j'\neq j}}^3 \frac{1}{(y_{j}-y_{j'})^{2}} (\mathcal{J}_{y_{j}})_{y_{j'}}^{(-1)} = n^{(0)}\partial_{r} \log r + n^{(-1)} \partial_{r}\int_{M}^{r} \frac{d\tilde{r}}{\tilde{r}a(\tilde{r})},
\end{align*}
where $n^{(0)}$ is given by \eqref{npop prop} and where $n^{(-1)}$ is given by
\begin{align*}
n^{(-1)} = \frac{1}{16}\left( \frac{1}{p(y_{1})}+\frac{1}{p(y_{2})}+\frac{1}{p(y_{3})} \right)\left( x_{1}p(y_{1}) \frac{x_{2}-x_{3}}{x_{1}-x_{2}}-x_{2}p(y_{2}) \frac{x_{1}-x_{3}}{x_{1}-x_{2}}+x_{3}p(y_{3}) \right).
\end{align*}
Substituting the explicit expression \eqref{p} for $p$ leads to further simplifications and we find  \eqref{npm1p prop}.
\end{proof}
\noindent The next proposition shows that $\int_{M}^{r} \frac{d\tilde{r}}{\tilde{r}a(\tilde{r})}$ is proportional to $\log r$ and contains no oscillations of order $1$. 
\begin{proposition}\label{prop: asymptotic of Abel integral}
As $r \to + \infty$, we have
\begin{align}\label{Abel integral asymptotics}
\int_{M}^{r} \frac{d\tilde{r}}{\tilde{r} \, a(\tilde{r})} = 2 \int_{y_{1}}^{+\infty} \frac{\varphi'(x)}{x}dx \log r + \tilde{C} +  \bigO (r^{-1}),
\end{align}
where $\tilde{C}$ is a constant independent of $r$.
\end{proposition}
\begin{proof}
Let $r_{0} > M$ be a large constant smaller than $r$ and such that $a(r_{0})=y_{1}$. Since $\Omega > 0$, the values of $r\geq r_{0}$ for which $a(r) = +\infty$ or $a(r) = y_{1}$ are given by
\begin{align*}
r_{j} := \Big(r_{0}^\frac{3}{2} + j \frac{\pi}{\Omega}\Big)^{\frac{2}{3}}, \qquad j \geq 0.
\end{align*} 
Let us write
\begin{align*}
r = \Big(r_{0}^\frac{3}{2} + 2k \frac{\pi}{\Omega} + \lambda \frac{\pi}{\Omega}\Big)^{\frac{2}{3}}, \quad \mbox{ where } \quad k \in \mathbb{N} \quad \mbox{ and } \quad 0 \leq \lambda < 2,
\end{align*}
so that
\begin{align*}
\int_{M}^{r} \frac{d\tilde{r}}{\tilde{r} \, a(\tilde{r})} = \int_{M}^{r_{0}} \frac{d\tilde{r}}{\tilde{r} \, a(\tilde{r})} + \sum_{j=0}^{k-1} \bigg[ \int_{r_{2j}}^{r_{2j+1}} \frac{d\tilde{r}}{\tilde{r} \, a(\tilde{r})} + \int_{r_{2j+1}}^{r_{2j+2}} \frac{d\tilde{r}}{\tilde{r} \, a(\tilde{r})} \bigg] + \int_{r_{2k}}^{r_{2k} + \lambda \frac{\pi}{\Omega}} \frac{d\tilde{r}}{\tilde{r} \, a(\tilde{r})}.
\end{align*}
From the change of variables $\nu = - \frac{\Omega}{2\pi}r^\frac{3}{2}$ and using the notation $\nu_{j} = - \frac{\Omega}{2\pi}r_{j}^\frac{3}{2}$, we obtain
\begin{align*}
& \int_{r_{2j}}^{r_{2j+1}} \frac{d\tilde{r}}{\tilde{r} \, a(\tilde{r})} = \frac{2}{3}\int_{\nu_{2j}}^{\nu_{2j+1}} \frac{d\nu}{\nu \varphi_{\mathbb{C}}^{-1}(\nu)} = \frac{2}{3} \int_{\nu_{2j}}^{\nu_{2j+1}} \frac{1}{\nu_{2j}}\left( 1 + \bigO \left( \frac{\nu-\nu_{2j}}{\nu_{2j}} \right)\right) \frac{d\nu}{\varphi_{\mathbb{C}}^{-1}(\nu)} \\
& =-\frac{4\pi}{3 \Omega}r_{2j}^{-\frac{3}{2}} \left( \int_{\nu_{2j}}^{\nu_{2j+1}} \frac{d\nu}{\varphi_{\mathbb{C}}^{-1}(\nu)} + \bigO \left( \frac{\nu_{2j+1}-\nu_{2j}}{\nu_{2j}} \right) \right)=-\frac{4\pi}{3 \Omega}r_{2j}^{-\frac{3}{2}} \left( \int_{\nu_{2j}}^{\nu_{2j+1}} \frac{d\nu}{\varphi_{\mathbb{C}}^{-1}(\nu)} + \bigO \left( r_{2j}^{-1} \right) \right)
\end{align*}
as $r_{2j} \to +\infty$. Noting that $\nu_{2j+1} = \nu_{2j}-\frac{1}{2}$, $\varphi_{\mathbb{C}}^{-1}(\nu_{2j+1}) = + \infty$, and $\varphi_{\mathbb{C}}^{-1}(\nu)$ is strictly increasing for $\nu$ in the oriented segment $(\nu_{2j},\nu_{2j+1})$, we use the change of variables $x = \varphi_{\mathbb{C}}^{-1}(\nu)$ to obtain
\begin{align*}
\int_{r_{2j}}^{r_{2j+1}} \frac{d\tilde{r}}{\tilde{r} \, a(\tilde{r})} = \frac{4\pi}{3 \Omega}r_{2j}^{-\frac{3}{2}} \left( \int_{y_{1}}^{+\infty} \frac{\varphi'(x)}{x}dx + \bigO \left( r_{2j}^{-1} \right) \right), \qquad j=0,\ldots,k-1.
\end{align*}
Similarly,
\begin{align*}
\int_{r_{2j+1}}^{r_{2j+2}} \frac{d\tilde{r}}{\tilde{r} \, a(\tilde{r})} = \frac{4\pi}{3 \Omega}r_{2j+1}^{-\frac{3}{2}} \left( \int_{y_{1}}^{+\infty} \frac{\varphi'(x)}{x}dx + \bigO \left( r_{2j+1}^{-1} \right) \right), \qquad j=0,\ldots,k-1.
\end{align*}
As $r \to + \infty$,
\begin{align*}
\sum_{j=0}^{2k-1} r_{j}^{-\frac{3}{2}} = \frac{\Omega}{\pi} \log k + d_{0} + \bigO(k^{-1}) = \frac{3\Omega}{2\pi} \log r + d_{1} + \bigO(r^{-\frac{3}{2}}), \qquad \sum_{j=0}^{2k-1} r_{j}^{-\frac{5}{2}} = d_{2} + \bigO(r^{-1}),
\end{align*}
for certain constants $d_{0},d_{1},d_{2}$, so we have
\begin{align*}
\sum_{j=0}^{k-1} \bigg[ \int_{r_{2j}}^{r_{2j+1}} \frac{d\tilde{r}}{\tilde{r} \, a(\tilde{r})} + \int_{r_{2j+1}}^{r_{2j+2}} \frac{d\tilde{r}}{\tilde{r} \, a(\tilde{r})} \bigg] & = 2 \int_{y_{1}}^{+\infty} \frac{\varphi'(x)}{x}dx + \bigO \left( r^{-1} \right)
\end{align*}
as $r \to + \infty$. Since $a(\tilde{r}) \geq y_{1}$ for all $\tilde{r}$,
\begin{align*}
\int_{r_{2k}}^{r_{2k} + \lambda \frac{\pi}{\Omega}} \frac{d\tilde{r}}{\tilde{r} \, a(\tilde{r})} = \bigO(r^{-1}), \qquad \mbox{ as } r \to + \infty,
\end{align*}
and the proof is completed.
\end{proof}
In Propositions \ref{prop: log r term first term} and \ref{prop: log r term second term}, we were able to express the right-hand side of \eqref{secontermexpression} as a perfect derivative and in Proposition \ref{prop: asymptotic of Abel integral} we obtained the leading order $r$ behavior of the only non-explicit term appearing on the right-hand side of \eqref{secontermexpression}.  Thus, by combining Propositions \ref{prop: decomposition of part with Rp1p}, \ref{prop: log r term first term}, \ref{prop: log r term second term}, \ref{prop: asymptotic of Abel integral}, we have, as $r \to + \infty$,
\begin{align}\label{secontermexpression 2}
\int_{M}^{r} \frac{1}{2\pi i \tilde{r}^{\frac{5}{2}}}\sum_{j=1}^3 (-1)^{j+1} x_j\Big[E_{y_j}(y_j)^{-1} R^{(1)\prime}(y_j) E_{y_j}(y_j)\Big]_{21}d\tilde{r} =c_{2} \log r + \tilde{C}_{2} + \bigO(r^{-1}),
\end{align}
where $M$ and $\tilde{C}_{2}$ are independent of $r$, and $c_{2}$ is given by
\begin{align}\label{c2 without theta}
c_{2} = -\frac{1}{16}\sum_{j=1}^{3} \frac{x_{j} p'(y_{j})}{p(y_{j})} + n^{(0)} + 2 n^{(-1)} \int_{y_{1}}^{+\infty} \frac{\varphi'(x)}{x}dx.
\end{align}
This expression for $c_{2}$ could be obtained thanks to the considerable simplifications of Propositions \ref{prop: log r term first term}-\ref{prop: asymptotic of Abel integral}, and does not involve the $\theta$-function. Surprisingly, it is possible to simplify $c_{2}$ significantly more.
\begin{proposition}\label{prop: c2 is -1/2}
We have $c_{2} = -\frac{1}{2}$.
\end{proposition}
\begin{proof}
Using \cite[eqs 3.131.4, 3.134.28]{GRtable}, we obtain
\begin{align}\label{inty1inftyvarphi}
    \int_{y_{1}}^{+\infty} \frac{\varphi'(x)}{x}dx=\frac{1}{2y_2}\left(1-\frac{\textbf{E}(k_\star)}{\textbf{K}(k_\star)}\right),
\end{align}
where $k_\star=\sqrt{\frac{y_2}{y_1}}$.  Recall from the proof of Proposition \ref{p_prop} that $p(z)$ can be written
\begin{align}
    p(z)=-z^2+\frac{z}{2}(y_1+y_2-x_3)+\alpha, ~~~ \alpha=\frac{y_1}{2}\left[\frac{y_1-y_2}{3}+x_3-\frac{\mathbf{E}(k_{\star})}{\mathbf{K} (k_{\star})}\left(\frac{y_1+y_2}{3}+x_3\right)\right]. \label{alt p rep}
\end{align}
Since $y_3=0$, and hence $p(y_3)\equiv p(0)=\alpha$, we have the relation
\begin{align}
    1-\frac{\textbf{E}(k_\star)}{\textbf{K}(k_\star)}=\frac{6p(y_3)+2y_1y_2}{y_1(x_1+x_2+x_3)} \nonumber
\end{align}
which together with (\ref{inty1inftyvarphi}) and the definition of $n^{(-1)}$ \eqref{npm1p prop} implies that
\begin{align}\label{npm1p simplified}
    2 n^{(-1)} \int_{y_{1}}^{+\infty} \frac{\varphi'(x)}{x}dx&=-\frac{1}{16}\left( \frac{1}{p(y_{1})}+\frac{1}{p(y_{2})}+\frac{1}{p(y_{3})} \right)(3p(y_3)+y_1y_2).
\end{align}
Using \eqref{alt p rep}, we find the following identites:
\begin{align*}
    -\frac{x_3p(y_3)}{16y_1p(y_1)}-\frac{x_1p'(y_1)}{16p(y_1)}&=\frac{\frac{x_1}{2}(3y_1-y_2+x_3)-\frac{x_3}{y_1}\alpha}{16p(y_1)}, \\
    -\frac{x_2p'(y_2)}{16p(y_2)}+\frac{x_3p(y_3)}{16y_2p(y_2)}&=\frac{\frac{x_2}{2}(-y_1+3y_2+x_3)-\frac{x_3}{y_2}\alpha}{16p(y_2)}, \\
    -\frac{x_{1} \frac{y_{2}}{y_{1}}p(y_{1})-x_{2} \frac{y_{1}}{y_{2}}p(y_{2})}{16 p(y_{3})(x_{1}-x_{2})}-\frac{x_3p'(y_3)}{16p(y_3)}&=\frac{y_1y_2}{16p(y_3)}-\frac{1}{16}\left(1-\frac{x_1}{y_1}-\frac{x_2}{y_2}\right).
\end{align*}
Substituting the above identites and \eqref{npm1p simplified} into \eqref{c2 without theta}, we obtain
\begin{multline}
    c_2=  -\frac{1}{16}\left(4-\frac{x_1}{y_1}-\frac{x_2}{y_2}\right)+\frac{\frac{x_1}{2}(3y_1-y_2+x_3)-y_1y_2-\alpha(3+\frac{x_3}{y_1})}{16p(y_1)}+ 
    	\\
    +\frac{\frac{x_2}{2}(-y_1+3y_2+x_3)-y_1y_2-\alpha(3+\frac{x_3}{y_2})}{16p(y_2)}. \label{c2 half simplified}
\end{multline}
The two numerators in \eqref{c2 half simplified} can be simplified as follows:
\begin{align*}
    \frac{x_1}{2}(3y_1-y_2+x_3)-y_1y_2-\alpha\left(3+\frac{x_3}{y_1}\right)&=-p(y_1)\left(3+\frac{x_3}{y_1}\right), \\
    \frac{x_2}{2}(-y_1+3y_2+x_3)-y_1y_2-\alpha\left(3+\frac{x_3}{y_2}\right)&=-p(y_2)\left(3+\frac{x_3}{y_2}\right),
\end{align*}
and hence 
\begin{align*}
    c_2=-\frac{1}{16}\left(4-\frac{x_1}{y_1}-\frac{x_2}{y_2}+3+\frac{x_3}{y_1}+3+\frac{x_3}{y_2}\right)=-\frac{1}{2},
\end{align*}
as desired.
\end{proof}
\subsection{Evaluation of the third term}

We now consider the evaluation of the third term in (\ref{partialrlogdet}) given by
\begin{align}\label{third term simple sum}
\frac{1}{2\pi i r}  \sum_{j=1}^3 (-1)^{j+1}x_j \Big[ E_{y_j}(y_j)^{-1} E_{y_j}'(y_j)\Big]_{21}.
\end{align}

\begin{proposition}[Contribution from $y_1$]\label{thirdtermy1prop}
The first term in the sum \eqref{third term simple sum} is given by
\begin{align*}
\frac{x_1}{2\pi i r} \Big[ E_{y_1}(y_1)^{-1} E_{y_1}'(y_1)\Big]_{21}
= \frac{8\pi c_{0} x_1 p(y_{1}) }{3\Omega y_{1}(y_{1}-y_{2})}\frac{\partial}{\partial r}\log(\theta(\nu)).
\end{align*}
\end{proposition}
\begin{proof}
A long but direct computation using \eqref{Ey1y1} and $\theta'(0)=0$ gives
\begin{align*}
\frac{x_1}{2\pi i r} \Big[ E_{y_1}(y_1)^{-1} E_{y_1}'(y_1)\Big]_{21} = - x_1 \sqrt{r}\sqrt{c_{y_1}} \varphi_{y_1}^{(\frac{1}{2})} \frac{\theta'(\nu)}{\theta(\nu)} = - \frac{2c_{0} \sqrt{r} x_{1}p(y_{1})}{y_{1}(y_{1}-y_{2})}\frac{\theta'(\nu)}{\theta(\nu)}.
\end{align*}
Recalling that $\nu = -\frac{\Omega}{2\pi}r^\frac{3}{2}$, the claim follows.
\end{proof}
The second and third terms in the sum \eqref{third term simple sum} are significantly more complicated to analyze than the first term computed in Proposition \ref{thirdtermy1prop} for two reasons: 1) they require several identities derived in Section \ref{Section: theta identities} (while the first term only uses $\theta'(0) = 0$), and 2) the related computations are significantly longer.
\begin{proposition}[Contribution from $y_2$]\label{thirdtermy2prop}
The second term in the sum \eqref{third term simple sum} is given by
\begin{align*}
\frac{-x_2}{2\pi i r} \Big[ E_{y_2}(y_2)^{-1} E_{y_2}'(y_2)\Big]_{21}
= -\frac{8\pi}{3\Omega} \frac{x_2 p(y_2) c_0}{y_2 (y_1-y_2)}
\frac{\partial}{\partial r}\log \theta(\nu).
\end{align*}
\end{proposition}
\begin{proof}
The proof uses Propositions \ref{p_prop}, \ref{prop: Abel map and theta function} and Corollaries \ref{coro: abel map identity 1}, \ref{coro:identities without nu}. After long computations and considerable simplifications, we obtain
\begin{align*}
\frac{-x_2}{2\pi i r} \Big[ E_{y_2}(y_2)^{-1} E_{y_2}'(y_2)\Big]_{21} = \frac{2c_{0}\sqrt{r}x_{2}p(y_{2})}{(y_{1}-y_{2})y_{2}}\frac{\theta'(\nu)}{\theta(\nu)}.
\end{align*}
The claim follows as in Proposition \ref{thirdtermy1prop} by using $\nu = - \frac{\Omega}{2\pi}r^\frac{3}{2}$.
\end{proof}

\begin{proposition}[Contribution from $y_3$]\label{thirdtermy3prop}
The third term in the sum \eqref{third term simple sum} is given by
\begin{align*}
\frac{x_3}{2\pi i r} \Big[ E_{y_3}(y_3)^{-1} E_{y_3}'(y_3)\Big]_{21}
 & = \frac{8\pi}{3\Omega} \frac{x_3 p(y_3) c_0}{y_1y_2} \partial_{r} \log \theta(\nu).
\end{align*}
\end{proposition}
\begin{proof}
The proof is similar to the proof of Proposition \ref{thirdtermy2prop}, and uses also Propositions \ref{p_prop}, \ref{prop: Abel map and theta function} and Corollaries \ref{coro: abel map identity 1}, \ref{coro:identities without nu}.
\end{proof}
It follows from Propositions \ref{thirdtermy1prop}, \ref{thirdtermy2prop} and \ref{thirdtermy3prop} that
\begin{align}\label{third term simple sum and c3}
\frac{1}{2\pi i r}  \sum_{j=1}^3 (-1)^{j+1}x_j \Big[ E_{y_j}(y_j)^{-1} E_{y_j}'(y_j)\Big]_{21} = c_{3} \, \partial_{r} \log \theta(\nu),
\end{align}
where the constant $c_{3}$ is given by
\begin{align}\label{def of c3 first}
c_{3} = \frac{8\pi c_{0}}{3\Omega} \left( \frac{x_{1} p(y_{1})}{(x_{1}-x_{3})(x_{1}-x_{2})} - \frac{x_{2}p(y_{2})}{(x_{2}-x_{3})(x_{1}-x_{2})} + \frac{x_{3}p(y_{3})}{(x_{1}-x_{3})(x_{2}-x_{3})} \right).
\end{align}
The three terms of the sum appearing on the left-hand side of \eqref{third term simple sum and c3} have been simplified separately in Propositions \ref{thirdtermy1prop}, \ref{thirdtermy2prop} and \ref{thirdtermy3prop}, so that the expression \eqref{def of c3 first} for $c_{3}$ does not involve the $\theta$-function. Remarkably, these three terms  combine together and it possible to further simplify $c_{3}$. The tools needed to simplify $c_{3}$ are of a different nature than those needed to simplify $c_{2}$ in \eqref{secontermexpression 2}. In Proposition \ref{prop: simplify c3} below, we show that $c_{3} = 1$ by using Riemann's bilinear identity.
\begin{proposition}\label{prop: simplify c3}
We have $c_{3} = 1$.
\end{proposition}
\begin{proof}
First, we use the explicit expression \eqref{p} for $p$ to note that \eqref{def of c3 first} can be simplified into
\begin{align}\label{is C3 0 ?}
c_{3} = - \frac{4\pi c_{0}}{3\Omega} (x_{1}+x_{2}+x_{3}).
\end{align}
We consider the following one-forms on the Riemann surface $X$:
\begin{align*}
\omega_{1} = \frac{dz}{\sqrt{\mathcal{R}(z)}}, \hspace{2cm} \omega_{2} = \frac{(p(z)+g_{1})dz}{\sqrt{\mathcal{R}(z)}}.
\end{align*}
Let us deform the cycle $B$ of Figure \ref{fig: cycles Riemann surface} such that it has finite length and surrounds the interval $[y_{2},y_{1}]$ in the positive direction. Riemann's bilinear identity gives the relation
\begin{align*}
A_{1}B_{2}-A_{2}B_{1} = 2\pi i \sum_{p=y_{3},y_{2},y_{1},\infty} \mbox{Res} \left( \omega_{2} \int_{\infty}^{P} \omega_{1} ; P=p \right),
\end{align*}
where
\begin{align*}
& A_{1} = \int_{A} \omega_{1} = 2 \int_{0}^{y_{2}} \frac{dx}{\big| \sqrt{\mathcal{R}(x)} \big|}, & & B_{1} = \int_{B} \omega_{1} = 2i \int_{y_{2}}^{y_{1}} \frac{dx}{\big| \sqrt{\mathcal{R}(x)} \big|}, \\
& A_{2} = \int_{A} \omega_{2} = 2 \int_{0}^{y_{2}} \frac{(p(x)+g_{1})dx}{\big| \sqrt{\mathcal{R}(x)} \big|}, & & B_{2} = \int_{B} \omega_{2} = 2i \int_{y_{2}}^{y_{1}} \frac{(p(x)+g_{1})dx}{\big| \sqrt{\mathcal{R}(x)} \big|}.
\end{align*}
A simple computation shows that the Abelian differential
\begin{align*}
\omega_{2} \int_{\infty}^{P} \omega_{1}
\end{align*}
has no residue at $y_{3},y_{2},y_{1}$, and has a residue at $\infty$ given by $-\frac{2}{3}(x_{1}+x_{2}+x_{3})$ in the local coordinate. Therefore, we have
\begin{align}\label{Riemann Bilinear identity}
A_{1}B_{2}-A_{2}B_{1} = - \frac{4\pi i}{3}(x_{1}+x_{2}+x_{3}).
\end{align}
Using the definition of $g_{1}$, $\Omega$ and $c_{0}$, it is a direct computation to verify from \eqref{is C3 0 ?} that the claim $c_{3}=1$ is equivalent to \eqref{Riemann Bilinear identity}.
\end{proof}
Theorem \ref{thm: main result} follows directly by combining Propositions \ref{prop: asymp splitting}, \ref{prop: leading term}, \ref{prop: c2 is -1/2} and \ref{prop: simplify c3}, together with equations \eqref{secontermexpression 2} and \eqref{third term simple sum and c3}.

%%%%%%%%%% APPENDIX %%%%%%%%%%

\appendix
\section{Bessel Model RH problem}\label{Section:Appendix}
In this section, we recall the Bessel model RH problem (see \cite{DIZ}) whose solution is denoted by $\Phi_{\mathrm{Be}}$.
\begin{itemize}
\item[(a)] $\Phi_{\mathrm{Be}} : \mathbb{C} \setminus \Sigma_{\mathrm{Be}} \to \mathbb{C}^{2\times 2}$ is analytic, where
$\Sigma_{\mathrm{Be}}$ is shown in Figure \ref{figBessel}.
\item[(b)] $\Phi_{\mathrm{Be}}$ satisfies the jump conditions
\begin{equation}\label{Jump for P_Be}
\begin{array}{l l} 
\Phi_{\mathrm{Be},+}(z) = \Phi_{\mathrm{Be},-}(z) \begin{pmatrix}
0 & 1 \\ -1 & 0
\end{pmatrix}, & z \in \mathbb{R}^{-}, \\

\Phi_{\mathrm{Be},+}(z) = \Phi_{\mathrm{Be},-}(z) \begin{pmatrix}
1 & 0 \\ 1 & 1
\end{pmatrix}, & z \in e^{ \frac{2\pi i}{3} }  \mathbb{R}^{+}, \\

\Phi_{\mathrm{Be},+}(z) = \Phi_{\mathrm{Be},-}(z) \begin{pmatrix}
1 & 0 \\ 1 & 1
\end{pmatrix}, & z \in e^{ -\frac{2\pi i}{3} }  \mathbb{R}^{+}. \\
\end{array}
\end{equation}
\item[(c)] As $z \to \infty$, $z \notin \Sigma_{\mathrm{Be}}$, we have
\begin{equation}\label{large z asymptotics Bessel}
\Phi_{\mathrm{Be}}(z) = ( 2\pi z^{\frac{1}{2}} )^{-\frac{\sigma_{3}}{2}}M
\left(I+ \frac{\Phi_{\mathrm{Be},1}}{z^{\frac{1}{2}}} + \bigO(z^{-1}) \right) e^{2 z^{\frac{1}{2}}\sigma_{3}},
\end{equation}
where $\Phi_{\mathrm{Be},1} = \frac{1}{16}\begin{pmatrix}
-1 & -2i \\ -2i & 1
\end{pmatrix}$ and $M = \frac{1}{\sqrt{2}}\begin{pmatrix}
1 & i \\ i & 1
\end{pmatrix}$.
\item[(d)] As $z$ tends to 0, the behavior of $\Phi_{\mathrm{Be}}(z)$ is
\begin{equation}\label{local behavior near 0 of P_Be}
\Phi_{\mathrm{Be}}(z) = \left\{ \begin{array}{l l}
\begin{pmatrix}
\bigO(1) & \bigO(|\log z|) \\
\bigO(z) & \bigO(|\log z|) 
\end{pmatrix}, & |\arg z| < \frac{2\pi}{3}, \\
\begin{pmatrix}
\bigO(|\log z|) & \bigO(|\log z|) \\
\bigO(|\log z|) & \bigO(|\log z|) 
\end{pmatrix}, & \frac{2\pi}{3}< |\arg z| < \pi.
\end{array}  \right.
\end{equation}
\end{itemize}
\begin{figure}[t]
    \begin{center}
    \setlength{\unitlength}{1truemm}
    \begin{picture}(100,55)(-5,10)
%        \dottedline{4}(1,0){30}
        \put(50,40){\line(-1,0){30}}
        \put(50,39.8){\thicklines\circle*{1.2}}
        \put(50,40){\line(-0.5,0.866){15}}
        \put(50,40){\line(-0.5,-0.866){15}}
        %\qbezier(53,40)(52,43)(48.5,42.598)
        %\put(53,43){$\frac{2\pi}{3}$}
        \put(50.3,36.8){$0$}
        %\put(65,39.9){\thicklines\vector(1,0){.0001}}
        \put(35,39.9){\thicklines\vector(1,0){.0001}}
        \put(41,55.588){\thicklines\vector(0.5,-0.866){.0001}}
        \put(41,24.412){\thicklines\vector(0.5,0.866){.0001}}
        %\put(60,60){$\Sigma_{\widehat{\Psi}}$}
    \end{picture}
    \caption{\label{figBessel}The jump contour $\Sigma_{\mathrm{Be}}$ for $\Phi_{\mathrm{Be}}$.}
\end{center}
\end{figure}
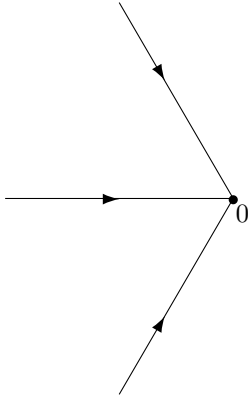
The unique solution to this RH problem is given by 
\begin{equation}\label{Psi explicit}
\Phi_{\mathrm{Be}}(z)=
\begin{cases}
\begin{pmatrix}
I_{0}(2 z^{\frac{1}{2}}) & \frac{ i}{\pi} K_{0}(2 z^{\frac{1}{2}}) \\
2\pi i z^{\frac{1}{2}} I_{0}^{\prime}(2 z^{\frac{1}{2}}) & -2 z^{\frac{1}{2}} K_{0}^{\prime}(2 z^{\frac{1}{2}})
\end{pmatrix}, & |\arg z | < \frac{2\pi}{3}, \\

\begin{pmatrix}
\frac{1}{2} H_{0}^{(1)}(2(-z)^{\frac{1}{2}}) & \frac{1}{2} H_{0}^{(2)}(2(-z)^{\frac{1}{2}}) \\
\pi z^{\frac{1}{2}} \big( H_{0}^{(1)} \big)^{\prime} (2(-z)^{\frac{1}{2}}) & \pi z^{\frac{1}{2}} \big( H_{0}^{(2)} \big)^{\prime} (2(-z)^{\frac{1}{2}})
\end{pmatrix}, & \frac{2\pi}{3} < \arg z < \pi, \\

\begin{pmatrix}
\frac{1}{2} H_{0}^{(2)}(2(-z)^{\frac{1}{2}}) & -\frac{1}{2} H_{0}^{(1)}(2(-z)^{\frac{1}{2}}) \\
-\pi z^{\frac{1}{2}} \big( H_{0}^{(2)} \big)^{\prime} (2(-z)^{\frac{1}{2}}) & \pi z^{\frac{1}{2}} \big( H_{0}^{(1)} \big)^{\prime} (2(-z)^{\frac{1}{2}})
\end{pmatrix}, & -\pi < \arg z < -\frac{2\pi}{3},
\end{cases}
\end{equation}
where $H_{0}^{(1)}$ and $H_{0}^{(2)}$ are the Hankel functions of the first and second kind, and $I_0$ and $K_0$ are the modified Bessel functions of the first and second kind.

For certain computations, we will need a more precise expansion than \eqref{local behavior near 0 of P_Be}. By \cite[Section 10.30(i)]{NIST}), as $z \to 0$ we have
\begin{align}
& \Phi_{\mathrm{Be}}(z) = \begin{pmatrix}
1 + \bigO(z) & \bigO(|\log z|) \\
\bigO(z) & \bigO(1) 
\end{pmatrix}, \qquad z \to 0, \quad |\arg z | < \frac{2\pi}{3},
\label{asymp Bessel at 0 1}	\\
& \Phi_{\mathrm{Be}}(z)^{-1} = \begin{pmatrix}
\bigO(1) & \bigO(|\log z|) \\
\bigO(z) & 1 + \bigO(z) 
\end{pmatrix}, \qquad z \to 0, \quad |\arg z | < \frac{2\pi}{3},
\label{asymp Bessel at 0 2}	\\
& \Phi_{\mathrm{Be}}(z)^{-1}\Phi_{\mathrm{Be}}'(z) = \begin{pmatrix}
\bigO(|\log z|) & \bigO(z^{-1}) \\
2\pi i + \bigO(z) & \bigO(|\log{z}|) 
\end{pmatrix}, \qquad z \to 0, \quad |\arg z | < \frac{2\pi}{3}. \label{asymp Bessel at 0 3}
\end{align}

%%%%%%%%%%%%%%     REFERENCES     %%%%%%%%%%%%%%


\begin{thebibliography}{99}
%\addcontentsline{toc}{section}{References}



%\bibitem{AS64}  M. Abramowitz and I. A. Stegun, Handbook of Mathematical Functions with Formulas,
%Graphs, and Mathematical Tables, Dover, 1964.



%\bibitem{AblowitzSegur} M.J. Ablowitz and H. Segur, Asymptotic solutions of the Korteweg-de Vries equation,
%  {\em Stud. Appl. Math.} {\bf 57} (1976/77), no. 1, 13--44.
%
%\bibitem{AdlervanMoerbeke} M. Adler and P. van Moerbeke, Completely integrable systems, Euclidean Lie algebras, and curves,
%{\em Adv. in Math.} {\bf 38} (1980), no. 3, 267--317. 
%
%\bibitem{AmirCorwinQuastel} G. Amir, I. Corwin, and J. Quastel, Probability distribution of the free energy of the continuum
%directed random polymer in 1 + 1 dimensions, {\em Comm. Pure Appl. Math.} {\bf 64} (2011), 466--537.
%\bibitem{ABB}
%L.-P. {Arguin}, D.~{Belius}, and P.~{Bourgade}, \newblock {Maximum of the Characteristic Polynomial of Random Unitary
%  Matrices}.
%\newblock {\em Comm. Math. Phys.} {\bf 349} (2017, 703--751.

\bibitem{BBD}
    J. Baik, R. Buckingham, and J. Di Franco, Asymptotics of Tracy-Widom distributions and the total integral of a Painlev\'e
    II function, {\em Comm. Math.
    Phys.} {\bf 280} (2008), 463--497.
    
\bibitem{BaikDeiftRains} J. Baik, P. Deift, and E. Rains, A Fredholm determinant identity and the convergence of moments for random Young tableaux, {\em Comm. Math. Phys.} {\bf 223} (2001), 627--672. 

\bibitem{Bornemann} F. Bornemann, On the numerical evaluation of Fredholm determinants, \textit{Math. Comp.} \textbf{79} (2010), 871--915.

\bibitem{BOO} A. Borodin, A. Okounkov, and G. Olshanski, Asymptotics of Plancherel measures for symmetric groups, {\em J. Amer. Math. Soc.} {\bf 13} (2000), 481--515. 
%
%\bibitem{BothnerBuckingham} T. Bothner and R. Buckingham, Large deformations of the Tracy-Widom distribution I. Non-oscillatory asymptotics, {\em Comm. Math. Phys.} {\bf 359} (2018), 223--263.

\bibitem{BotDeiItsKra1} T. Bothner, P. Deift, A. Its and I. Krasovsky, On the asymptotic behavior of a log gas in the bulk scaling limit in the presence of a varying external potential I, \textit{Comm. Math. Phys.} \textbf{337} (2015), 1397--1463.

\bibitem{BEY} P. Bourgade, L. Erd\H{o}s, and H.-T. Yau, Edge universality of beta ensembles, \textit{Comm. Math. Phys.} \textbf{332} (2014), 261--353.
%
%\bibitem{ChhitaJohanssonYoung}
%S. Chhita, K. Johansson, and B. Young, Asymptotic domino statistics in the Aztec diamond, \emph{Ann. Appl. Prob.} \textbf{25} (2015), no. 3, 1232--1278.
%
%\bibitem{Charlier} C. Charlier, Asymptotics of Hankel determinants with a one-cut regular potential and Fisher-Hartwig singularities, {\em Int. Math. Res. Not.} rny009, https://doi.org/10.1093/imrn/rny009.

%\bibitem{CharlierBessel} C. Charlier, Large gap asymptotics in the piecewise thinned Bessel point process, preprint.

\bibitem{ClaeysDoeraene} T. Claeys and A. Doeraene, The generating function for the Airy point process and a system of coupled Painlev\'e II equations, {\em Stud. Appl. Math.} {\bf 140} (2018), 403--437.

\bibitem{CIK} T. Claeys, A. Its, and I. Krasovsky, Higher order analogues of the Tracy-Widom distribution and the Painlev\'e II hierarchy, {\em Comm. Pure Appl. Math.} {\bf 63} (2010), 362--412.

\bibitem{CloizeauxMehta} J. des Cloizeaux, M.L. Mehta, Asymptotic behavior of spacing distributions for the eigenvalues of random matrices, \textit{J. Math. Phys.} \textbf{14} (1973), 1648--1650.
%
%\bibitem{CorwinGhosal} I. Corwin and P. Ghosal, Lower tail of the KPZ equation, arXiv:1802.03273.
\bibitem{Deift}
P. Deift, Orthogonal polynomials and random matrices: a Riemann-Hilbert approach, \textit{Courant Lecture Notes in Mathematics}, 3. New York University, Courant Institue of Mathematical Sciences, New York; American Mathematical Society, Providence, RI, 1999.

\bibitem{DeiftGioev} P. Deift, D. Gioev, Universality at the edge of the spectrum for
unitary, orthogonal and symplectic ensembles of random matrices,
{\em Comm. Pure Appl. Math.} {\bf 60} (2007), 867--910.

\bibitem{DIKZ1999}
P. Deift, A. Its, A. Kapaev, and X. Zhou, On the algebro-geometric integration of the Schlesinger equations, {\it Comm. Math. Phys.} {\bf 203} (1999), 613--633. 

\bibitem{DIK}
P. Deift, A. Its, and I. Krasovsky, Asymptotics for the Airy-kernel determinant,
{\em Comm. Math. Phys.} {\bf 278} (2008), 643--678.
%
%\bibitem{DeiftItsZhou}
%P. Deift, A. Its, and X. Zhou, A Riemann-Hilbert approach to asymptotic problems arising in the theory of random matrix models, and also in the theory of integrable statistical mechanics, \emph{Ann. Math.} \textbf{278} (1997), 149--235.
%
\bibitem{DKMVZ1}
P. Deift, T. Kriecherbauer, K. McLaughlin, S. Venakides, and X. Zhou, Strong asymptotics of orthogonal polynomials with respect to exponential weights, \emph{Comm. Pure Appl. Math.} {\bf 52} (1999), 1491--1552.

\bibitem{DKMVZ2} P. Deift, T. Kriecherbauer, K.T.-R. McLaughlin, S. Venakides and X. Zhou, Uniform asymptotics for polynomials orthogonal with respect to varying exponential weights and applications to universality questions in random matrix theory, \textit{Comm. Pure Appl. Math.} \textbf{52} (1999), 1335--1425.

\bibitem{DIZ} P. Deift, A. Its and X. Zhou, A Riemann-Hilbert approach to asymptotic problems arising in the theory of random matrix models, and also in the theory of integrable statistical mechanics, \textit{Ann. of Math.} \textbf{146} (1997), 149--235.

\bibitem{DeiftZhou} P. Deift and X. Zhou, A steepest descent method for oscillatory Riemann--Hilbert problems. Asymptotics for the MKdV equation, {\em Ann. of Math.} {\bf 137} (1993), 295--368.

\bibitem{Dyson} F.J. Dyson, Fredholm determinants and inverse scattering problems, \textit{Comm. Math. Phys.} \textbf{47} (1976), 171--183.
%\bibitem{DeiftZhou}
%P. Deift and X. Zhou, A steepest descent method for oscillatory Riemann-Hilbert problems, \emph{Bull. Amer. Math. Soc. (N.S.)} {\bf 26} (1992) 119--124.
%
\bibitem{Ehr sine} T. Ehrhardt, Dyson's constant in the asymptotics of the Fredholm determinant of the sine kernel, \textit{Comm. Math. Phys.} \textbf{262} (2006), 317--341.

\bibitem{FahsKrasovsky} B. Fahs and I. Krasovsky, Splitting of a gap in the bulk of the spectrum of random matrices, \textit{Duke Math. J.} \textbf{168} (2019), 3529--3590.

\bibitem{FK20} B. Fahs and I. Krasovsky, Sine-kernel determinant on two large intervals, preprint, arXiv:2003.08136.

\bibitem{Forrester2} P.J. Forrester, The spectrum edge of random matrix ensembles, {\em Nuclear Phys. B} {\bf 402}
(1993), 709--728.
%\bibitem{Forrester} P.J. Forrester, Asymptotics of spacing distributions 50 years later, {\em MSRI Publications} {\bf 65}, (2014) 199--222.
%\bibitem{FoulquieMartinezSousa} A. Foulquie Moreno, A. Martinez-Finkelshtein, and V. L. Sousa, Asymptotics of orthogonal polynomials for a weight with a jump on [-1,1], {\em Constr. Approx.} {\bf 33} (2011), 219--263.
%\bibitem{GMT} A. Grabsch, S. Majumdar, and C. Texier, Truncated linear statistics associated with the top eigenvalues of random matrices, {\em J. Stat. Phys.} {\bf 167} (2017), no. 2, 234--259.
%
\bibitem{GRtable} I.S. Gradshteyn and I.M. Ryzhik, Table of integrals, series, and products. Seventh edition. Elsevier Academic Press, Amsterdam, 2007.

%\bibitem{Hagg} J. Hagg, Gaussian fluctuations in some determinantal
%processes, PhD thesis, 2007, ISBN 978-91-7178-603-6, http://kth.diva-portal.org/smash/get/diva2:11902/FULLTEXT01.pdf .
%\bibitem{HarnadTracyWidom} J. Harnad, C.A. Tracy, and H. Widom, Hamiltonian structure of equations appearing in random matrices (Cambridge,
%1992), {\em NATO Adv. Sci. Inst. Ser. B Phys.} {\bf 315} (1993), Plenum, New York, 231--245.
%
%\bibitem{HastingsMcLeod} S.P. Hastings and J.B. McLeod,
%A boundary value problem associated with the second
%Painlev\'e transcendent and the Korteweg-de Vries equation,
%\textit{Arch. Rational Mech. Anal.} \textbf{73}  (1980), 31--51.
%
%\bibitem{Hone} A. Hone, Coupled Painlev\'e systems and quartic potentials, {\em J. Phys. A: Math. Gen.} {\bf 34} (2001), 2235--2245.
%
\bibitem{IIKS}
A. Its, A.G. Izergin, V.E. Korepin, and N.A. Slavnov, Differential equations for quantum correlation functions, \emph{Proceedings of the Conference on Yang-Baxter Equations, Conformal Invariance and Integrability in Statistical Mechanics and Field Theory}, Volume \textbf{4}, (1990) 1003--1037.
%
%\bibitem{ItsKrasovsky} A. Its and I. Krasovsky, Hankel determinant and orthogonal polynomials for the Gaussian weight with a jump, {\em Contemporary Mathematics} {\bf 458} (2008), 215--247.
\bibitem{Johansson} K. Johansson, The arctic circle boundary and the Airy process, {\em Ann. Probab.} {\bf 33} (2005), 1--30.
%
%
%\bibitem{Johansson2}K. Johansson,
%Random matrices and determinantal processes, {\em Mathematical statistical physics}, 1--55, Elsevier B. V., Amsterdam, 2006. 

%\bibitem{KrajenbrinkLeDoussalProlhac}
%A. Krajenbrink, P. Le Doussal, and S. Prolhac, Systematic time expansion for the Kardar-Parisi-Zhang equation, linear statistics of the GUE at the edge and trapped fermions, {\em  Nucl. Phys. B} {\bf 936} (2018), 239--305.
\bibitem{KrasovskySine} I. Krasovsky, Gap probability in the spectrum of random matrices and asymptotics of polynomials orthogonal on an arc of the unit circle, \textit{Int. Math. Res. Not.} \textbf{2004} (2004), 1249--1272.

\bibitem{Krasovsky}I. Krasovsky, Large gap asymptotics for random matrices, in ``New Trends in Mathematical Physics'', XVth International Congress on Mathematical Physics, \textit{Springer}, 2009.

\bibitem{KraMarou} I. Krasovsky and T. Maroudas, Airy-kernel determinant on two large intervals, preprint.

%\bibitem{KMVV} A.B.J. Kuijlaars, K. T-R McLaughlin, W. Van Assche, M. Vanlessen, The Riemann-Hilbert approach to strong asymptotics for orthogonal polynomials on $[-1,1]$, \textit{Adv. Math.} \textbf{188} (2004), 337--398.
%
%\bibitem{LavancierMollerRubak}F. Lavancier, J. Moller, and E. Rubak, Determinantal point process models
%and statistical inference:
%Extended version, {\em J. Royal Stat. Soc.: Series B} {\bf 77} (2015), no. 4, 853--877.
%
%\bibitem{Manakov} S. A. Manakov, On the theory of two-dimensional stationary self-focusing of electromagnetic waves, \emph{Sov. Phys. JETP} \textbf{38}
%(1974) 248-253.
%
%\bibitem{Okounkov} A. Okounkov, Random matrices and random permutations, {\em Internat. Math. Res. Notices} {\bf 2000} (2000), no. 20, 1043--1095.
\bibitem{NIST} F.W.J. Olver, D.W. Lozier, R.F. Boisvert, and C.W. Clark, NIST handbook of mathematical functions, (2010) \textit{Cambridge University Press}.
%\bibitem{PraehoferSpohn} M. Praehofer and H. Spohn, Scale invariance of the PNG droplet and the
%Airy process, {\em J. Statist. Phys.} {\bf 108} (2002), 1076--1106.
%
%\bibitem{Romik} D. Romik, The surprising mathematics of longest increasing subsequences, Institute of Mathematical Statistics Textbooks, Cambridge University Press, New York, 2015, xi+353 pp.
%
%\bibitem{Sasano} Y. Sasano, Coupled Painlev\'e II systems in dimension four and the systems of type $A_4^{(1)}$, {\em Tohoku Math. J.} {\bf 58} (2007), 223--245.
%
\bibitem{Soshnikov2}
A. Soshnikov, Universality at the edge of the spectrum in Wigner random matrices, {\em Comm.
Math. Phys.} {\bf 207} (1999), 697--733.

%\bibitem{SoshnikovSineAiryBessel} A. Soshnikov, Gaussian fluctuation for the number of particles in Airy, Bessel, sine, and other determinantal random point fields, \textit{J. Statist. Phys.} \textbf{100} (2000), 491--522.

\bibitem{Soshnikov}
A. Soshnikov, Determinantal random point fields, \emph{Russian Math. Surveys} \textbf{55} (2000), 923--975.



%
\bibitem{TracyWidom}
C. Tracy and H. Widom, Level-spacing distributions and the Airy kernel, \emph{Comm. Math. Phys.} \textbf{159} (1994), 151--174.

\bibitem{Widom1995} H. Widom, Asymptotics for the Fredholm determinant of the sine kernel on a union of intervals, \textit{ Comm. Math. Phys.} \textbf{171} (1995), 159--180.
%
%\bibitem{TracyWidom2}
%C. Tracy and H. Widom, A system of differential equations for the Airy process, {\em Electron. Comm. Prob.} {\bf 8} (2003), 93--98.
%
%\bibitem{TracyWidom3}
%C. Tracy and H. Widom, Differential Equations for Dyson Processes, {\em Commun. Math.
%Phys.} {\bf 252} (2004), 7--41.
%
%
%\bibitem{VMC}C. Verhoeven, M. Musette, and R. Conte,
%General solution for Hamiltonians with extended cubic
%and quartic potentials, {\em Theor. Math. Phys.} {\bf 134} (2003), no. 1, 148--159.
%
%
%\bibitem{WBF} N. Witte, F. Bornemann, and P. Forrester,
%Joint distribution of the first and second eigenvalues at the soft edge of unitary ensembles, {\em Nonlinearity} {\bf 26} (2013), no. 6, 1799--1822. 
%
%\bibitem{XuDai} S.-X. Xu and D. Dai, Tracy-Widom distributions in critical unitary random matrix ensembles and the coupled Painlev\'e II system, {\em Comm. Math. Phys.} (2018), https://doi.org/10.1007/s00220-018-3257-y.
%
%
%\bibitem{XuZhao}
%S.-X. Xu and Y.-Q. Zhao, Painlev\'e XXXIV asymptotics of orthogonal polynomials for the Gaussian weight with a jump at the edge, \emph{Stud. Appl. Math.} {\bf 127} (2011), no. 1, 67--105.


\end{thebibliography}
\end{document}